\documentclass[12pt]{article}
\pdfoutput=1

\usepackage{geometry}

\usepackage[nosort]{cite}

\usepackage[dvipsnames]{xcolor}
\usepackage{cancel}

\usepackage{caption}
\usepackage{epsfig}
\usepackage{amsfonts}
\usepackage{subcaption}
\usepackage{amscd}
\usepackage{latexsym}
\usepackage{dsfont}
\usepackage{amsmath,amssymb, amsthm}
\usepackage{verbatim}
\usepackage{setspace}
\usepackage{color}
\usepackage{fancyhdr}
\usepackage{cite}
\usepackage{hyperref}
\usepackage{cleveref}
\usepackage{tikz}
\usepackage{slashed}
\usepackage{multirow}
\usetikzlibrary{calc}
\usetikzlibrary{topaths}
\usetikzlibrary{decorations}
\usetikzlibrary{decorations.pathmorphing}
\usetikzlibrary{patterns}
\usetikzlibrary{arrows,decorations.markings,cd}
\usetikzlibrary{calc,arrows,cd,decorations.markings,snakes}
\tikzset{
->-/.style args={#1rotate#2}{decoration={markings, mark=at position #1 with {\arrow[scale=1.5,rotate = #2 ]{stealth}}}, postaction={decorate}}
}
\usetikzlibrary{shapes.geometric}
\usetikzlibrary{knots}
\usepackage{tikz-3dplot}

\usepackage{ulem}

 \usepackage{draft}
\usepackage{hyperref}
\usepackage{graphicx}
\usepackage{cite}
\usepackage{mciteplus}
\usepackage{skak}
\usepackage{bbm}
\usepackage[bottom,flushmargin]{footmisc}

\numberwithin{equation}{section}

\def\ie{\begin{equation}\begin{aligned}}
\def\fe{\end{aligned}\end{equation}}

\renewcommand{\H}{\mathcal{H}}
\renewcommand{\A}{\mathcal{A}}
\renewcommand{\B}{\mathcal{B}}
\newcommand{\R}{R}
\newcommand{\Rbulk}{\mathfrak{R}}

\newcommand{\N}{\mathcal{N}}
\renewcommand{\S}{\mathcal{S}}

\renewcommand{\hat}{\widehat}
\renewcommand{\bar}{\overline}
\renewcommand{\tilde}{\widetilde}
\newcommand{\wt}{\widetilde}

\newcommand{\bi}{\begin{itemize}}
\newcommand{\ei}{\end{itemize}}
\newcommand{\bfig}{\begin{figure}\begin{center}}
\newcommand{\efig}{\end{center}\end{figure}}
\newcommand{\ol}{\overline}

\newtheoremstyle{break}
{\topsep}{\topsep}
{}{}%
{\bfseries}{.}
{5pt plus 1pt minus 1pt}{}
\theoremstyle{break}

\newtheorem{thm}{Theorem}
\newtheorem{lem}{Lemma}

\begin{document}
 \begin{titlepage}

\hfill  MIT-CTP/5901\\
 
\title{Disjoint additivity and local quantum physics
}

\author{Daniel Harlow, Shu-Heng Shao, Jonathan Sorce, and Manu Srivastava} 

\textit{Center for Theoretical Physics --- a Leinweber Institute, Massachusetts Institute of Technology}

\vspace{2cm}

\abstract{
Quantum systems of physical interest are often local, but there are at least three competing perspectives on how ``locality'' should be formalized: an algebraic framework, a path-integral framework, and a lattice framework.
One puzzle in this competition is that systems with higher-form symmetries, which are perfectly local from the path-integral and lattice perspectives, can violate the algebraic principle of ``additivity''.
In this paper, we propose a resolution to this puzzle by introducing a weaker locality principle, \textit{disjoint} additivity, which together with Haag duality should always be satisfied in local quantum systems. 
As evidence, we give examples in which disjoint additivity is preserved when ordinary additivity is violated; we show that Haag duality and disjoint additivity are satisfied in rather general lattice systems with local symmetry constraints; we give examples of nonlocal theories in which either disjoint additivity or Haag duality is violated; and finally we give examples of systems with nonlocal symmetry constraints in which disjoint additivity is violated, but can be restored by passing to a local ``SymTFT'' system in one higher dimension.}

\end{titlepage}

\tableofcontents
 
\section{Introduction}

Locality is a core feature of the laws of physics as we currently understand them, but defining it precisely is more difficult than it first appears.  In relativistic quantum field theories one form of locality is \textit{microcausality}, which says that if $O_1$ and $O_2$ are local operators and $x$ and $y$ are spacelike separated then we have
\be\label{mccond}
[O_1(x),O_2(y)]_\pm =0.
\ee
Here $[\cdot ,\cdot ]_\pm$ indicates a commutator unless $O_1$ and $O_2$ are both fermionic, in which case it is an anticommutator.  In non-relativistic systems this condition is weakened to equal-time (anti)commutativity at spatial separation. 

It has long been understood however that microcausality alone is not sufficient to imply all the phenomena one might like to view as consequences of locality.  This is clearly true for non-relativistic systems, since (anti)commutativity at spatial separation puts no constraint whatsoever on the nature of the Hamiltonian.  Moreover even in the context of relativistic theories, where microcausality does constrain the Hamiltonian since it is a statement about the algebra of operators at different times, there are theories which obey microcausality but are not genuinely local.  These include
\bi
\item[(1)] \textbf{Restriction of correlation functions to a Lorentzian hyperplane:} The correlation functions of a $D$-dimensional quantum field theory may be restricted to a $(D-1)$-dimensional timelike hyperplane such as $x=0$, on which they define a set of operator-valued distributions in $D-1$ spacetime dimensions that are consistent with microcausality.  
\item[(2)] \textbf{Generalized free fields:}  A generalized free field theory consists of a set of Gaussian correlation functions that obey microcausality but do not satisfy any local wave equation.  Generalized free fields can be obtained as the boundary limits of quantum fields in AdS space, so this example is similar to the previous one.
\item[(3)] \textbf{Invariant sector under a global symmetry:} Given a quantum field theory with a global symmetry, we can restrict to correlation functions of symmetry-invariant operators.
\item[(4)] \textbf{Virasoro multiplet of the identity:} In a generic two-dimensional conformal field theory, we can consider the ``sub-theory'' containing only the stress tensor correlation functions.
\ei
All four of these examples obey the classical Wightman and Haag/Kastler axioms for relativistic quantum field theory \cite{streater2000pct,Haag:axioms}, but we feel they are too perverse to be classified as local quantum theories. Theories (1) and (2) are easy to criticize: neither one has a conserved energy-momentum tensor, which a reasonable local field theory should surely have.  It is more subtle to pinpoint what is wrong with theories (3) and (4).

Perhaps the most common critique of theories (3) and (4) is that their partition functions cannot be defined on general spacetime manifolds in a way that is compatible with diffeomorphism invariance.  More concretely, in a local quantum field theory the partition function on any spacetime manifold $M$ evaluated as a functional of its background fields $\chi$ should obey
\be\label{diffZ}
Z[f_*\chi]=Z[\chi],
\ee
where $f:M\to M$ is any diffeomorphism and $f_*\chi$ is the action of $f$ on the background fields $\chi$.  A classic example of this requirement arises in the context of 2D conformal field theories, where we demand that the thermal partition function on a spatial circle gives a torus partition function which is invariant under the modular  S  transformation $\tau'=-1/\tau$ \cite{Cardy:modular, Moore:1988uz,Moore:1988qv}.\footnote{In fermionic theories this statement is only true if we choose anti-periodic boundary conditions on the spatial circle, as otherwise the spin structure (which we view as part of $\chi$) is not invariant under this transformation. 
}  There are various ways to justify \eqref{diffZ}, with perhaps the most intuitive being that it is true for any quantum field theory which is constructed from a local scalar Lagrangian density and a path integral measure that respects diffeomorphism invariance.  There is also a more formal categorical description of the requirement \eqref{diffZ} based on cutting and gluing path integrals \cite{Segal:CFT,Atiyah:TQFT, Lurie:extended}.

In principle the condition \eqref{diffZ} may be a sufficient definition of locality in field theory, but it does have some undesirable features.  For one thing it is difficult to check in practice, as we are required to formulate the theory on all possible spacetime manifolds $M$.  Moreover it does not work as stated in theories that have diffeomorphism anomalies such as a pair\footnote{We use a pair because the theory with a single right-mover is not local according to our criteria, as we will see in section \ref{sec:fermions}.} of right-moving real fermions in $1+1$ dimensions or the chiral $(\mathfrak{e}_8)_1$ WZW model (which features in the heterotic string theory), and it seems excessive to exclude such theories by fiat. 
A somewhat more nebulous concern is that \eqref{diffZ} has the feeling of being a consequence of locality, rather than its essential formulation. 

A second starting point for excluding theories such as (3) and (4) is the lattice: we restrict to theories that can be obtained by starting with a tensor product Hilbert space at short distances and a Hamiltonian with interactions that only mix sites at $O(1)$ lattice separation.  This is the predominant approach to locality in condensed matter theory, and it is arguably a more intuitive definition of locality than one based on \eqref{diffZ}.  However, this approach also has its problems --- in particular the lattice theory has non-universal structure at short distances which is removed in the continuum limit.  There can also be emergent features of the continuum limit (such as Lorentz invariance) which are not present at any finite lattice spacing.  Moreover it is again challenging to discuss theories with diffeomorphism anomalies from this point of view, and there could also be intrinsically strongly-coupled quantum field theories that do not admit any controllable lattice formulation.

A third approach to formalizing locality in quantum systems is the algebraic approach, which assigns an operator algebra to each spatial subregion and then demands that these algebras obey compatibility conditions \cite{Haag:axioms,Haag:book}.  We will review these conditions in more detail in section \ref{sec:review}, but two that are often discussed are ``Haag duality,'' which says that the commutant of the algebra $\A(\R)$ associated to a spatial region $\R$ is the algebra $\A(\R')$ associated to its spatial complement, and ``additivity,'' which says that the algebra $\A(\R_1\cup\R_2)$ associated to the union of two spatial regions $\R_1$ and $\R_2$ in the same time slice is the algebra $\A(\R_1) \vee \A(\R_2)$ generated by $\A(\R_1)$ and $\A(\R_2)$.  This approach to locality has several advantages in comparison to the previous two: it is formulated directly in the continuum, it allows for theories with diffeomorphism anomalies, and it has constraining power directly in Minkowski space without need to consider other manifolds.  One recent illustration of the power of this approach is that in \cite{Benedetti:2024dku} it was shown that in 2D conformal field theories additivity together with Haag duality implies invariance of the thermal partition function under the modular S transformation (see also \cite{Kawahigashi:modular, Rehren:modular} for earlier related work). Relatedly it has been understood for some time that theory (3) above, the invariant sector of a quantum field theory under a global symmetry, violates Haag duality and/or additivity \cite{Casini:2019kex,Casini:2020rgj,Shao:2025mfj}. 
These algebraic axioms thus have substantial constraining power for excluding non-local theories such as (1-4).

An important problem with using Haag duality and additivity to diagnose locality was explained in a series of papers by Casini, Mag\'{a}n, and collaborators \cite{Casini:2020rgj,Casini:2021tax,Casini:completeness}.  The problem is that there are quantum field theories that should clearly count as local, for example free Maxwell theory, in which we cannot assign algebras to regions in a way that obeys both Haag duality and additivity.  This happens in theories which have higher-form global symmetries \cite{Gaiotto:2014kfa}.  This is because theories with such symmetries have ``unbreakable'' extended operators, 
which cannot be generated out of local operators restricted to the region where the surface is defined.  For other examples of physically relevant theories with higher-form symmetries see \cite{Gaiotto:2014kfa,McGreevy:2022oyu, Cordova:2022ruw, Schafer-Nameki:2023jdn, Brennan:2023mmt, Bhardwaj:2023kri, Shao:2023gho}.  We thus need a better algebraic criterion for locality to allow for such theories.
  
The central proposal of this paper is that any local quantum theory must obey Haag duality together with a weaker form of additivity that we term \textit{disjoint additivity}.  Disjoint additivity is like additivity, except that we only require $\A[\R_1\cup \R_2]=\A[\R_1]\vee \A[\R_2]$ when $R_1$ and $R_2$ are disjoint in an appropriate sense. (For spatial regions in a continuum non-relativistic theory the appropriate sense is $\bar{\R_1}\cap \bar{\R_2}=\varnothing$.)  In the remainder of this paper we will argue that this rule is obeyed by  ``good'' theories such as free Maxwell theory which have higher-form symmetries, but not by ``bad'' theories such as examples (1-4) above. 

In section \ref{sec:review}, we define additivity and Haag duality and review the tension between them in systems with higher-form symmetries. We then define \textit{disjoint} additivity as a relaxation of the notion of additivity, and give additivity-violating examples in which disjoint additivity is preserved.
In section \ref{sec:lattice}, we prove that Haag duality is satisfied in any lattice system with a compact Lie group constraint; this includes lattice gauge theory and the ground space of a generic stabilizer Hamiltonian.  We also show that disjoint additivity holds provided that the action of the symmetry group is sufficiently locally generated (as it is for lattice gauge theory and local stabilizer codes).  
In section \ref{sec:violation-examples}, we show that Haag duality and/or disjoint additivity are violated in the nonlocal theories (1-4) above.  In section \ref{sec:fermions} we explain how to generalize our algebraic definitions to fermionic systems, and use them to show that a Majorana chain with an odd number of fermions is not local according to our rules.   Finally, in section \ref{sec:SymTFT}, we revisit the global symmetry-invariant sectors of local systems, in which both additivity and disjoint additivity are typically violated, and show that in some cases disjoint additivity can be restored if we reinterpret the theory as a boundary condition for a topological ``SymTFT'' theory living in one higher dimension.  

\section{Algebras and higher-form symmetries}
\label{sec:review}

\subsection{Algebraic notions}
\label{sec:algebraic-notions}
A general quantum mechanical system comes equipped with a Hilbert space and a Hamiltonian.  To turn it into a local quantum system we must give it some additional structure: a collection of regions in which degrees of freedom can be localized, together with a rule assigning an algebra of operators to each region \cite{Haag:axioms,Haag:book}.\footnote{A potential source of confusion is that the same quantum mechanical system can be equipped with inequivalent local algebras, just as the same abstract set can be equipped with inequivalent topologies.  A simple example of this is that Maxwell theory on $\mathbb{R}^4$ can have gauge group $U(1)$ or gauge group $\mathbb{R}$; the Hilbert space and Hamiltonian are the same in either case, but the assignment of operators to regions is different since the theories have different Wilson and 't Hooft lines; they are thus distinct as quantum field theories \cite{Aharony:2013hda} (see also section 3.4 of \cite{Harlow:2018tng}).}  We will take the term ``region'' to mean slightly different things in relativistic and non-relativistic systems.  
In a non-relativistic lattice system where the global Hilbert space is a tensor product of local Hilbert spaces, we define a region to be a subset of these tensor factors.  The associated (bosonic) algebra is the collection of operators supported on those factors.  In the non-relativistic continuum these regions become open subsets of a constant time slice; for technical reasons we further require such subsets to be regular (an open set is regular if it is equal to the interior of its closure). 

In continuum relativistic quantum field theory, however, there is ambiguity in the literature for how regions should be defined.  The definition we will adopt is that in a Lorentzian spacetime $M$, a region $R$ is a subset of $M$ which obeys
\be
R''=R.
\ee
Here $R'$ is the \textit{spatial complement} of $R$, which is defined as the interior of the causal complement $R^c=M-J(R)$, with $J(R)$ being the set of points in $M$ which can be connected to $R$ by a causal curve.  Note that with this definition $R$ is automatically open, and it is regular as well since it is the interior of a closed set.\footnote{Regions of this type were called ``wedges'' in \cite{Bousso:2022hlz}, but we will not use this term.}  The motivations for this definition are (1) it is covariant, (2) it makes sense in any spacetime without causal restrictions such as time-orientability or global hyperbolicity, and (3) compared to the various alternatives it most closely parallels the mathematical structure of von Neumann algebras.\footnote{An alternative that works when $M$ is globally hyperbolic, and which more closely resembles the non-relativistic definition, is to require $R$ to be the domain of dependence of a regular open subset $T$ of an acausal Cauchy surface $\Sigma$ for $M$. The complement $R'$ is then defined as $D[\mathrm{Int}(\Sigma-T)]$.  Our definition does not require $M$ to be globally hyperbolic, and even in the globally hyperbolic setting our definition includes all regions defined this way but also allows for more regions that seem natural to include (see appendix \ref{app:causality}).}  We explain this motivation in more detail in appendix \ref{app:causality}.  In non-relativistic lattice theories the spatial complement $R'$ simply indicates the complementary set of lattice factors, so $R=R''$ is automatic, and in the non-relativistic continuum we will take $R'$ to be the interior of the complement of $R$ in its time slice.  In figure \ref{opregionsfig} we illustrate these definitions for the non-relativistic and relativistic cases.  

\bfig
\includegraphics[height=5cm]{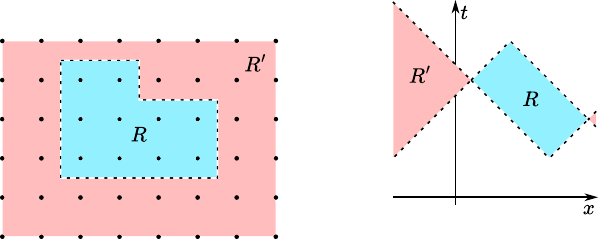}
\caption{Defining a region $\R$ and its spatial complement $R'$ in non-relativistic and relativistic systems.  On the left we have a Hamiltonian lattice at fixed time, while on the right we have a spacetime diagram where light moves on 45-degree lines.  The dashed boundaries are not included in either $R$ or $R'$.}\label{opregionsfig}
\efig
In both the nonrelativistic and relativistic cases, we refer to the algebra associated to a region $\R$ as $\A(\R).$  Roughly speaking, $\A(\R)$ is the set of operators generated by fields in $\R$, although it can also include operators such as twist operators and 't Hooft loops which are created by removing pieces of $R$ and imposing boundary conditions on the boundary which this creates.  In our convention, local algebras are always represented by bounded operators on the Hilbert space of the theory, and moreover we require each algebra to be a \textit{von Neumann algebra}, meaning that it is closed with respect to a natural way of taking limits.\footnote{For a physics-friendly introduction to von Neumann algebras, see the reviews \cite{Witten:notes, Sorce:notes}.} Given a set $\mathcal{X}$ of bounded operators acting on the Hilbert space $\H$, the commutant of $\mathcal{X}$ is the set
\begin{equation}
    \mathcal{X}' \equiv \{x' \in \B(\H)\,|\, [x', x] = 0 \text{ for all } x \in \mathcal{X}\},
\end{equation}
with $\B(\H)$ the set of all bounded operators acting on $\H$.  
A useful theorem proved by von Neumann states that an adjoint-closed set $\A$ acting on $\H$ is a von Neumann algebra if and only if it is equal to its own double commutant, $\A = \A''.$
In particular, this means that for any adjoint-closed set of operators $\mathcal{X},$ the set $\mathcal{X}''$ is the smallest von Neumann algebra containing $\mathcal{X}.$  Moreover it is immediate that $\mathcal{X}'=\mathcal{X}'''$ for $\mathcal{X}$ any set of bounded operators, so if $\mathcal{X}$ is adjoint-closed then $\mathcal{X}'$ is always a von Neumann algebra.
The smallest von Neumann algebra containing two distinct von Neumann algebras $\A_1$ and $\A_2$ is written $\A_1 \vee \A_2$, equivalently
\begin{equation}
    \A_1 \vee \A_2 \equiv (\A_1 \cup \A_2)''.
\end{equation}

For simplicity of exposition, in this section and in most of the rest of the paper, we will restrict to bosonic theories where operators in $R$ and $R'$ commute; we will discuss fermionic theories in section \ref{sec:fermions}. We thus have the microcausality relation
\begin{equation} \label{eq:microcausality}
    \A(\R') \subseteq \A(\R)'.
\end{equation}
If this inclusion is an equality for all regions $\R$, then we say the theory satisfies \textit{Haag duality}.
While condition \eqref{eq:microcausality} is always satisfied, there are many interesting situations in which Haag duality can be violated depending on how one chooses to assign algebras to regions; see for example \cite{Haag:book, Benedetti:2024dku, Casini:2020rgj, Casini:2013rba, Shao:paper1}.
As we will explain in section \ref{sec:lattice}, however, our perspective is that in local quantum systems the local algebras $\A(\R)$ always satisfy Haag duality.  A simple consequence of Haag duality is that the following statements are equivalent:
\begin{itemize}
    \item The algebra associated to the empty set contains only scalar multiples of the identity, 
            \be\label{Aempty}
            \A(\varnothing)=\lambda I.
            \ee
    \item 
    The algebra associated to all of space(time) is the full set of bounded operators on the Hilbert space. 
\end{itemize}
Both of these statements are quite natural, so we will assume they are both true.  

Another natural property of local algebras is nesting (sometimes called isotony):
\begin{equation}
    (\R_1 \subseteq \R_2) \Rightarrow (\A(\R_1) \subseteq \A(\R_2)).
\end{equation}
In particular this implies
\begin{equation}\label{addeq}
    \A(\R_1) \vee \A(\R_2) \subseteq \A\left((\R_1 \cup \R_2)''\right),
\end{equation}
where on the right-hand side the double complement is necessary to get a valid region in the relativistic case (in the non-relativistic case it does nothing).  In non-relativistic theories, if this inclusion is an equality for all regions $R_1$ and $R_2$  then the theory is said to satisfy \textit{additivity}.  Additivity is trickier to formulate in relativistic theories on general Lorentzian spacetimes.  However, if we take $M$ to be globally hyperbolic, then a fairly simple definition is possible: additivity is satisfied if \eqref{addeq} is saturated for all regions $\R_1,\R_2$ with the property that as submanifolds of $M$ they admit acausal Cauchy surfaces $\Sigma_1,\Sigma_2$ which lie in a single Cauchy surface $\Sigma$ for $M$.\footnote{The difficulty in formulating additivity for general spacetimes is perhaps another indication that we should not afford it the same level of respectability as Haag duality.  There is no such problem with the ``disjoint additivity'' that we introduce below.}  

Additivity and Haag duality are locality principles --- they are formal ways of stating that the local degrees of freedom in a quantum theory satisfy a form of completeness.  Haag duality captures the idea that we should take $\A(\R)$ to be the ``maximal'' set of operators which microcausality ``allows'' to be in $\R$, while additivity captures the idea that in a local quantum theory all operators in $\A(\R)$ should be built out of local operators in $\R$.  We can also motivate these principles from the lattice: in a tensor product Hilbert space they hold essentially by definition.  On the other hand we will see in a moment that there is some tension between these two principles in general.  The perspective we will advocate for in this paper is that Haag duality is sacred while additivity may be violated.

\subsection{Violation of additivity due to higher-form global symmetries}
\label{sec:symmetry-violation}

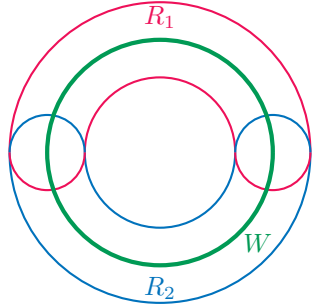
\begin{figure}
\centering
\begin{tikzpicture}

\draw[ thick, RoyalBlue] (-2,0) arc[start angle=180, end angle=360, radius=2];
\draw[ thick, OrangeRed] (2,0) arc[start angle=0, end angle=180, radius=2];

\draw[ thick, RoyalBlue] (-1,0) arc[start angle=180, end angle=360, radius=1];
\draw[ thick, OrangeRed] (1,0) arc[start angle=0, end angle=180, radius=1];

\draw[ thick, RoyalBlue] (-1,0) arc[start angle=0, end angle=180, radius=.5];
\draw[ thick, OrangeRed] (-2,0) arc[start angle=180, end angle=360, radius=.5];

\draw[ thick, RoyalBlue] (2,0) arc[start angle=0, end angle=180, radius=.5];
\draw[ thick, OrangeRed] (1,0) arc[start angle=180, end angle=360, radius=.5];

\draw[ultra thick, ForestGreen] (1.5,0) arc[start angle=0, end angle=360, radius=1.5];

\node at (1.3, -1.2) {\footnotesize \textbf{\color{ForestGreen}$W$}};
\node at (0, 1.8) {\footnotesize \textbf{\color{OrangeRed}$\R_1$}};
\node at (0, -1.8) {\footnotesize \textbf{\color{RoyalBlue}$\R_2$}};

\end{tikzpicture}

\caption{A non-contractible region $\R=\R_1\cup \R_2$ that encloses a line operator $W$ (green). Here $\R_1$ (red) and $\R_2$ (blue) are overlapping, contractible regions. If $W$ carries a nontrivial topological 1-form global symmetry charge, it cannot end on local operators. In this case, we have $W\in {\cal A}(\R_1\cup \R_2)$ but $W \notin {\cal A}(\R_1) \vee {\cal A}(\R_2)$.  
This violates additivity.  Here we have shown spatial regions; to get relativistic regions, we can take the domain of dependence of each spatial region.}\label{fig:tube}
\end{figure}

Although Haag duality and additivity are natural requirements for a local theory to obey, there exist physical systems in which they cannot simultaneously be satisfied \cite{Casini:2020rgj, Casini:completeness} (see also section III.4.2 of \cite{Haag:book}).   
In fact this happens already in free Maxwell $U(1)$ gauge theory without matter fields.  This is because the Wilson loop operator $W = \exp(i q\oint A)$ can only be defined on a closed loop in spacetime; it cannot be defined on a curve with endpoints, as there are no charged matter fields on which the Wilson line could end. 
In the presence of such an ``unbreakable" line operator $W$, additivity and Haag duality become incompatible for non-contractible regions where $W$ wraps a nontrivial cycle.  If $W$ is not included in the algebra for such a region, then Haag duality is violated since $W$ commutes with everything in the complementary region.  If $W$ is included in the algebra, then Haag duality is preserved, but additivity is violated as in figure \ref{fig:tube}.\footnote{To preserve Haag duality, it is important that we only allow Wilson loop charges $q$ which respect the periodicity of the gauge group. Otherwise $W$ would not commute with 't Hooft loops (or more generally surfaces) with which it links.  See e.g. section 3.4 of \cite{Harlow:2018tng} for further explanation.}
This is because a non-contractible region can be built up topologically as a union of contractible regions, but the unbreakable operator $W$ cannot be constructed from operators in these regions.  In this paper we view Haag duality as mandatory, so we from now on describe this situation by saying that the local operator algebras violate additivity.

The same additivity violation occurs in any QFT with a higher-form global symmetry. 
For simplicity, consider a continuum quantum field theory in $D$ spacetime dimensions with a   one-form  global symmetry.  
This means that there are one-dimensional charged objects $W$ that are acted on by a $(D-2)$-dimensional topological operator $\mathcal{L}$ by linking.
For example, in free Maxwell theory, the $U(1)$ one-form symmetry operator is the Gauss law operator ${\cal L} = \exp({\frac{i \theta}{e^2}}\oint_\Sigma\star F)$, where $F$ is the two-form field strength, $\Sigma$ is a codimension-two manifold in spacetime, and $\theta$ takes values in $[0,2\pi)$.\footnote{There is also a magnetic $U(1)$ one-form global symmetry which leads to a similar violation of additivity from the unbreakable 't Hooft line.}  The symmetry operator $\mathcal{L}$ may be passed through the charged object $W$ at the cost of introducing an eigenvalue, which is $e^{iq\theta}$ in the case of the Maxwell theory  \cite{Gaiotto:2014kfa}.
Since the linking rule is  purely local, the action of $\mathcal{L}$ on $W$ is independent of whether or not $W$ ``ends'' on any point  operators. 

The operator  $\cal L$ generating a one-form global  symmetry is topological in the sense that it is invariant under small deformations of its support in spacetime. 
In the Maxwell theory example, this is a consequence of the  equation of motion $d\star F=0$. 
It follows that a line operator $W$ carrying nontrivial one-form symmetry charge \textit{cannot} end on point operators, because then $\mathcal{L}$ could either pass through $W$ producing a nontrivial eigenvalue, or be ``pulled out'' through the ends of $W$ without producing any eigenvalue; this would give the inconsistency sketched in figure \ref{fig:endable} \cite{Rudelius:2020orz,Heidenreich:2021xpr}. 
To conclude, a one-form global symmetry implies the presence of an unbreakable charged line operator, which in turn violates additivity \cite{Casini:2020rgj, Casini:completeness,Haag:book}. 

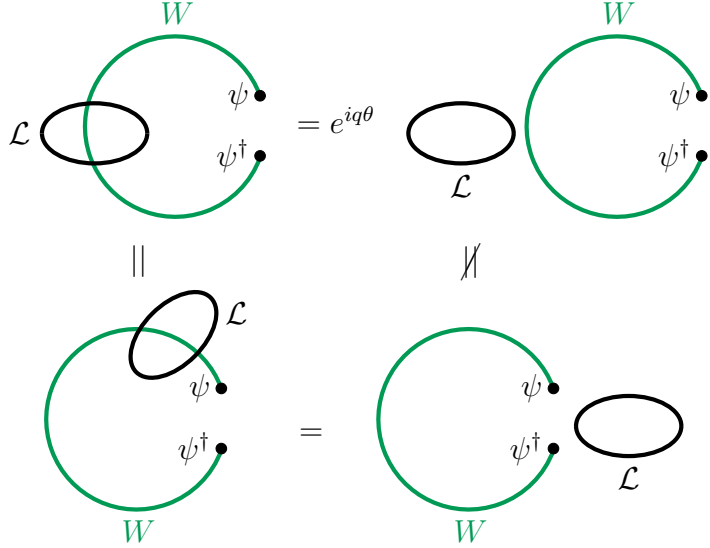
\begin{figure}
\centering

\begin{tikzpicture}
    \draw[ultra thick,ForestGreen] (0,0) arc[start angle=20, end angle=340, radius=1.2];

    \draw[ ultra thick] (-1.5,-0.5) arc[start angle=0, end angle=180, x radius=.7, y radius=0.4];
 
    \draw[ ultra thick] (-2.9,-0.5) arc[start angle=180, end angle=360, x radius=.7, y radius=0.4];

    \draw node at (-3.2,-0.5) {$\cal L$};

    \draw node at (-1.1,1.1) {\color{ForestGreen}$W$};

    \filldraw[black] (0,0) circle (2pt) node[anchor=east]{$\psi$};
    \filldraw[black] (0,-.8) circle (2pt) node[anchor=east]{$\psi^\dagger$};
\end{tikzpicture}
\raisebox{1.3cm}{~~$=e^{iq\theta}$~~}
\begin{tikzpicture}
    \draw[ForestGreen,ultra thick] (0,0) arc[start angle=20, end angle=340, radius=1.2];

    \draw[ ultra thick] (-2.5,-0.5) arc[start angle=0, end angle=360, x radius=.7, y radius=0.4];

    \draw node at (-3.2,-1.2) {$\cal L$};

    \draw node at (-1.1,1.1) {\color{ForestGreen}$W$};

    \filldraw[black] (0,0) circle (2pt) node[anchor=east]{$\psi$};
    \filldraw[black] (0,-.8) circle (2pt) node[anchor=east]{$\psi^\dagger$};
\end{tikzpicture}
~\\
\begin{tikzpicture}
    \draw[ultra thick, ForestGreen] (0,0) arc[start angle=20, end angle=340, radius=1.2];

    \begin{scope}
        \draw[ ultra thick, rotate around={45:(0,0.5)}] (-.3,1.1) ellipse (.7 and 0.4);
    \end{scope}

     \draw node at (0.2,1) {$\cal L$};

    \draw node at (-1.1,1.7) {$||$};
    \draw node at (-1.1,-1.9) {\color{ForestGreen}$W$};

    \filldraw[black] (0,0) circle (2pt) node[anchor=east]{$\psi$};
    \filldraw[black] (0,-.8) circle (2pt) node[anchor=east]{$\psi^\dagger$};
\end{tikzpicture}
\raisebox{1.5cm}{~~~$=$~~~}
\begin{tikzpicture}
    \draw[ultra thick,ForestGreen] (0,0) arc[start angle=20, end angle=340, radius=1.2];

     \draw node at (-1.1,1.7) {$\cancel{||}$};

    \draw[ultra thick] (1.7,-0.5) arc[start angle=0, end angle=180, x radius=.7, y radius=0.4];
 
    \draw[ultra thick] (.3,-0.5) arc[start angle=180, end angle=360, x radius=.7, y radius=0.4];

    \draw node at (1,-1.2) {$\cal L$};

    \draw node at (-1.1,-1.9) {\color{ForestGreen}$W$};

    \filldraw[black] (0,0) circle (2pt) node[anchor=east]{$\psi$};
    \filldraw[black] (0,-.8) circle (2pt) node[anchor=east]{$\psi^\dagger$};
\end{tikzpicture}
\caption{A line operator carrying a topological one-form  symmetry charge must be unbreakable. Here we demonstrate this in three spacetime dimensions for a $U(1)$ one-form symmetry generated by a topological line operator  $\cal L$. Let $W$  be a line operator that  carries charge $q$ under ${\cal L}$. This means that if we locally unlink $W$ and ${\cal L}$, this process gives a phase $e^{iq\theta}$. If $W$ can end on a pair of operators $\psi,\psi^\dagger$, then we reach a contradiction by deforming the topological line $\cal L$ in two different ways.}\label{fig:endable}
\end{figure}

\subsection{Disjoint additivity: definitions and examples}
\label{sec:definitions}

The violation of additivity from the previous subsection rested crucially on expressing a non-contractible region as a union of \textit{overlapping} contractible regions.
If we had not allowed $\R_1$ and $\R_2$ to overlap, then additivity would not have been violated.
We thus are motivated to introduce a notion of \textit{disjoint additivity}, which we first naively define as requiring $\A(\R_1 \cup \R_2) = \A(\R_1) \vee \A(\R_2)$ when $R_1$ and $R_2$ are in the same time slice and $\R_1 \cap \R_2 = \varnothing.$ 
This serves as a preliminary definition, but we will soon see that we need to be a bit more careful about how close the regions $R_1$ and $R_2$ can be; we will also introduce a covariant version for relativistic theories.

Clearly the violation of additivity shown in figure \ref{fig:tube} does not constitute a violation of our preliminary notion of  disjoint additivity.
However, this is not enough to convince us that disjoint additivity is satisfied in local quantum systems.
We will provide more systematic discussions on this point in sections \ref{sec:lattice} and \ref{sec:SymTFT}.
For the moment, let us consider two concrete lattice examples in which we can show explicitly that (a refined version of) disjoint additivity is satisfied.

Our first example, which will lead us to the more refined definition of disjoint additivity, will be $\mathbb{Z}_2$ gauge theory on a periodic chain   with a qubit on every edge $e$. 
Starting with a tensor product Hilbert space, we restrict to a two-dimensional subspace by strictly enforcing the local constraints
\ie\label{Z2constraint}
X_e X_{e+1} =1 \,,~~~\forall ~e\,.
\fe
This is a topological theory with only two states: (i) $|+\,\rangle$, where all the $X_e$ operators have eigenvalue $1$; and (ii) $|-\,\rangle$ where all the $X_e$ operaotrs have eigenvalue $(-1)$.  As shown in figure \ref{fig:Z2-chain}, the theory has two nontrivial operators: a line operator $Z=\prod_{e} Z_e$ acting as
\begin{equation}
    Z |+\,\rangle = |-\, \rangle, \quad Z |-\,\rangle = |+\,\rangle,
\end{equation}
and a local operator $X$ acting as
\begin{equation}
    X |+\,\rangle = |+\,\rangle, \quad X |-\,\rangle = (-1) |-\,\rangle.
\end{equation}
The operator $X$ can be chosen to be $X_e$ for any edge $e$. 
Using the constraint \eqref{Z2constraint}, it can be equated to the Pauli-$X$ operator at any other edge. Therefore, $X$ is a topological local operator that  belongs to every nontrivial region, while the operator $Z$ only belongs to the algebra of the full circle $S^1.$\footnote{The continuum field theory description of the 1+1D $\mathbb{Z}_2$ gauge theory is given by ${\cal L} = \frac{2i}{2\pi} ad\phi$, where $a$ is a $U(1)$ one-form gauge field and $\phi\sim \phi+2\pi$ is a compact scalar field.  The local operator $X$ is represented by $e^{i \phi}$, which generates a $\mathbb{Z}_2$ one-form global symmetry.  The line operator $Z$ is represented as a Wilson line $e^{i \oint_{S^1}a}$, which is unbreakable and charged under the said one-form symmetry. Additivity is therefore violated via the argument in Section \ref{sec:symmetry-violation}. \label{fn:BF}}
The topological local operator $X$ generates a $\mathbb{Z}_2$ one-form global symmetry. 
The algebra $XZ = -ZX$ implies that the line operator $Z$ is charged under this one-form symmetry, and is thus unbreakable. 
Additivity is violated because the circle $S^1$ can be written as a union of overlapping regions that contain only $X$ operators, and the $Z$ operator can never be produced from these.

In fact, on the lattice, due to the underlying discreteness of the system, it is possible to build the full circle using two disjoint sets of edges.
This appears to be a violation of disjoint additivity as we formulated it above.  On the other hand, the violation only happens because the regions are directly adjacent at the lattice scale, and it is not so clear that we should view such regions as really being disjoint in the continuum limit.  We therefore propose that in lattice systems based on an underlying tensor product Hilbert space, perhaps with some gauge constraints imposed, the system satisfies \textit{disjoint additivity} if there is a definition of adjacency such that we have
    \begin{equation}\label{def:precise-lattice}
        \A(\R_1 \cup \R_2) = \A(\R_1) \vee \A(\R_2)
    \end{equation}
    whenever we have $\R_1 \cap \R_2 = \varnothing$ and $\R_1, \R_2$ non-adjacent.

\begin{figure}
\centering
\begin{tikzpicture}[scale=0.7]  
    \draw[ultra thick] (45:2) arc[start angle=45, end angle=45+360, radius=2];

    \foreach \angle in {0, 45, 90, 135, 180, 225, 270, 315} {
      
        \pgfmathsetmacro\x{2*cos(\angle)}
        \pgfmathsetmacro\y{2*sin(\angle)}
        
        \pgfmathsetmacro\dx{0.2*cos(\angle)}
        \pgfmathsetmacro\dy{0.2*sin(\angle)}
        \draw[ultra thick] (\x-\dx, \y-\dy) -- (\x+\dx, \y+\dy);
    }

    \pgfmathsetmacro\x{1.5*cos(22.5)}
    \pgfmathsetmacro\y{1.5*sin(22.5)}
    \node at (\x, \y) {\footnotesize \textbf{$X$}};

    \begin{scope}[xshift=8cm]
        \draw[ultra thick] (45:2) arc[start angle=45, end angle=45+360, radius=2];

    \foreach \angle in {0, 45, 90, 135, 180, 225, 270, 315} {
      
        \pgfmathsetmacro\x{2*cos(\angle)}
        \pgfmathsetmacro\y{2*sin(\angle)}
        
        \pgfmathsetmacro\dx{0.2*cos(\angle)}
        \pgfmathsetmacro\dy{0.2*sin(\angle)}
        \draw[ultra thick] (\x-\dx, \y-\dy) -- (\x+\dx, \y+\dy);

        \pgfmathsetmacro\X{2.5*cos(\angle+22.5)}
        \pgfmathsetmacro\Y{2.5*sin(\angle+22.5)}
        \node at (\X, \Y) {\footnotesize \textbf{$Z$}};
    }

    \end{scope}

\end{tikzpicture}
\caption{In one spatial dimension, $\mathbb{Z}_2$ gauge theory contains an ``$X$'' operator that acts as the Pauli $X$ on local edges, and a ``$Z$'' operator that acts as a Pauli-$Z$ chain on all edges simultaneously.}\label{fig:Z2-chain}
\end{figure}
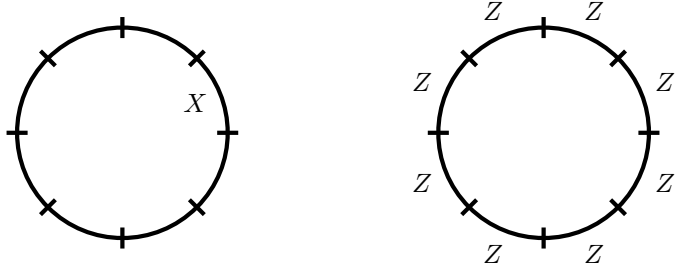

The precise notion of ``non-adjacent" depends on the details of the lattice model. 
In section \ref{sec:lattice}, we will give a precise definition for a large class of lattice models, including general lattice gauge theories and stabilizer Hamiltonians. 
For now, we note that the 1+1D $\mathbb{Z}_2$ gauge theory satisfies disjoint additivity with adjacency defined by the ordinary lattice neighbor relation.
This is because for any two regions with a buffer of at least one site between them, the combined algebra and the individual algebras are both generated by the single operator $X$.

A more involved example is provided by lattice $\mathbb{Z}_2$ gauge theory in $2+1$ dimensions quantized on a spatial torus, a four-state topological system that is sometimes called the ground space of the toric code \cite{Kitaev:1997wr}. 
The space is a square lattice on a torus, with a qubit on every edge $e$. The ground space is found by requiring the operators
\begin{align}
&U_v = \prod_{e\ni v}X_e \,, \label{TCv}\\
&U_f =\prod_{e\in f}Z_e\,,\label{TCf}
\end{align}
to be $+1$ for every vertex $v$ and face $f$. 
There are no local degrees of freedom, but there are nontrivial topological lines of ``$X$-type'' and ``$Z$-type,'' shown in figure \ref{fig:toric-lines}, generating two $\mathbb{Z}_2$ one-form global symmetries. 
A topological $Z$-line that wraps the torus horizontally is distinct, as an operator, from one that wraps the torus vertically, and the four ground states are specified by choosing eigenvalues $+1$ or $-1$ for each of these two operators.
Acting with a horizontal $X$-line flips the eigenvalue of the vertical $Z$-line, and acting with a vertical $X$-line flips the eigenvalue of the horizontal $Z$-line. 
In this sense, the $X$-line is charged under the one-form symmetry generated by the $Z$-line, and vice versa. It follows that neither of the $X$ and $Z$-lines are breakable in the ground space of toric code.

\begin{figure}
\centering
\begin{tikzpicture}
    \foreach \y in {0, 0.75, 1.5, 2.25, 3} {
        \draw [thin, dashed] (0, \y) to (3, \y);
    }
    \foreach \x in {0, 0.75, 1.5, 2.25, 3} {
        \draw [thin, dashed] (\x, 0) to (\x, 3);
    }

    \draw [ultra thick, BrickRed] (0, 1.5) -- ++ (0, 0.75);
    
    \draw [ultra thick, BrickRed] (0.75, 1.5) -- ++ (0, 0.75) -- ++(0.75, 0) -- ++(0, 0.75);

    \draw [ultra thick, BrickRed] (1.5, 2.25) -- ++ (0.75, 0) -- ++(0, -0.75);

    \draw [ultra thick, BrickRed] (3, 1.5) -- ++ (0, 0.75);

    \node at (4, 1.5) {$=$};

    \node at (-0.5, 1.8) {\color{BrickRed} $X$};

    \begin{scope}[xshift=5cm]
    \foreach \y in {0, 0.75, 1.5, 2.25, 3} {
        \draw [thin, dashed] (0, \y) to (3, \y);
    }
    \foreach \x in {0, 0.75, 1.5, 2.25, 3} {
        \draw [thin, dashed] (\x, 0) to (\x, 3);
    }
    \draw [ultra thick, BrickRed] (0, 1.5) -- ++ (0, 0.75);

    \draw [ultra thick, BrickRed] (0.75, 1.5) -- ++ (0, 0.75);
    
    \draw [ultra thick, BrickRed] (1.5, 1.5) -- ++(0, 0.75);
    
    \draw [ultra thick, BrickRed] (2.25, 1.5) -- ++ (0, 0.75);

    \draw [ultra thick, BrickRed] (3, 1.5) -- ++ (0, 0.75);

    \node at (3.5, 1.8) {\color{BrickRed} $X$};
    
    \end{scope}


    \begin{scope}[yshift=-4cm]
    \foreach \y in {0, 0.75, 1.5, 2.25, 3} {
        \draw [thin, dashed] (0, \y) to (3, \y);
    }
    \foreach \x in {0, 0.75, 1.5, 2.25, 3} {
        \draw [thin, dashed] (\x, 0) to (\x, 3);
    }

    \draw [ultra thick, blue] (0, 1.5) -- ++(0.75, 0) -- ++(0, 0.75) -- ++(0.75,0) -- ++(0.75,0) -- ++(0, -0.75) -- ++(0.75, 0);

    \node at (-0.5, 1.5) {\color{blue} $Z$};

    \node at (4, 1.5) {$=$}; 

    \begin{scope}[xshift=5cm]
    \foreach \y in {0, 0.75, 1.5, 2.25, 3} {
        \draw [thin, dashed] (0, \y) to (3, \y);
    }
    \foreach \x in {0, 0.75, 1.5, 2.25, 3} {
        \draw [thin, dashed] (\x, 0) to (\x, 3);
    }

    \draw [ultra thick, blue] (0, 1.5) -- ++(0.75, 0) -- ++(0.75,0) -- ++(0.75,0) -- ++(0.75, 0);

    \node at (3.5, 1.5) {\color{blue} $Z$};
    \end{scope}
    \end{scope}
\end{tikzpicture}
\caption{In the ground space of the toric code, there are  two kinds of nontrivial operators, which are invariant under topological deformations.
The $X$-type line operators must pass through a face of the lattice when moving from one segment of the path to the next segment.
By contrast, the $Z$-type line operators must pass through a vertex.
}\label{fig:toric-lines}
\end{figure}
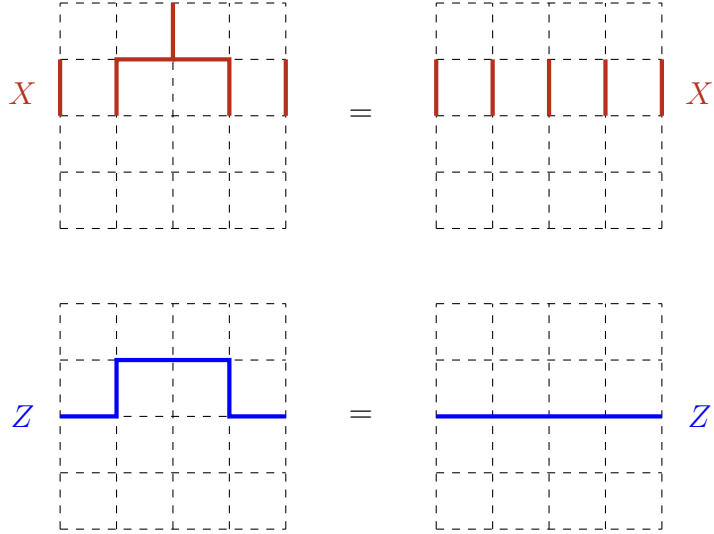

To any set of edges for the ground space of the toric code, we assign all multiples of the identity together with any line operator that can be represented on those sites edges.
This means that contractible regions contain only multiples of the identity, while strips that wrap a cycle on the torus contain the corresponding line operators.

It is not hard to show that the toric code algebras satisfy disjoint additivity in the sense of \eqref{def:precise-lattice}, provided that we say two edges are adjacent if they share a vertex or a face.
Note that as shown in figure \ref{fig:toric-lines}, ``$Z$-type'' topological lines pass through vertices, while ``$X$-type'' topological lines jump across faces.
To stop ourselves from building a strip containing a $Z$-type topological line out of a pair of disjoint, contractible regions, it is necessary to avoid pairs of regions that are adjacent across vertices; similarly, for the $X$-type topological line, it is necessary to avoid pairs of regions that are adjacent across faces.
However, for any non-adjacent pair of regions, disjoint additivity is clearly satisfied.
The only way that combining two regions can produce new operators is if the combined region contains topologically nontrivial loops that are not present in the individual regions; when we consider only non-adjacent regions, this is forbidden.

\bfig
\includegraphics[height=6cm]{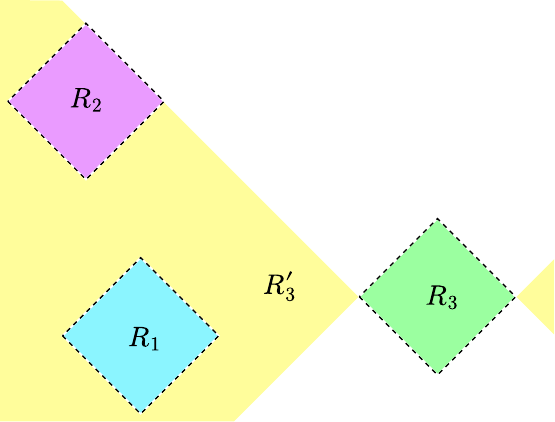}
\caption{Illustrating spatial disjointness.  The regions $R_1$ and $R_2$ are disjoint as sets, but they are not spatially disjoint.  Moreover $R_1\cup R_2$ is not a region at all, as it is not equal to $(R_1\cup R_2)''$.  The regions $R_1$ and $R_3$ are spatially disjoint.  The regions $R_2$ and $R_3$ are \textit{not} spatially disjoint, as the closure of $R_2$ is not contained in the spatial complement of $R_3$ (which is shaded yellow).}\label{dafig}
\efig
In the continuum, it is similarly useful to be slightly more restrictive on the regions $\R_1$ and $\R_2$ for which disjoint additivity is supposed to be satisfied.
For example in the nonrelativistic continuum case, suppose that $\R_1$ and $\R_2$ are two disjoint spatial regions with overlapping closures.
Since $\R_1 \cup \R_2$ and the interior of $\bar{\R_1} \cup \bar{\R_2}$ only differ by a set of measure zero, it might be natural in some settings to assign to the region $\R_1 \cup \R_2$ all of the operators contained in the interior of $\bar{\R_1} \cup \bar{\R_2}.$  But we can easily find examples of contractible $R_1$ and $R_2$ where the interior of $\bar{\R_1} \cup \bar{\R_2}$ is not contractible, for example we could shrink $R_1$ and $R_2$ in figure \ref{fig:tube} a bit, and the latter region can then contain nonlocal degrees of freedom such as Wilson lines which are not generated by the algebras for $R_1$ and $R_2$.  We do not want to rule out this choice, so we say that a nonrelativistic continuum field theory obeys \textit{disjoint additivity} if 
\begin{equation}\label{dadef}
\A\left(\R_1 \cup \R_2\right) = \A(\R_1) \vee \A(\R_2)
\end{equation}
for all open spatial regions $R_1$ and $R_2$ obeying
\be
\ol{R_1}\cap \ol{R_2}=\varnothing.
\ee
We think of the closures here as the continuum version of non-adjacency.  In the relativistic setting we instead require that the regions are \textit{spatially disjoint}, meaning that\footnote{In Minkowski spacetime $\overline{R_1}\subseteq R_2'$ implies that $\overline{R_2}\subseteq R_1'$, so we need only mention one of these.  In general however this implication does not hold, so it is better (and also more symmetric) to require both.}
\begin{align}\nonumber
\overline{R_1}&\subseteq R_2'\\
\overline{R_2}&\subseteq R_1'.\label{spdis}
\end{align}
In general however the set $R_1\cup R_2$ is not a region even when \eqref{spdis} holds, so in a relativistic theory we say that disjoint additivity holds if \eqref{dadef} is true for all spatially disjoint regions $R_1$, $R_2$ such that $(R_1\cup R_2)''=R_1\cup R_2$.\footnote{Examples of spatially disjoint regions where $R_1\cup R_2$ is not a region are not so simple to construct.  In appendix \ref{app:causality} we give an example in a globally hyperbolic spacetime, but also show that when $R_1$ and $R_2$ are each the domain of dependence of a regular open subset of a Cauchy surface (the Cauchy surface can be different for the two regions) then indeed we always have $(R_1\cup R_2)''=R_1\cup R_2$.}  We illustrate spatial disjointness in figure \ref{dafig}.  We emphasize that unlike ordinary additivity, the definition of disjoint additivity does not require us to assume that $M$ is globally hyperbolic or make any arbitrary choice of slice.

\subsection{Non-topological  higher-form global symmetries}

Thus far, we have focused on higher-form global symmetries that are topological, which is always the case in relativistic QFT \cite{Gaiotto:2014kfa}. 
On the other hand, in non-relativistic QFT or lattice systems, higher-form global symmetries can be non-topological \cite{Seiberg:2019vrp,Qi:2020jrf,Oh:2023bnk,Choi:2024rjm,Gorantla:2024ocs}.
Here we point out that non-topological higher-form symmetries are compatible with additivity. 
We demonstrate this   in  toric code, whose Hamiltonian is
\ie
H = - \alpha \sum_v U_v - \beta \sum_f U_f\,.
\fe
In contrast to the discussion in Section \ref{sec:definitions}, here we do not enforce the constraints $U_v=U_f=1$ strictly. 
Instead, the Hilbert space $\H$ 
is a tensor product of $N$ qubits and dim $\H=2^N$. States violating these local constraints are retained, but penalized energetically. 

In this bigger Hilbert space, the $X$-type and $Z$-type line operators are no longer topological.  It means that even for curves $\gamma$ and $\gamma'$ on the lattice that are homologically equivalent, their corresponding $Z$-type operators act differently on the entire Hilbert space, i.e., $Z(\gamma)\neq Z(\gamma')$ on $\H$. 
Specifically,  the equality in Figure \ref{fig:toric-lines} no longer holds because there are states in $\H$ violating the relation $U_f=+1$. 
A similar statement holds for the $X$-type operators.
Nonetheless, the $X$- and $Z$-type line operators still commute with the Hamiltonian and lead to conserved quantities. 
In this sense they generate non-topological one-form global symmetries. 

Since these one-form symmetry operators are not topological, the argument in Section \ref{sec:symmetry-violation} (in particular, Figure \ref{fig:endable})  no longer applies. 
On the entire Hilbert space $\H$, both  the $X$- and $Z$-lines are breakable, and there is no violation of additivity. 
Indeed, the operator algebra on $\H$ is just a matrix algebra Mat($2^{N},\mathbb{C})$, which satisfies additivity and Haag duality. 
We conclude that non-topological higher-form global symmetries do not necessarily violate additivity.

\begin{table}[h!]
\begin{center}
\begin{tabular}{|c|c|c|}
\hline 
 &toric code  & ground space \\
 && of toric code\\
\hline
additivity&$\checkmark$ & $\times$\\
\hline 
disjoint additivity & $\checkmark$ & $\checkmark$\\
\hline
~one-form global symmetry ~& non-topological&topological \\
\hline
\end{tabular}
\end{center}
\caption{While additivity is violated in the presence of topological higher-form global symmetries, it is compatible with non-topological ones.}
\end{table}

\section{Haag duality and disjoint additivity on the lattice}
\label{sec:lattice}

As reviewed in section \ref{sec:symmetry-violation}, ordinary additivity is violated in continuum quantum field theories in the presence of a higher-form symmetry.
Our proposal however is that all local quantum theories should respect Haag duality as defined in section \ref{sec:algebraic-notions} and disjoint additivity as defined in section \ref{sec:definitions}, even if higher-form symmetries are present.
To motivate this, in this section we study general lattice theories obtained by implementing a Lie-group constraint on a tensor product Hilbert space, and show that disjoint additivity and Haag duality are both respected.  Such theories are interesting in this context because they can have exact higher-form symmetries without requiring the subtlety of the continuum limit.  In addition to being a model for continuum behavior, this class of theories also contains many physically interesting condensed matter systems, such as Hamiltonian lattice gauge theories and the ground spaces of stabilizer Hamiltonians.\footnote{Interesting work in \cite{Naaijkens:1, Naaijkens:2, Naaijkens:3, Ogata:2021vom,Jones:2023ptg,Jones:2023xew} has studied Haag duality in lattice systems in infinite space where gauge constraints are not imposed exactly, and instead are imposed via energetic suppression. In these theories the algebras on finite lattices automatically satisfy Haag duality and additivity, while the algebras associated with infinite lattices only satisfy Haag duality due to the structure of the ground state wavefunction.
}

\subsection{Lattice gauge theories and stabilizer Hamiltonians}
\label{sec:Lie-constraint-definitions}

We consider a lattice system that, for simplicity, we will take to have a finite number $N$ of local degrees of freedom.
These could include  matter degrees of freedom on the sites and ``gauge field'' degrees of freedom on edges or higher-dimensional simplices.
While we restrict to a finite number of sites, we allow each on-site dimension to be infinite.
We assume that the Hilbert space carries the faithful, unitary action of a compact Lie group $\S$.
In the case of lattice gauge theory with gauge group $G$ and $N_v$ lattice vertices (note that $N\neq N_v$), this group is $\S = G^{N_v}$.  In the case of a stabilizer Hamiltonian, $\S$ is the stabilizer group. 
We note that a finite group is a special example of a compact Lie group, and we do not assume that $\S$ is abelian.
The assumption that the representation is ``faithful'' means that the map $s \mapsto U(s)$ with $s\in \S$ is injective.
As is usually assumed for representations of compact Lie groups, we also assume that the map is continuous with respect to the strong topology on Hilbert space.

The main assumption we will make about $\S$ is that its action factorizes between complementary lattice regions.
Namely, for any set of degrees of freedom $R$ and the complementary set $R',$ and for any $s \in \S,$ we have
\begin{equation} \label{eq:factorizing-unitary}
    U(s) = U_{R}(s) \otimes U_{R'}(s).
\end{equation}
While $U(s)$ defines a linear representation of $\S$, i.e., 
\ie
U(s_1 s_2) = U(s_1) U(s_2),
\fe
we do not make this assumption for its restrictions $U_R(s)$.  On the other hand we must have
\be
U_R(s_1)U_R(s_2)\otimes U_{R'}(s_1)U_{R'}(s_2)=U_R(s_1s_2)\otimes U_{R'}(s_1s_2),
\ee
so $U_R$ and $U_{R'}$ still furnish projective representations of $\S$ \cite{Else:2014vma}, since $A\otimes A'=B\otimes B'$ implies $A\propto B$ and $A'\propto B'$. To prove disjoint additivity, we will also need to make an assumption about $\S$ being generated by elements that only act on a few sites at a time; we will state this more precisely in section \ref{sec:lattice-DA}. 
Both assumptions hold for general stabilizer Hamiltonians and for lattice gauge theories.  Equation \eqref{eq:factorizing-unitary} can be thought of as imposing a form of what condensed matter theorists call an ``on-site'' symmetry action; this ensures that the symmetry does not have anomalies and can safely be gauged. 

If $\H$ is the``parent'' Hilbert space on which $\S$ acts, then there is an unambiguous assignment of operators $\A(R)$ to each subset $R$ of the $N$ degrees of freedom. These algebras trivially satisfy Haag duality and ordinary additivity due to the tensor product nature of $\H$.
Now consider the invariant subspace
\begin{equation}
    \H_{\S} \equiv \{|\psi\rangle \in \H\,|\, U(s)|\psi\rangle = |\psi\rangle \text{ for all } s \in \S\}.
\end{equation}
We let $P_{\S}$ be the projector onto this subspace.
To each region $R$, we assign the algebra
\begin{equation} \label{eq:gauge-invariant-algebra}
    \A_{\S}(R) \equiv \{P_\S a P_\S \, |\, a \in \A(R) \text{ and } [a, U(s)] = 0 \text{ for all } s \in \S\}.
\end{equation}
Moreover, we interpret this as an algebra acting on $\H_{\S}$.
In particular, Haag duality is the statement
\begin{equation} \label{eq:gauge-duality}
    \A_{\S}(R)' = \A_{\S}(R'),
\end{equation}
where the commutant on the left-hand side is taken within the space of operators acting on $\H_{\S},$ not within the space of operators acting on the parent Hilbert space $\H$.

We will show in the next subsection that Haag duality is satisfied in the form of equation \eqref{eq:gauge-duality}.
In the subsection after that, we will show that disjoint additivity is satisfied for regions separated by an adjacency relation that is set by details of the constraint group $\S.$

Before we proceed to the general proof, it is useful to have the following examples in mind.

\paragraph{Example 1: $\mathbb{Z}_2$ gauge theory} There is a qubit on every edge of a 2D spatial lattice with $N_v$ vertices. Let $X_e, Z_e$ be the Pauli operators acting on the qubit on edge $e$. 
At every vertex, the operator $U_v = \prod_{e \ni v} X_e$ in \eqref{TCv} 
implements a local $\mathbb{Z}_2$ gauge transformation. 
In the case of a $\mathbb{Z}_2$ gauge theory, the group is $\S =\mathbb{Z}_2^{N_v}$, and is generated by the operators $U_v$.
$\H_\S$ is the gauge-invariant subspace where $U_v=1$.\footnote{It is important to distinguish between the following two cases: (1) imposing the local constraints $U_v=1$ at every vertex, and (2) setting every $X$-loop to be 1. 
Case (2) is equivalent to restricting to the invariant sector of a one-form global symmetry. 
Case (2) enforces extra constraints coming from $X$-loops around non-contractible cycles, while case (1) only enforces $X$-loops around contractible cycles to be 1. As we will prove in this section, the operator algebra in case (1) satisfies both Haag duality and disjoint additivity, whereas it is straightforward to find a violation of disjoint additivity on a torus in case (2).}

\paragraph{Example 2: Toric code} 
The lattice setup is the same as the previous example, but in the case of the toric code, we impose both $U_v=1$ and $U_f=1$ as in \eqref{TCv} and \eqref{TCf}. 
In this case we have $\S= \mathbb{Z}_2^{N_v} \times \mathbb{Z}_2^{N_f}$ generated by $U_v$ and $U_f$, where $N_v$ is the number of vertices and $N_f$ is the number of faces. The subspace $\H_{\S}$ is the ground space of toric code where $U_v=1$ and $U_f=1$. 
It is clear that $U_v$ and $U_f$ satisfy \eqref{eq:factorizing-unitary}, but the restrictions to a subregion generally form a projective representation of $\S$.  

\paragraph{Example 3: Lattice $G$ gauge theory} On every (oriented) edge $e$, there is a Hilbert space $L^2(G),$ where $G$ is a compact Lie group. 
For finite $G$, the local Hilbert space is finite-dimensional and is  the group algebra $\mathbb{C}[G]$, i.e., the vector space of formal linear combinations of  group elements with complex coefficients. 
For continuous $G$, the local Hilbert space is infinite-dimensional.
At each vertex $v$ and for each group element $g \in G$, one defines the transformation $U_v(g)$ which left-multiplies each outward-pointing edge by $g$ and right-multiplies each inward-pointing edge by $g^{-1}.$
The constraint group $\S$ is generated by the set of $U_v(g)$ transformations, and is equal to $G^{N_v}$ where $N_v$ is the number of vertices. Charged matter degrees of freedom may also be included on the vertices, in which case $U_v(g)$ also implements a gauge transformation on the matter fields at vertex $v$.

For concreteness one may consider lattice $U(1)$ gauge theory \cite{Kogut:lattice, Banks:1975gq}, in which the (oriented) edge Hilbert space $L^2(S^1)$ is written in terms of a Fourier basis $|n\rangle_e$, and one defines an electric field operator $E_e$ and a Wilson line $W_{e}$ with action 
\ie
E_e \ket{n}_e = n \ket{n}_e \,,~~~~W_e \ket{n}_e = \ket{n-1}_e\,,~~~~W_e^\dagger \ket{n}_e = \ket{n+1}_e\,.
\fe
and commutation relations
\be
[E_e , W_{e'} ]= -\delta_{ee'} W_e\,,~~~~[E_e, W^\dagger_{e'}] =  \delta_{ee'}W_{e}^\dagger\,.
\ee
The unitary $U_v(\theta)$ that implements gauge transformations is given by
\ie
U_v(\theta) = \exp\left( -i \theta \sum_{e\ni v} \epsilon(e,v) E_e\right)\,,
\fe
where $\epsilon(e,v)$ is $+1$ if $e$ points away from $v$ and $-1$ if $e$ points toward $v.$
In this case, $\S=U(1)^{N_v}$, and the subspace $\H_\S$ is the gauge-invariant subspace satisfying the Gauss law $ \sum_{e\ni s} \epsilon(e,s) E_e=0$ at every vertex.

\paragraph{Example 4: Stabilizer code}
One begins with a collection of qubits, and picks an abelian subgroup $\S$ of the full Pauli group on all qubits, with the requirement that every element of $\S$ has an eigenvalue $+1.$ 
One then projects onto the simultaneous $+1$ eigenspace of every element of $\S$.
This is the same as summing over the action of the finite group $\S$.  The first two examples are special cases of this one.

\subsection{Haag duality}
\label{sec:lattice-HD}

Our first task will be to show that the commutation requirement in equation \eqref{eq:gauge-invariant-algebra} is redundant, and one can actually use the simpler identity
\begin{equation} \label{eq:projection-algebra}
    \A_{\S}(R) = \{P_\S a P_\S \, |\, a \in \A(R)\}.
\end{equation}

To see this, we first note that because $\S$ is a compact Lie group, it has a Haar measure $ds.$
The projection $P_\S$ can be implemented via the formula
\begin{equation} \label{eq:projector-integration}
    P_\S = \int ds\, U(s).
\end{equation}
To verify this, one uses the bi-invariance of the Haar measure and checks
\begin{equation}
    \left(\int ds\, U(s)\right)^2 = \int ds_1\, d s_2\, U(s_1) U(s_2) = \int ds_1\, ds_2\, U(s_1 s_2) = \int ds_1\, ds_2\, U(s_1) = \int ds_1\, U(s_1),
\end{equation}
together with the inversion-invariance of the Haar measure to obtain
\begin{equation}
    \left(\int ds\, U(s) \right)^{\dagger} = \int ds\, U(s^{-1}) = \int ds\, U(s).
\end{equation}
This verifies that $\int ds\, U(s)$ is an orthogonal projector.
It is easily seen that every vector $P_\S |\psi\rangle$ is invariant under the action of $\int ds\, U(s)$, and that every vector $\int ds\, U(s) |\psi\rangle$ is invariant under the action of $P_\S$; it follows that these two projections are equal.

Similarly one can show that for any operator $a$ acting on $\H$, the operator
\begin{equation} \label{eq:aS}
    a_\S \equiv \int ds\, U(s) a U^{\dagger}(s)
\end{equation}
commutes with every unitary $U(s)$ and thus also with $P_\S$. Moreover, by the factorization assumption \eqref{eq:factorizing-unitary}, if $a$ is in $\A(R)$, then so is $a_\S.$
If $a$ is already an operator commuting with each $U(s)$, then we clearly have $a = a_\S.$
Consequently, the set of operators in $\A(R)$ commuting with $\S$ is exactly the set of operators
\begin{equation}
    \{ a_{\S} \,|\, a \in \A(R)\}.
\end{equation}
We may therefore replace the definition in equation \eqref{eq:gauge-invariant-algebra} by
\begin{equation}
    \A_\S(R)
        = \{ P_\S a_\S P_\S | a \in \A(R)\}.
\end{equation}
But using equation \eqref{eq:projector-integration}, one obtains
\begin{align}
    \begin{split}\label{asa}
    P_\S a_\S P_\S
        & = \int ds_1\, ds_2\, ds_3\, U(s_1) U(s_2) a U(s_2^{-1}) U(s_3) \\
        & = \int ds_1\, ds_2\, ds_3\, U(s_1 s_2) a U(s_2^{-1} s_3) \\
        & = \int ds_1\, ds_3\, U(s_1) a U(s_3)\\
        & = P_\S a P_\S,
    \end{split}
\end{align}
where in the last line we have used the invariance of the Haar measure to substitute $s_3 \mapsto s_2 s_3$ and $s_1 \mapsto s_1 s_2^{-1}.$
This verifies equation \eqref{eq:projection-algebra}.

We now know that each $\A_\S(R)$ is of the form $P \A(R) P$ for a particular projection $P$.
This is reminiscent of a fundamental theorem in the study of von Neumann algebras, called ``compression duality.''
Given a von Neumann algebra $\A$ on a Hilbert space $\H$, and given a projection $P$, one defines $\A_P$ to be the algebra $P \A P$ acting on the image of the projection $P$.
In the case where $P$ is an element of $\A$, it is known --- see for example \cite[chapter 2.1]{dixmier2011neumann}--- that one has $(\A_P)' = (\A')_P.$
This is the statement we want for Haag duality (eq.\eqref{eq:gauge-duality}), but we cannot apply compression duality as it is presented in textbooks, because the  projection $P_\S$ is not contained in any of the local algebras $\A(R).$
Nevertheless, we will be able to follow the proof technique of \cite{dixmier2011neumann} to prove the direction $\A_\S(R)' \subseteq \A_\S(R')$ --- which does not actually require that the projection live in $\A(R)$ --- and we will be able to show the opposite inclusion by an independent argument.

To show $\A_\S(R)' \subseteq \A_\S(R')$, one wishes to show that for any operator $U' \in \A_\S(R)',$   there exists some $\hat{U}' \in \A(R')$ with $P_\S \hat{U}' P_\S = U'.$
Our first simplification will be to use the fact --- see e.g. \cite[chapter 1.3]{dixmier2011neumann} --- that every von Neumann algebra is generated by its unitary elements, so it suffices to take $ U'$ to be unitary.
This is a unitary operator acting within the invariant subspace $\H_\S$.
The operator $\hat{U}'$, which acts on the full space $\H$, is supposed to commute with every operator in $\A(R)$.
So for any $a \in \A(R)$ and any vector $|\psi\rangle \in \H,$ we should have
\begin{equation}
    \hat{U}' a P_\S |\psi\rangle = a \hat{U}' P_\S |\psi\rangle.
\end{equation}
We also want $P_{\S} \hat{U}' P_\S = U',$ and an easy way to guarantee this is to demand that $\hat{U}'$ acts on the support of $P_\S$ as $U'$, or in other words that $\hat{U}'P_\S=U'P_\S$.  This motivates the definition
\begin{equation} \label{eq:U-extension-definition}
    \hat{U}' a P_\S |\psi\rangle = a U' P_{\S} |\psi\rangle.
\end{equation}
This defines the action of $\hat{U}'$ on any vector of the form $a P_\S |\psi\rangle,$ and linearity and continuity tell us how it should act on the full subspace
\begin{equation} \label{eq:U-extension-subspace}
    V
        = \bar{\operatorname{span}\{a P_\S |\psi\rangle\,|\, a \in \A(R)\, ,|\psi\rangle \in \H\}}.
\end{equation}
To verify that this gives a well defined operator $\hat{U}',$ one computes
\begin{align}
    \begin{split}
    \left\lVert \hat{U}' \left( a P_\S |\psi\rangle - \tilde{a} P_\S |\tilde{\psi}\rangle \right) \right\rVert^2
        & = \left\lVert  a U' P_\S |\psi\rangle - \tilde{a} U' P_\S |\tilde{\psi}\rangle \right\rVert^2 \\
        & = \langle P_\S \psi | (U')^{\dagger} P_\S a^{\dagger} a P_\S U' |P_{\S} \psi\rangle
            + \text{terms with $a \mapsto \tilde{a}, |\psi\rangle \mapsto |\tilde{\psi}\rangle$}.
    \end{split}
\end{align}
In each of these terms, $U'$ may be commuted through the operator of the form $P_\S (\dots) P_\S,$ and by unitarity, it may be annihilated against $(U')^{\dagger}.$
This gives
\begin{align}
    \begin{split}
    \left \lVert \hat{U}' \left( a P_\S |\psi\rangle - \tilde{a} P_\S |\tilde{\psi}\rangle \right) \right \rVert^2
        & = \left \lVert \left( a P_\S |\psi\rangle - \tilde{a} P_\S |\tilde{\psi}\rangle \right) \right \rVert^2,
    \end{split}
\end{align}
from which one sees that $\hat{U}'$ is well defined on the subspace $V$ from equation \eqref{eq:U-extension-subspace}.
To complete the definition of $\hat{U}'$, we define it to vanish on the orthocomplementary subspace $V^{\perp}$.

By construction, we clearly have $P_\S \hat{U}' P_\S = U'.$
All that remains to demonstrate $\A_\S(R)' \subseteq \A_\S(R')$ is to show that $\hat{U}'$ is in $\A(R').$
Since Haag duality holds in the parent Hilbert space $\H$, we need only show that $\hat{U}'$ commutes with every $a \in \A(R).$
The definition given in equation \eqref{eq:U-extension-definition} clearly shows that $[\hat{U}', a]$ vanishes on the subspace $V$.
We must show that it also vanishes on the $V^{\perp}$, and since we have defined $\hat{U}'$ to vanish on $V^{\perp}$, this entails showing
\begin{equation}
    \hat{U}' a  |\psi\rangle = 0, \qquad |\psi\rangle \in V^{\perp}.
\end{equation}
If we let $Q$ be the projector onto $V$, then we may equivalently show the operator equation
\begin{equation}
    \hat{U}' a (1-Q) = 0.
\end{equation}
But $V$ is a subspace left invariant under the action of $\A(R)$, and it is an elementary calculation to show that any projection onto an invariant subspace for $\A(R)$ lives in the commutant algebra $\A(R)'.$\footnote{For $\A$ a von Neumann algebra with invariant subspace $V$ and projection $P$, one clearly has $P a P = a P$ and $P a^{\dagger} P = a^{\dagger} P$ for all $a \in \A,$ from which one may compute
\begin{equation}
    a P = P a P = (P a^{\dagger} P)^{\dagger} = (a^{\dagger} P)^{\dagger} = P a.
\end{equation}}
So $Q$ commutes with $a,$ hence $(1-Q)$ commutes with $a,$ and we have
\begin{equation}
    \hat{U}' a (1-Q) = \hat{U}' (1-Q) a = 0,
\end{equation}
as desired.
This completes the proof that $\hat{U}'$ is in $\A(R'),$ and hence completes the proof of the inclusion $\A_\S(R)' \subseteq \A_\S(R').$

The opposite inclusion is much easier.
For any operator $a' \in \A(R')$ and any $a \in \A(R)$, the projection $P_\S$ commutes with $a_\S$ and $a'_\S,$ so we have
\begin{align} \label{eq:gauge-projective-locality}
    \begin{split}
        (P_\S a' P_\S) (P_\S a P_\S)
            & = (P_\S a'_\S P_\S) (P_\S a_\S P_\S) \\
            & = P_\S a'_\S a_\S P_\S \\
            & = P_\S a_\S a'_\S P_\S \\
            & = (P_\S a_\S P_\S) (P_\S a'_\S P_\S) \\
            & = (P_\S a P_\S) (P_\S a' P_\S),
    \end{split}
\end{align}
where we have used $a_\S\in \A(R), a'_\S\in \A(R')$ and microcausality in the parent algebra. 
From equation \eqref{eq:gauge-projective-locality} we conclude that every operator in $\A_\S(R')$ commutes with every operator in $\A_\S(R)$; this gives the $\supseteq$ inclusion in equation \eqref{eq:gauge-duality}, and we have established Haag duality for the invariant algebras.

\subsection{Disjoint additivity}
\label{sec:lattice-DA}
\begin{figure}
\centering
\begin{tikzpicture}
    \draw [ultra thick] (-2, -1) circle (1);
    \draw [ultra thick] (2, -1) circle (1);
    \fill[color=RoyalBlue, opacity=0.4] (-2, -1) circle (1.2);
    \draw [dashed] (-2, -1) circle (1.2);
    \draw [dashed] (2, -1) circle (1.2);
    \fill[color=RoyalBlue, opacity=0.4] (2, -1) circle (1.2);

    \draw (-3.5, -2.5) rectangle ++(7, 3.5);
    
    \begin{scope}
    \path[clip] (-2, -1) circle (1)[insert path={(-3.5, -2.5) rectangle ++(7, 3.5)}];
    \path[clip] (2, -1) circle (1)[insert path={(-3.5, -2.5) rectangle ++(7, 3.5)}];
    \fill[color=BrickRed, opacity=0.2] (-3.5, -2.5) rectangle ++(7, 3.5);
    \end{scope}

    \node [RoyalBlue] at (-2, 0.5) {$\S_1$};
    \node [RoyalBlue] at (2, 0.5) {$\S_2$};
    \node [BrickRed] at (0, -1) {$\N$};
\end{tikzpicture}

\caption{The black circles denote two non-adjacent regions on a lattice system acted on by a group $\S$. The subgroup $\N$ acts only on the complementary region, shaded in red. The groups $\S_1$ and $\S_2$ act on the regions shaded in blue --- these are slight thickenings of $R_1$ and $R_2,$ but they do not overlap with one another. It is assumed that $\S_1, \S_2,$ and $\N$ collectively generate the full group $\S$.}

\label{fig:G1-G2}
\end{figure}
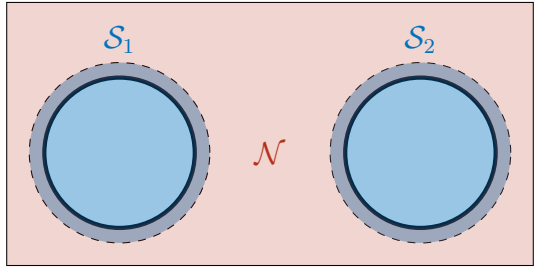
As emphasized in section \ref{sec:definitions}, we only expect disjoint additivity to hold on the lattice for regions that are ``non-adjacent.''
We will define $R_1$ and $R_2$ to be non-adjacent if $\S$ can be generated by a set of elements for which no single generator acts nontrivially on both $R_1$ and $R_2.$
Concretely, we fix non-overlapping regions $R_1$ and $R_2,$ and let $\N$ be the subgroup of $\S$ that acts trivially on both regions:
\begin{equation}
    \N \equiv \{s \in \S\,|\, U(s) \in \A(R_1 \cup R_2)'\}.
\end{equation}
We say $R_1$ and $R_2$ are \textit{non-adjacent} if there exist Lie subgroups $\S_1, \S_2 \subseteq \S$ such that (see figure \ref{fig:G1-G2})
\begin{enumerate}[(i)]
    \item $U(\S_1) \subseteq \A(R_2)'$ and $U(\S_2) \subseteq \A(R_1)'$;
    \item $\S_1$ and $\S_2$ commute, i.e., $s_1 s_2 = s_2 s_1$ for $s_{(1,2)} \in \S_{(1,2)}$;
    \item $\S_1, \S_2,$ and $\N$ collectively generate $\S$.
\end{enumerate}
These are the most general conditions under which we are able to prove disjoint additivity of $\A_\S(R_1)$ and $\A_\S(R_2)$.  In the most familiar lattice models with Lie group constraints however, such as lattice gauge theories and stabilizer codes, a stronger condition is true: the groups $\S_1, \S_2,$ and $\N$ do not just generate $\S$ --- we actually have the decomposition $\S = \S_1 \times \S_2 \times \N.$  In the main text we will therefore prove disjoint additivity under this additional assumption, with the substantially more technical proof of the general case relegated to appendix \ref{app:general-DA}.

The idea of the proof will be to show that the constraint projection $P_\S$ factorizes between operators in $R_1$ and $R_2$ --- that is, we wish to show
\begin{equation} \label{eq:projection-factorization}
    P_\S (O_1 O_2) P_\S = (P_\S O_1 P_\S) (P_\S O_2 P_\S),
\end{equation}
for $O_{(1,2)}\in \A(R_{(1,2)})$.
If we can prove this, it straightforwardly implies disjoint additivity.  Indeed by equation \eqref{eq:projection-algebra}, every operator in $\A_\S(R_1 \cup R_2)$ can be written as $P_{\S} T P_{\S}$ for an operator $T \in \A(R_1 \cup R_2).$
But by ordinary additivity $\A(R_1 \cup R_2) = \A(R_1) \vee \A(R_2)$, together with the commutativity of $\A(R_1)$ with $\A(R_2)$, every such $T$ can be written as a strongly convergent series of factorizing terms:
\begin{equation}
    T = \sum_{j} T_j^{(1)} T_j^{(2)}.
\end{equation}
Moving $P_{\S}$ into the sum gives
\begin{equation}
    P_{\S} T P_{\S}
        = \sum_{j} P_{\S} T_{j}^{(1)} T_{j}^{(2)} P_\S.
\end{equation}
So if equation \eqref{eq:projection-factorization} holds, then $P_\S T P_\S$ can be generated by elements of $\A_{\S}(R_1)$ and $\A_{\S}(R_2)$.
The inclusion $\A_\S(R_1 \cup R_2) \subseteq \A_\S(R_1) \vee \A_\S(R_2)$ follows, and the opposite inclusion is obvious.
We will therefore spend the rest of this section proving equation \eqref{eq:projection-factorization}.

To prove equation \eqref{eq:projection-factorization}, we use equation \eqref{asa} on both sides, and the fact that $a_\S$ and $P_\S$ commute, to rewrite it as
\begin{equation}
    P_\S (O_1 O_2)_\S P_\S
        = (P_\S (O_1)_\S P_\S) (P_\S (O_2)_\S P_\S) 
        = P_\S (O_1)_\S (O_2)_\S P_\S
\end{equation}
So it suffices to show $(O_1 O_2)_\S = (O_1)_\S (O_2)_\S$.
The left-hand side is defined as
\begin{equation} \label{eq:prefactorized-integral}
    (O_1 O_2)_\S
        = \int ds\, U(s) O_1 O_2 U(s)^{\dagger}.
\end{equation}
But because $\S$ factorizes as $\S_1 \times \S_2 \times \N,$ we can rewrite this integral as
\begin{equation}
    (O_1 O_2)_\S
        = \int_{\S_1} ds_1\, \int_{\S_2} ds_2\, \int_{\N} dx\, U(s_1 s_2 x) O_1 O_2 U(s_1 s_2 x)^{\dagger}.
\end{equation}
The $U(x)$ conjugation acts trivially on $O_1$ and $O_2,$ the $U(s_1)$ conjugation acts trivially on $O_2,$ and the $U(s_2)$ conjugation acts trivially on $O_1.$
This gives
\begin{equation}
    (O_1 O_2)_\S
        = \left[ \int_{\S_1} ds_1\, U(s_1) O_1 U(s_1)^{\dagger} \right] \left[  \int_{\S_2} ds_2\, U(s_2) O_2 U(s_2)^{\dagger} \right].
\end{equation}
Each integral on the right-hand side can be lifted to the full group $\S$ to obtain
\begin{align}
    \begin{split}
    (O_1 O_2)_\S
        & = (O_1)_\S (O_2)_\S,
    \end{split}\label{Osplit}
\end{align}
establishing equation \eqref{eq:projection-factorization} and completing the proof of disjoint additivity.  The general proof in appendix \ref{app:general-DA} still proceeds by establishing \eqref{Osplit}, but without the factorization of $\S$ we need to use more Lie theory to show it.

It is important to understand what this theorem does  and does not say.  Namely it promises that disjoint additivity holds for non-adjacent regions, but it does \textit{not} promise that continuum regions which are spatially disjoint in the non-relavistic ($\ol{R_1}\cap \ol{R_2}=\varnothing$) or relativistic ($\ol{R_1}\subseteq R_2'$ and $\ol{R_2}\subseteq R_1'$) senses defined in section \ref{sec:definitions} arise as the continuum limits of non-adjacent regions on the lattice.  For example nothing stops us from taking $\S$ to be a global symmetry, in which case all sites are adjacent to each other in the sense defined in this section.  In this case the proof we just gave does not imply that the continuum limit obeys disjoint additivity in the continuum sense, and indeed we will see in the next section that it does not.  For a Hamiltonian lattice theory to give rise to a continuum theory that obeys disjoint additivity we therefore need to introduce an additional requirement: we say that a lattice adjacency rule as defined by (i-iii) above is \textit{short-range} if regions whose minimal separation in lattice units goes to infinity with the system size are non-adjacent for all but finitely many system sizes.  This condition holds for all lattice gauge theories, as well for the standard stabilizer constructions such as the toric code, but it does not hold if we take $\S$ to be a global symmetry.  

This discussion is illustrative of a broader issue: our lattice theorems on Haag duality and disjoint additivity are exact results about lattice systems with a finite number of degrees of freedom, but spatial locality cannot really be defined unambiguously for a lattice system with a finite number of sites since we can always just think of it as living in $0+1$ spacetime dimensions.  The dimension of spacetime only becomes clear in a limit where the number of lattice sites goes to infinity, due to taking a continuum and/or thermodynamic limit.

\section{Some examples of violation}
\label{sec:violation-examples}

So far we have shown that Haag duality and disjoint additivity hold for a large class of theories that are usually considered to be local.  This is true even when additivity is violated due to higher-form symmetries. Another way to add credibility to our proposal of Haag duality and disjoint additivity as locality principles is to show that they are violated in theories which are expected to have some non-local character.
In particular, one or both of these principles should be violated in examples (1)-(4) from the introduction.  We will now argue that this is indeed the case.

\bfig
\includegraphics[height=7cm]{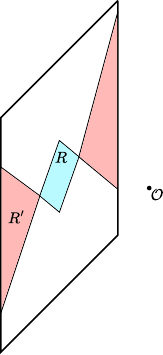}
\caption{Violating Haag duality in the hypersurface theory: if a local operator $\mathcal{O}$ is placed at a point off of the hypersurface which is spacelike to $R$, then $\mathcal{O}$ is in $\A(R)'$. If we also choose $\mathcal{O}$ to be spacelike separated from $R'$, however, then it cannot be in $\A(R')$. If it were, then $\mathcal{O}$ would commute with everything in a small neighborhood of itself, and in a field theory with nontrivial local operators this cannot be true. Thus we have $\A(R)' \neq \A(R')$.
The same figure also shows a violation of Haag duality for any generalized free field theory obtained as the boundary limit of a free field in $AdS_{d+1}$.}\label{hypersurfacefig}
\efig

\subsection{Restriction to a hyperplane}

We first consider example (1) from the introduction --- the ``quantum field theory'' obtained by taking a standard quantum field theory, then assigning local algebras only to regions obeying $R=R''$ within the hypersurface $x=0$ in $d$-dimensional Minkowski space (see figure \ref{hypersurfacefig}).
To be precise, we consider the theory whose Hilbert space is the Hilbert space of the full $d$-dimensional quantum field theory, and where a region $R$ is assigned the algebra of all operators that can be expressed in an arbitrarily small thickening of $R$.
In other words, $\A(R)$ is the set of all operators that can be expressed by smearing fields in $R \times (-\epsilon, \epsilon)$ for $\epsilon$ arbitrarily small.
The timelike tube theorem \cite{Borchers:tube, Strohmaier:2023opz} tells us that the algebra of the full hypersurface is equal to the algebra of all operators on Hilbert space.
The local algebras $\A(R)$ clearly obey nesting and spacelike commutativity, and disjoint additivity is inherited from the higher-dimensional theory.  However, they do not obey Haag duality, as illustrated in figure \ref{hypersurfacefig}.\footnote{The violation requires the existence of a point spacelike to $R$ and $R'$.  Such points exist for generic choices of $R$.  For example in $2+1$ dimensions we can parameterize the half-space using coordinates $(t,x,z)$ with $z \geq 0$.  Choosing $R$ to be the domain of dependence within the hypersurface $z=0$ of the interval $-1<x<1$ at $t=0$, any point $(0,0,z)$ with $z>1$ is spacelike to both $R$ and $R'$.  If we interpret the half-space as the Poincare patch of AdS, as we should for the generalized free field of the following subsection, then the statement that such points exist is equivalent to saying that either $R$ or $R'$ has the property that its causal wedge is smaller than its entanglement wedge, which is the generic situation.}

\subsection{Generalized free fields}

Our next example, example (2) from the introduction, is a generalized free field.  This is a field with Gaussian correlators that obeys microcausality, but that does not obey any equation of motion.  For example in $d$-dimensional Minkowski space, a scalar generalized free field has Heisenberg representation
\be
\Phi(x)=\int \frac{d^d k}{(2\pi)^d} e^{i k\cdot x}f(k^2)a_k+\mathrm{h.c.},
\ee
where $k$ is restricted to be timelike or null and future-pointing, but is not required to obey an on-shell condition. The annihilation operator $a_k$ obeys the usual algebra
\be
[a_k,a_{k'}^\dagger]=(2\pi)^d \delta^d(k-k'),
\ee 
and $f$ is an arbitrary function.  $\A(R)$ is defined as the algebra generated by bounded functions of the operators obtained by smearing $\Phi$ against smooth functions compactly supported in $R$.  It is not difficult to check that this field obeys the microcausality condition \eqref{mccond}.  

A nice way to obtain such a field is as the boundary limit 
\be
\Phi(x)=\lim_{z\to 0}z^{-(d/2+\nu)}\Phi_{\text{bulk}}(z,x)
\ee
of an ordinary free scalar field in $AdS_{d+1}$ spacetime with $\nu=\frac{1}{2}\sqrt{d^2+4m^2}$.  This gives a generalized free field with
\be
f(k^2)=\frac{\sqrt{\pi}}{2^{\nu}\Gamma(1+\nu)}\left(\sqrt{-k^2}\right)^\nu,
\ee
so by varying $\nu$ we can get a large class of examples.  In all of these examples, Haag duality is violated in just the same way as for the hypersurface theory: a local operator in the AdS field theory which is located away from the AdS boundary and spacelike to the boundary spacetime regions $R$ and $R'$, as in figure \ref{hypersurfacefig}, is in $\A(R)'$ but not in $\A(R')$.  

This violation of Haag duality was recently discussed in the context of large-$N$ gauge theories in \cite{Leutheusser:2024yvf}; our $\mathcal{A}(R)$ is the same as their ``single-trace algebra'' $\mathcal{Y}_R$. 
It was also noted in \cite{Leutheusser:2024yvf} that additivity as we have defined it is violated for these algebras, although a weaker ``set-theoretic'' version still holds. 
Disjoint additivity however is not violated, as $\A(R_1\cup R_2)$ is generated by field operators in $R_1$ and $R_2$ and thus by elements of $\A(R_2)$ and $\A(R_2)$.  It is interesting to consider what happens if we instead use the bulk entanglement wedge algebras, denoted  $\mathcal{X}_R$ in \cite{Leutheusser:2024yvf}, as our boundary algebras instead of the $\mathcal{Y}_R$: Haag duality is then restored, since the entanglement wedges of complementary regions are complementary, but disjoint additivity is now violated since the entanglement wedge of the union of two boundary regions that are spatially disjoint can contain points which are not in the entanglement wedge of either region separately.  For more on the emergence of bulk algebras from large-$N$ gauge theories see \cite{Kelly:2016edc,Faulkner:2020hzi,Leutheusser:2021frk,Leutheusser:additivity}.

\subsection{Invariant sector under a global symmetry}\label{sec:inv}

Example (3) from the introduction is the theory constructed from a standard quantum field theory with a compact internal global symmetry group $G$ by restricting to states that are invariant under $G$.
In this theory, the local algebras contain only the operators which are invariant under the symmetry.  In the notation of section \ref{sec:lattice-HD}, we take the local algebras to be $\A_G(R)$.  In a lattice setting this theory obeys Haag duality by the proof of section \ref{sec:lattice-HD}, and we expect the same to be true in the continuum.\footnote{In \cite{Shao:2025mfj} however it was shown that the symmetric sector with respect to a non-invertible global symmetry can violate Haag duality. }  It does not however obey disjoint additivity, as we will now explain.\footnote{This theory was previously discussed in \cite{Casini:2019kex,Casini:2020rgj,Shao:2025mfj,Jia:2025bui} (see also section III.4.2 of \cite{Haag:book}) as an example where additivity is violated; the new point here is that disjoint additivity is also violated, unlike for the examples of additivity violation that arise from higher-form global symmetries. See also \cite{Ji:2019jhk,Kong:2020cie,Ji:2021esj,Chatterjee:2022kxb,Liu:2022cxc,Chatterjee:2022jll,Jones:2023ptg,Jones:2023xew,Inamura:2023ldn,Jones:2024lws}  for recent discussions of such invariant sectors on the lattice  in the context of topological order.  }  

Before giving the explanation however, to avoid confusion we first emphasize that this construction does \textbf{not} give the theory we would get by \textit{gauging} the global symmetry $G$.  In the gauged theory there is a dynamical Wilson line that can connect pairs of charged operators, and also twisted sector operators (and states) that do not exist in the ungauged theory.  Moreover the Wilson line must be included in the support of any operator which contains it.  In particular for the situation shown in figure \ref{fig:symsector} there would need to be a Wilson line connecting $\mathcal{O}_a$ and $\mathcal{O}^\dagger_a$, and then the resulting operator would not lie in $\A(R_1\cup R_2)$, so disjoint additivity would not be violated.  In general we expect that gauging any non-anomalous global symmetry in a theory obeying Haag duality and disjoint additivity results in a theory that also obeys these conditions; section \ref{sec:lattice-HD} basically showed that this is the case in a lattice context.

 \begin{figure}
\centering
\begin{tikzpicture}[scale=0.7]  
    \draw[ultra thick] (45:2) arc[start angle=45, end angle=135, radius=2];

      \draw[ultra thick] (225:2) arc[start angle=225, end angle=315, radius=2];

    \draw[ultra thick,blue] (315:2) arc[start angle=-45, end angle=45, radius=2];

    \draw[ultra thick, blue] (135:2) arc[start angle=135, end angle=225, radius=2];

    \foreach \angle in { 45, 135,  225, 315} {
   \pgfmathsetmacro\x{2*cos(\angle)}
    \pgfmathsetmacro\y{2*sin(\angle)}
    \pgfmathsetmacro\dx{0.2*cos(\angle)} 
    \pgfmathsetmacro\dy{0.2*sin(\angle)}
        
    \draw[ultra thick] (\x-\dx, \y-\dy) -- (\x+\dx, \y+\dy);
    }

\filldraw[black] (-2,0) circle (3pt) node[anchor=east]{$\mathcal{O}_a(x_1)$};

\filldraw[black] (2,0) circle (3pt) node[anchor=west]{$\mathcal{O}_a^\dagger(x_2)$};

    \node at (1.6, 0) {\footnotesize \textbf{\color{blue}$\R_2$}};
    \node at (-1.6, 0) {\footnotesize \textbf{\color{blue}$\R_1$}};

\end{tikzpicture}
\caption{Disjoint additivity is violated in the invariant sector under a compact global symmetry group $G$.}\label{fig:symsector}
\end{figure}
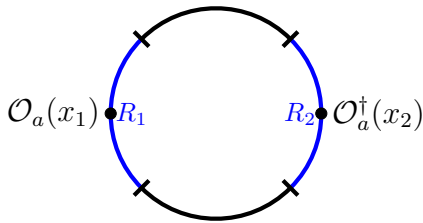
\bfig
\includegraphics[height=4cm]{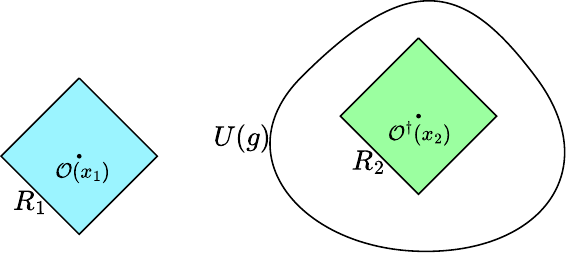}
\caption{A symmetry insertion under which the bilocal operator \eqref{OO} transforms nontrivially.}\label{pifig}
\efig
Returning now to the invariant sector theory with algebras $\A_G(R)$, let $R_1$ and $R_2$ be spatially disjoint regions either in the non-relativistic continuum sense that $\overline{R_1}\cap \overline{R_2}=\varnothing$ or in the relativistic continuum sense that $\overline{R_1}\subseteq R_2'$ and $\overline{R_2}\subseteq R_1'$, and let $x_1\in R_1$ and $x_2\in R_2$ (in the relativistic case we also assume $(R_1\cup R_2)''=R_1\cup R_2$).  Moreover, let $\mathcal{O}_a$ be a local operator transforming in a finite-dimensional representation of a compact internal symmetry group $G$.  All such representations are unitary, so the operator
\be\label{OO}
\sum_a \mathcal{O}_a(x_1) \mathcal{O}_a^\dagger(x_2)
\ee
is an invariant element of $\A_G(R_1\cup R_2)$. On the other hand it seems clear that it cannot be an element of $\A_G(R_1)\vee \A_G(R_2)$ since it is built out of operators in $R_1, R_2$ that are charged under the symmetry.  See figure \ref{fig:symsector} for an illustration.  This claim is most easily formalized using the path integral approach: the global symmetry $G$ is implemented by a codimension-one surface operator $U(g)$ that we can insert on a submanifold that links with $x_2$ but not $x_1$.  See figure \ref{pifig}.  The symmetry insertion causes the bilocal operator to acquire a nontrivial transformation, which would not be the case if it was built only out of neutral operators.  To convert this into an algebraic argument however we need the idea of \textit{splittability}: this is the statement that for any quantum field theory in  Minkowski space every global symmetry has the property that for each region $R$ we can introduce a localized symmetry operator $U(g,R)$ that acts like $U(g)$ on operators in $R$ and as the identity on operators in $R'$.\footnote{Strictly speaking if we want $U(g,R)$ to be a well-defined operator we should smooth out the transition around the boundary of $R$.  One way to do this is to only require $U(g,R)$ to act as the identity in a region $\hat{R}\subset R'$ such that there is a positive minimum distance between points in $R$ and points in $\hat{R}$.}  Splittability is automatic for lattice symmetries obeying our condition \eqref{eq:factorizing-unitary}, and for relativistic field theory in Minkowski space it has been proven under rather weak assumptions on the spectrum of the energy density \cite{Buchholz:1986dy,Buchholz:1986bg,Buchholz:1985ii}.\footnote{See sections 2.1-2.2 of \cite{Harlow:2018tng} for more discussion of splittability, where in particular it was pointed out that splittability can fail on spacetime manifolds other than $\mathbb{R}^d$ when higher-form global symmetries are present.}  We can then conjugate the bilocal operator by $U(g,R_1)$ or $U(g,R_2)$ to get the same nontrivial transformation we had in the path integral approach, which would not be possible if the operator lived in $\A_G(R_1)\vee \A_G(R_2)$.

\subsection{Virasoro identity multiplet}

Example (4) from the introduction is the ``Virasoro identity multiplet'' sector constructed from the energy-momentum tensor correlation functions of a full-fledged $1+1$ dimensional unitary CFT.  More precisely, the algebra $\mathcal{A}(R)$ is defined to be the von Neumann algebra generated by bounded functions of the smeared energy-momentum tensor. 
While this sector of operators is closed under OPE, they cannot satisfy disjoint additivity and Haag duality simultaneously. 
This follows from the result of \cite{Benedetti:2024dku}, which shows that any theory obeying both of these algebraic principles must be modular invariant.\footnote{Although the authors of \cite{Benedetti:2024dku} assume  additivity in their proof, it is straightforward to see that only disjoint additivity is required. It is also worth noting that not every modular-invariant theory satisfies additivity and Haag duality. One such example is the 1+1d $\mathbb{Z}_2$ gauge theory, which is a topological field theory (and thus a trivial example of a CFT). It is modular invariant (i.e., its torus partition function equals 2), but it violates additivity due to the presence of an unbreakable Wilson loop.  See footnote \ref{fn:BF}.} 
For $c\ge 1$, the Virasoro identity character (i.e., the torus zero-point conformal block) is $|\chi_0(q)|^2 ={|q|^{-\frac{c-1}{12}} \over|\eta(q)|^2} |1-q|^2$, 
which is not modular-invariant.
(Here $\eta(q)$ is the Dedekind eta function.)  
Therefore, the Virasoro identity multiplet has to violate either disjoint additivity or Haag duality, or both. 

There are special cases where we can be more explicit about the violation of Haag duality and/or disjoint additivity.  Indeed for the $c<1$ minimal models, as well as the more general (diagonal) rational CFTs, it was shown in \cite{Shao:2025mfj} that the identity multiplet for the chiral algebra cannot satisfy both disjoint additivity and Haag duality.\footnote{The only exceptions are tensor products of holomorphic and antiholomorphic CFTs,  such as the non-chiral $(\mathfrak{e}_8)_1$ WZW model or the Monster$\times \overline{\text{Monster}}$ CFT, where the identity multiplet with respect to the extended chiral algebra is the entire modular-invariant CFT.}
More specifically, it was shown that  disjoint additivity is violated if the RCFT fusion algebra contains an invertible subgroup, and Haag duality is violated if the total fusion algebra is not a group. 
For instance, the fusion algebra  of the Ising CFT is
\ie
\epsilon\times \epsilon=1\,,\quad\sigma\times\epsilon=\epsilon\times\sigma = \sigma\,,\quad \sigma\times \sigma =1+\epsilon.
\fe
This is not a group, but it contains a $\mathbb{Z}_2$ subgroup generated by $\epsilon$. It follows that  the Virasoro identity multiplet violates both disjoint additivity and Haag duality.\footnote{In fact, the identity multiplet of a diagonal RCFT is the invariant sector under the non-invertible global symmetry generated by the Verlinde lines \cite{Verlinde:1988sn,Petkova:2000ip,Chang:2018iay}. Therefore, this also counts as an example of the kind in Section \ref{sec:inv}. For instance, the Virasoro identity multiplet of the Ising CFT is the invariant sector of the non-invertible Kramers-Wannier symmetry \cite{Oshikawa:1996dj,Petkova:2000ip,Frohlich:2004ef,Chang:2018iay} (which includes a $\mathbb{Z}_2$ subgroup).} 

Finally we can discuss a similar nonlocal theory constructed from an irrational CFT: the free compact scalar field theory of $\phi$ with $c=1$ at a generic radius. 
The global symmetry contains two $U(1)$ symmetries, known as the momentum and winding symmetries. 
(See, for example,  \cite{Lin:2019kpn,Thorngren:2021yso,Cheng:2022sgb,Pace:2024oys} for recent reviews of $c=1$ CFTs and their symmetries.)
The $\mathfrak{u}(1)$ current algebra multiplets are labeled by their charges under these two $U(1)$ symmetries, and in particular the identity and its descendants are neutral.  On the other hand, $e^{in\phi}$ carries charge $n$ under the momentum $U(1)$, while the charged operators under the winding $U(1)$ involve the dual scalar field.  Our claim is that the theory obtained by defining the local algebras to be generated only by the energy-momentum tensor and the two $U(1)$ currents violates disjoint additivity. Although the bi-local operator $e^{i\phi(x_1)}e^{-i\phi(x_2)}$ is neutral under both $U(1)$ symmetries, and hence belongs to this identity multiplet, it cannot be expressed as a product of two operators from the identity multiplet supported in disjoint regions. Consequently, it violates disjoint additivity, analogous to the discussion around \eqref{OO}.

Thus we see that the combination of Haag duality and disjoint additivity is a powerful criterion for enforcing locality.

\section{Fermions}
\label{sec:fermions}
In this section we discuss how to generalize our locality rules to incorporate system with fermions.  In fermionic theories, the algebra of operators is graded by a $\mathbb{Z}_2$ symmetry operator $(-1)^F$, usually called fermion parity. 
An operator $\cal O$ is said to be bosonic if ${\cal O}(-1)^F = (-1)^F {\cal O}$, and fermionic  if ${\cal O} (-1)^F = - (-1)^F {\cal O}$.  In particular, $(-1)^F$ itself is bosonic.  In relativistic theories, $(-1)^F$ can be interpreted as a $2\pi$ spatial rotation in any plane.  In this section we will refer to the quantity appearing in the microcausality condition \eqref{mccond} as the \textit{supercommutator}
\be
[A,B]_\pm = AB - (-1)^{|A||B|}BA\,,
\ee
where $|A|=0$ ($=1$) if $A$ is bosonic (fermionic).  Given a set $\mathcal{X}$ of operators in a $\mathbb{Z}_2$-graded algebra on a Hilbert space $\mathcal{H}$, the \textit{supercommutant} $\mathcal{X}'$ is defined in the obvious way:
\be
\mathcal{X}'=\Big\{x'\in \mathcal{B}(\mathcal{H})\Big| \,[x,x']_{\pm}=0, \forall x\in \mathcal{X}\Big\}.
\ee
With commutants replaced by supercommutants, we advocate that the algebras of operators ${\cal A}({\cal R})$ in a local fermionic system should obey all the properties discussed in earlier sections, including disjoint additivity and Haag duality.

We now illustrate this definition using one of the simplest fermionic quantum systems, the Majorana chain in 1+1D \cite{Kitaev:2000nmw}. 
The space is a discrete set of lattice points labeled by $j=1,2,\cdots, L$, with a single Majorana operator $\chi_j$ at each site. 
These Majorana operators are self-adjoint and obey the standard Clifford algebra
\ie
\{ \chi_j , \chi_{j'} \} = 2\delta_{jj'}\,.
\fe
Importantly, we do not associate a local Hilbert space with a Majorana site; a Hilbert space appears after choosing a representation of the algebra, but the local operator algebras are generated by the Majorana degrees of freedom rather than by tensor factors of the representation.

When $L$ is even, the fermion parity operator can be defined as
\ie
(-1)^F = i^{L/2} \chi_1\chi_2\cdots \chi_L\,~~~~(\text{even}~L)\,,
\fe
where the phase $i^{L/2}$ is inserted to ensure $[(-1)^F]^2=1$.
The fermion parity operator obeys $(-1)^F \chi_j (-1)^F = - \chi_j$. 
There is a natural assignment of an algebra to each spatial region $\cal R$ (which is just a discrete set of sites), where $\A(\R)$ is defined to be the von Neumann algebra generated by $\chi_j$ with $j\in \R$. 
This algebra assignment automatically obeys microcausality and additivity (and hence disjoint additivity), and in appendix \ref{sec:Even_Majo_chain_Haag_duality} we show that it also obeys Haag duality.

The situation is more subtle when $L$ is odd.   
(See \cite{Lee:2013wya,Stanford:2019vob,Delmastro:2021xox,Witten:2023snr,Seiberg:2023cdc,Freed:2024apc,Seiberg:2025zqx} for various other subtleties  with an odd number of fermions.)
Let us start with the case of $L=3$ Majorana operators.  
Following \cite{Seiberg:2023cdc}, we consider two representations of the Clifford algebra. 
The first is an irreducible two-dimensional representation in terms of the Pauli matrices:
\ie
&\chi_1 = X ,~~~\chi_2 = Y ,~~~\chi_3 = Z .
\fe
However,  in this representation, there is no fermion parity operator $(-1)^F$ implementing the automorphism $\chi_j \to -\chi_j$. In other words, fermion parity is an outer automorphism. 
As a result, this representation is not $\mathbb{Z}_2$-graded, and it is not possible to consistently assign a fermion parity to every operator.\footnote{For instance, assigning all three Majorana operators $\chi_j$ to be fermionic is inconsistent with the relation $\chi_1 \chi_2 = i \chi_3$ in this irreducible representation.}
Consequently the supercommutant is not defined, so Haag duality cannot be formulated for this representation.\footnote{It is also inconsistent to use the ordinary commutator in an ungraded Hilbert space to discuss microcausality and Haag duality;  this would incorrectly require two different fermions $\chi_{j_1}$ and $\chi_{j_2}$ to commute, rather than anticommute. }

We thus consider a second representation, which is four-dimensional and reducible:
\ie
\label{eqn:3_Majo_Rep}
&\chi_1 = X \otimes \mathbf{1} ,~~~\chi_2 =Y\otimes \mathbf{1} ,~~~\chi_3 = Z\otimes X.
\fe
The automorphism $\chi_j\to -\chi_j$ is inner in this representation; it is implemented by 
\ie
(-1)^F=  Z\otimes Z \,.
\fe
Hence, this representation is  $\mathbb{Z}_2$-graded, allowing us to define the supercommutant and examine Haag duality. 
We define ${\cal A}( R)$ in the same way as in the even $L$ case so that additivity is obeyed.  
 For example, let $\R$ be the subregion with sites 
  $j\in\{1,2\}$, then $\A(\R)$ is generated by $\{\chi_1,\chi_2\}$ and $\A(\R^\prime)$ is generated by $\{\chi_3\}$.
 Now, consider the operator $\mathcal{O}=Z\otimes Y$. 
 It is a fermion, and 
 is in $\mathcal{A}(R)^\prime$ but not in $\mathcal{A}(R^\prime)$. 
 Therefore, Haag duality is violated.
 Another way to see the violation is to note that $\cal O$ supercommutes  with every operator generated by $\chi_1,\chi_2,\chi_3$, and is therefore in $\A(\{1,2,3\})'$.
However, it is not the identity operator. 
This causes a violation of Haag duality for the entire space, since we required the algebra of the empty set to be trivial (see equation \eqref{Aempty}).  Either way this violation has a straightforward interpretation: this representation really has four fermions, with $\chi_4=Z\otimes Y$, and ``forgetting'' this fermion got us in trouble just as ``forgetting'' the rest of spacetime in the hypersurface/generalized free field theories or the charged operators under a global symmetry did.  

 This discussion can be trivially generalized to $L=2\ell-1$ fermions. 
 If we choose to work with the $2^{\ell-1}$-dimensional irreducible representation of the Clifford algebra, we cannot formulate the notion of Haag duality because of the lack of a $\mathbb{Z}_2$ grading. 
 Instead, we work with the $2^{\ell}$-dimensional, reducible, $\mathbb{Z}_2$-graded representation, given explicitly in \eqref{eqn:Majo_Rep}. 
 Then the non-identity operator $\mathcal{O}=(Z)^{\otimes \ell-1}\otimes Y$ is in $\A(\{1,2,\cdots, L\})'$, violating Haag duality for the entire space. 
We can again interpret the violation operator as a   $2\ell$-th fermion $\chi_{2\ell}$ that we irresponsibly ``forgot''. 

We conclude that both (disjoint) additivity and Haag duality hold for the Majorana chain with an even number of Majorana fermions, but they cannot be simultaneously satisfied when the number of fermions is odd.
Therefore the odd Majorana chain is not a local quantum theory according to our rules, although it makes perfect sense as a quantum mechanical system.

Finally, let us comment on continuum field theories with fermions. 
While the number of fermion modes in field theory is always infinite, it is meaningful to count the number of fermion zero modes. 
For instance, consider the 1+1d Lagrangian for a free, massless, Majorana fermion
\ie 
{\cal L} ={ i \over2}\chi_\text{L}(\partial_t -\partial _x )\chi_\text{L}
+{i\over2}\chi_\text{R} (\partial_t +\partial _x )\chi_\text{R}
\,.
\fe
Here $\chi_\text{L}$ and $\chi_\text{R}$ are the one-component, left- and right-moving Majorana-Weyl fermions, respectively. They pair up to be a non-chiral Majorana fermion.

Let space be a circle. 
This field theory arises from the continuum limit of the Majorana chain with the following gapless Hamiltonian 
\ie
H=  i \sum_{j=1}^L \chi_j \chi_{j+1}\,.
\fe
More specifically, the even-$L$ lattice model corresponds to imposing periodic boundary conditions for both $\chi_\text{L}$ and $\chi_\text{R}$.  
This is known as the Ramond-Ramond (RR) boundary condition in string theory. 
In this case, there are two fermion zero modes, one from each of $\chi_\text{L}$ and $\chi_\text{R}$. 
If, on the other hand, $L$ is odd, the lattice model corresponds to imposing anti-periodic and periodic 
boundary conditions for $\chi_\text{L}$ and $\chi_\text{R}$, respectively \cite{Seiberg:2023cdc}. 
This is known as the Neveu-Schwarz-Ramond (NSR) boundary condition in string theory. 
This theory has a single fermion zero mode from $\chi_\text{R}$, while all the non-zero modes come in pairs.\footnote{A similar field theory is the (chiral) Majorana-Weyl fermion of $\chi_\text{R}$ in 1+1d on a spatial circle with periodic (Ramond) boundary condition. 
It also has an unpaired fermion zero mode. 
However, this field theory has a gravitational anomaly and cannot arise from simple lattice models.} 

Since the theory with the NSR boundary condition arises from the odd-$L$ Majorana chain, we expect a similar violation of Haag duality. 
Indeed, this follows from quantizing the unpaired fermion zero mode, denoted as $\tilde \chi$, on a spatial circle.
Similar to our lattice discussion, we have two options: 
(1) We  choose   an irreducible representation for $\tilde \chi$ acting on the one-dimensional ground space. Then $\tilde\chi$ acts as the identity operator, and we cannot define a fermion parity $(-1)^F$ operator. It follows that the supercommutator cannot be defined and we cannot formulate Haag duality in this case. 
(2) We choose a reducible representation for $\tilde\chi$ acting on the two-dimensional ground space. For instance, we can choose $\tilde\chi = X$ and $(-1)^F=Z$ on this ground space. But then there is an operator ${\cal O}=Y$ which supercommutes with every other operator. This again causes a violation ${\cal A}(S^1)'\neq{\cal A}(\varnothing)$ of Haag duality under the axiom \eqref{Aempty} that the empty set algebra is trivial.  This  shows that  a single free Majorana fermion with the NSR boundary condition is  nonlocal as a 1+1d quantum theory according to our rules. 
This nonlocality is also reflected in the thermal partition function of this theory computed from the path integral presentation. This calculation produces an overall factor of $\sqrt{2}$, which defies a standard Hilbert space interpretation. See \cite{Stanford:2019vob,Delmastro:2021xox,Witten:2023snr,Seiberg:2023cdc,Freed:2024apc} for recent discussions on this point.

\section{Disjoint additivity restoration and SymTFT}
\label{sec:SymTFT}

In sections \ref{sec:violation-examples} and \ref{sec:fermions} we constructed some examples of algebra assignments to regions which obey microcausality but not Haag duality and/or disjoint additivity.  From our point of view the quantum theories equipped with these algebra assignments are not local.  Each of these examples was constructed by ``forgetting'' some of the operators in a genuinely local theory, so the obvious way to restore locality is to ``remember'' the operators that we forgot.  In this section however we show that there is sometimes another possibility: instead of changing the algebras, we can change the regions to which they are assigned.  More concretely, we will show that a nonlocal algebra assignment in $d$ spacetime dimensions can sometimes be given a local interpretation in $d+1$ dimensions.  

The example we consider is closely related to an idea in the quantum field literature called Symmetry Topological Field Theory (SymTFT).  
For conventional symmetries, the symmetry algebra and possible operator transformations are determined by the symmetry group $G$.  For more general notions of symmetry, the symmetry algebra, as well any generalized 't Hooft anomalies, are encoded in certain categorical data.  
The idea of SymTFT is to  specify a topological field theory in one dimension higher whose topological defects capture all of this symmetry data on the boundary.    The simplest example of this idea is that when we have a standard zero-form symmetry with a finite symmetry group $G$, the SymTFT is just the $G$-gauge theory in one dimension higher. 
See \cite{Gaiotto:2014kfa,Kong:2015flk,Freed:2018cec,Pulmann:2019vrw,Ji:2019jhk,Kong:2020cie,Gaiotto:2020iye,Apruzzi:2021nmk,Lin:2022dhv,Freed:2022qnc} for an incomplete list of references on SymTFTs, and also \cite{Elitzur:1989nr} for earlier, related setups.

In this section, we will consider again the algebras $\A_G(R)$ of operators that are invariant under a compact internal global symmetry group $G$ in a local field theory on a $d$-dimensional spacetime $M$.
In section \ref{sec:violation-examples}, we saw that these algebras violate disjoint additivity.   What we will argue is that when $G$ is a finite group these algebras can be re-interpreted as local algebras for a $G$ gauge theory on a contractible $(d+1)$-dimensional spacetime $N$ obeying $M=\partial N$, with  the theory we started with at the boundary of $N$.  With this re-interpretation, disjoint additivity is now satisfied, as we should expect since this theory manifestly respects locality.  

One way to interpret this result, and indeed also our hyperplane and generalized free field examples from section \ref{sec:violation-examples}, is that ``a local theory should have the decency to know in which spacetime dimension it lives.''

\subsection{A continuum example}
As a warmup we will first consider the theory of a charged scalar field in $d$ dimensions on a spacetime $M$ coupled to a $d+1$ dimensional $U(1)$ gauge field on a spacetime $N$ such that $M=\partial N$.  The bulk Lagrangian density is
\be
L=-\frac{1}{2e^2}F\wedge \star F
\ee
and the boundary Lagrangian density is
\be
\ell=-D\phi^*\wedge \star D\phi,
\ee
with
\begin{align}\nonumber
D\phi&=d\phi-iq A \phi\\
D\phi^*&=d\phi^*+iq A \phi^*.
\end{align}
In our conventions $q$ is an integer.  The variation of the action is\footnote{It is easier to check this if you know that if $\omega$ and $\sigma$ are $p$-forms then $\omega\wedge \star \sigma=\sigma\wedge \star \omega$.}
\be\label{varS}
\delta S=-\frac{1}{e^2}\int_N \delta A\wedge d\star F+\int_M\left( \delta \phi^*D\hspace{-.1cm}\star \hspace{-.1cm}D\phi+\delta\phi D\hspace{-.1cm}\star \hspace{-.1cm}D\phi^*+\delta A\wedge\left(\star\hspace{0cm} J-\frac{1}{e^2}\star F\right)\right),
\ee
with 
\be
\star J=iq\big(\phi \star \hspace{-.1cm}D\phi^*-\phi^*\star\hspace{-.1cm} D\phi\big)
\ee
being the boundary electric current, and as usual we have dropped boundary terms in the future/past (see e.g. \cite{Harlow:2019yfa}).  The first three terms in \eqref{varS} vanish by the bulk and boundary equations of motion, so to get a good variational theory we need to choose boundary conditions to make the last term vanish.  One way to do this is to adopt Dirichlet boundary conditions, where we fix the pullback of $A$ to $M$ so that $\delta A|_M=0$.  With this choice however the bulk and boundary are decoupled, which is not interesting for our purposes.  A better choice is to instead set
\be\label{U1bc}
\frac{1}{e^2}\star F|_M=\star J,
\ee
which are the boundary conditions for the surface of a perfect conductor (which indeed has mobile surface charges).  These boundary conditions are gauge-invariant, so to get a well-defined theory we need to quotient by all gauge transformations, including those that are nontrivial on $M$.\footnote{More formally, this follows because all infinitesimal gauge transformations are zero modes of the presymplectic form. In the notation of \cite{Harlow:2019yfa}, we have 
\be
X_\epsilon\cdot \wt{\Omega}=\int_{\partial\Sigma}\epsilon \,\delta\left(\frac{1}{e^2}\star F-\star J\right),
\ee
which vanishes identically by the boundary condition \eqref{U1bc} for all infinitesimal gauge transformations $\epsilon$.
}We thus need to restrict to boundary operators which are invariant under the $U(1)$ gauge symmetry.  

In this example however the local algebras do not coincide with the ones we would construct by treating $U(1)$ as a global symmetry on the boundary and going to the invariant sector.  For one thing there are nontrivial bulk operators such as the field strength $F$, and for another there is a dynamical boundary Wilson line. This is related to the fact that the standard SymTFT formalism so far only seems to work for finite symmetries.  For example there is no standard notion of a topological $U(1)$ gauge theory that could be used as a SymTFT.\footnote{When the boundary dimension is two, the $U(1)_k$ Chern-Simons theory, for example, does not fit the bill. This is  because the Wilson line charges have periodicity $k$ while in the boundary theory all integer charges are distinct.} See \cite{Brennan:2024fgj,Antinucci:2024zjp,Apruzzi:2024htg,Bonetti:2024cjk,Gagliano:2024off,Jia:2025jmn} for recent ideas on how to formulate SymTFTs for continuous global symmetries.

To get an example that does work, we will take the same boundary theory but instead couple to a $\mathbb{Z}_p$ gauge field in the bulk.  This has bulk Lagrangian density
\be\label{ZpL}
L=-\frac{p}{2\pi}dA\wedge B,
\ee
where $A$ a $U(1)$ 1-form gauge field and $B$ is a $U(1)$ $(d-1)$-form gauge field \cite{Maldacena:2001ss,Banks:2010zn,Kapustin:2014gua}.  The variation of the action is now
\be\label{varS2}
\delta S=-\frac{p}{2\pi}\int_N \left(\delta A\wedge dB+dA\wedge \delta B\right)+\int_M\Big( \delta \phi^*D\hspace{-.1cm}\star \hspace{-.1cm}D\phi+\delta\phi D\hspace{-.1cm}\star \hspace{-.1cm}D\phi^*+\delta A\wedge\big(\star\hspace{-.1cm} J-\frac{p}{2\pi}B\big)\Big),
\ee
so to get a good variational principle we need either 
\be
\delta A=0
\ee
or 
\be\label{Zpbc}
\frac{p}{2\pi}B|_M=\star J.
\ee
The former again would decouple the bulk and boundary, so we will adopt the latter.  We therefore must quotient by all $A$ gauge transformations $A'=A+d\Lambda_A$, but only by $B$ gauge transformations $B'=B+d\Lambda_B$ where $\Lambda_B|_M=0$ (transformations where $\Lambda_B|_M$ is a nonzero closed form are also allowed but they are not gauged).\footnote{More formally, for infinitesimal gauge transformations $\epsilon_A$ and $\epsilon_B$ from the presymplectic form we have
\begin{align}\nonumber
X_{\epsilon_A}\cdot \wt{\Omega}&=\int_{\partial \Sigma}\epsilon_A\,\delta\left(\frac{p}{2\pi}B-\star J\right)\\
X_{\epsilon_B}\cdot \wt{\Omega}&=\frac{p}{2\pi}\int_{\partial \Sigma}\delta A \wedge \epsilon_B,
\end{align}
so we have a zero mode for any $\epsilon_A$ due to the boundary condition \eqref{Zpbc} but only for those $\epsilon_B$ which vanish at the boundary.
}
The $A$ quotient again projects onto neutral operators in the boundary theory, but now the Wilson line which connects them is topological so it can be deformed into the bulk.  Moreover the Wilson line is undetectable if the charge is a multiple of $p$, so local operators whose charges are multiples of $p$ are still allowed.\footnote{More explicitly we can define a local operator which removes an infinitesimal ball in spacetime and then puts $2\pi$ units of $dB$ flux on boundary of this ball.  This operator is normally trivial due to the equation of motion $dB=0$, but it can fuse with the endpoint of a Wilson line of charge $p$ to produce a gauge-invariant endpoint.  It can therefore also fuse with a matter operator of charge $p$ to get a gauge-invariant local operator.  If we view the $\mathbb{Z}_p$ gauge theory as arising from Higgsing a $U(1)$ gauge field, then this local operator is the phase of the Higgs field.}  

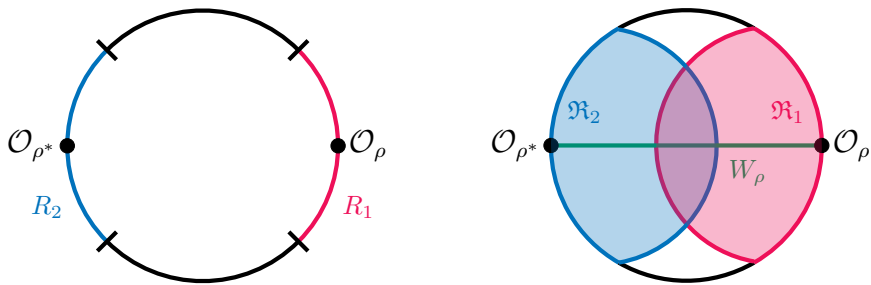
\begin{figure}[h!]
\centering
\raisebox{.5cm}{\begin{tikzpicture}[scale=.9]  
    \draw[ultra thick] (45:2) arc[start angle=45, end angle=135, radius=2];

      \draw[ultra thick] (225:2) arc[start angle=225, end angle=315, radius=2];

    \draw[ultra thick, OrangeRed] (315:2) arc[start angle=-45, end angle=45, radius=2];

    \draw[ultra thick, RoyalBlue] (135:2) arc[start angle=135, end angle=225, radius=2];

    \foreach \angle in { 45, 135,  225, 315} {
        
        \pgfmathsetmacro\x{2*cos(\angle)}
        \pgfmathsetmacro\y{2*sin(\angle)}

        \pgfmathsetmacro\dx{0.2*cos(\angle)}  
        \pgfmathsetmacro\dy{0.2*sin(\angle)}

        \draw[ultra thick] (\x-\dx, \y-\dy) -- (\x+\dx, \y+\dy);
    }

\filldraw[black] (-2,0) circle (3pt) node[anchor=east]{${\cal O}_{\rho^*}$};

\filldraw[black] (2,0) circle (3pt) node[anchor=west]{${\cal O}_{\rho}$};

    \node at (2.3, -.9) {\footnotesize \textbf{\color{OrangeRed}$\R_1$}};
    \node at (-2.3, -0.9) {\footnotesize \textbf{\color{RoyalBlue}$\R_2$}};

\end{tikzpicture}}
~~~~~
\raisebox{.6cm}{
\begin{tikzpicture}[scale=0.9]  
    \draw[ultra thick] (60:2) arc[start angle=60, end angle=120, radius=2];

      \draw[ultra thick] (240:2) arc[start angle=240, end angle=300, radius=2];

    \draw[ultra thick, OrangeRed] (300:2) arc[start angle=300, end angle=420, radius=2];

    \draw[ultra thick, RoyalBlue] (120:2) arc[start angle=120, end angle=240, radius=2];

    \draw[ultra thick,OrangeRed] (60:2) arc[start angle=100, end angle=260, radius=1.75];
    
    \draw[ultra thick,RoyalBlue] (240:2) arc[start angle=-80, end angle=80, radius=1.75];

 \draw [ultra thick, ForestGreen] (-2,0) -- (2,0);

  \filldraw[black] (-2,0) circle (3pt) node[anchor=east]{${\cal O}_{\rho^*}$};

\filldraw[black] (2,0) circle (3pt) node[anchor=west]{${\cal O}_{\rho}$};

\node at (0.9, -.4) {\footnotesize \textbf{\color{ForestGreen}$W_\rho$}};

\node at (1.5, 0.5) {\footnotesize \textbf{\color{OrangeRed}${\Rbulk}_1$}};

\node at (-1.5, 0.5) {\footnotesize \textbf{\color{RoyalBlue}${\Rbulk}_2$}};

  \fill[RoyalBlue, opacity=0.3]
    (120:2) arc[start angle=120, end angle=240, radius=2] 
    arc[start angle=-80, end angle=80, radius=1.75]
    -- cycle;

  \fill[OrangeRed, opacity=0.3]
    (300:2) arc[start angle=300, end angle=420, radius=2]
    arc[start angle=100, end angle=260, radius=1.75]
    -- cycle;
    
\end{tikzpicture}
}
\caption{Restoration of disjoint additivity by coupling to a bulk gauge field. Left: A pair of charged local operators violates disjoint additivity in the $G$-symmetric sector ${\cal A}_G$ of a boundary system because it belongs to ${\cal A}_G(\R_1\cup \R_2)$ but not in ${\cal A}_G(\R_1)\vee{\cal A}_G(\R_2)$ with $\ol{\R_1}\cap \ol{\R_2}=\emptyset$. Right: A 2+1D $G$ gauge theory on a disk coupled to the 1+1d  system by gauging $G$ on the boundary. The pair of charged operators is now connected by a Wilson line $W_\rho$ in the bulk.  Disjoint additivity is restored since the two regions ${\Rbulk}_1,{\Rbulk}_2$ in the bulk overlap, i.e., ${\Rbulk}_1\cap{\Rbulk}_2\neq \emptyset$.}\label{fig:restoration}
\end{figure} 
Specializing now to the case where $N$ is contractible, there are no gauge-invariant bulk degrees of freedom so the full Hilbert space of this theory is just the set of states obtained by acting on the vacuum with boundary operators that are invariant under the $\mathbb{Z}_p$ subgroup of $U(1)$.  This is precisely the same Hilbert space we would get by going to the invariant sector under the $\mathbb{Z}_p$ \textit{global} symmetry of the free complex scalar field, and for any $d+1$ dimensional region $R$ which does not intersect the boundary the local operator algebra $\A(R)$ is trivial.  So the nontrivial operator algebras are those where $R$ intersects the boundary, and it is natural to require this intersection to be equal to its double spacelike complement since the only local propagation in this theory is within the boundary.  On the other hand the bulk description allows for a new feature: a boundary region which is disconnected can be connected through the bulk.  This allows for bilocal operators in two boundary regions built out of operators which are charged under $\mathbb{Z}_p$ to be connected through a Wilson line in the bulk,  leading to a restoration of disjoint additivity as shown in figure \ref{fig:restoration}. 

\subsection{A lattice example}

We'll now demonstrate this restoration of disjoint additivity quite explicitly in the case of a 1+1D lattice model which was discussed in \cite{Ji:2019jhk}. 
In fact, the choice of the Hamiltonian will not be important for the following discussion of operator algebra. 
Consider a 1D spatial lattice on a circle where the qubits live on the edges $e=1,2,\cdots, L$. 
The Hilbert space is  a tensor product of local qubit Hilbert space, ${\cal H} = \bigotimes_{e=1}^L {\cal H}_e$, where ${\cal H}_e\simeq \mathbb{C}^2$, and is $2^L$-dimensional.
The full algebra of operator ${\cal A} (S^1) = \text{Mat}(2^L,\mathbb{C})$ is the algebra of $2^L\times 2^L$ matrices, which obeys (disjoint) additivity and Haag duality. 

Next, we restrict to the $\mathbb{Z}_2$-even sector with respect to the spin-flip operator  $U= \prod_{e=1}^L 
Z_e$. 
The $\mathbb{Z}_2$-even subalgebra (also known as the bond algebra) is 
\ie
{\cal A}_{\mathbb{Z}_2}(S^1)
&\equiv
\Big\{ 
{\cal O}\in \text{Mat}(2^L ,\mathbb{C}) ~\Big|  {\cal O} U =U{\cal O}
\Big\}
=  \Biggl\langle
~Z_e~,~ X_{e}X_{e+1}~
 \Biggr\rangle_{e=1,2,\cdots ,L} \,.
\fe
This subalgebra acts on the $\mathbb{Z}_2$-even subspace of ${\cal H}$, defined as ${\cal H}_{\mathbb{Z}_2} = \left\{ \ket{\phi} \in {\cal H} ~| ~U\ket{\phi}=\ket{\phi}\right\}$, which is $2^{L-1}$-dimensional and is  no longer a tensor product Hilbert space.  Defining regions $R_1=e_1$ and $R_2=e_2$, $X_{e_1} X_{e_2}$ belongs to ${\cal A}_{\mathbb{Z}_2}(\R_1\cup \R_2)$ but is not in ${\cal A}_{\mathbb{Z}_2}(\R_1)\vee{\cal A}_{\mathbb{Z}_2}(\R_2)$.  Disjoint additivity is therefore violated for any adjacency relation that does not say all links are adjacent; in particular it is violated for any adjacency relation which is ``short-range'' in the sense discussed at the end of section \ref{sec:lattice-DA}. 

 We now couple this boundary theory to a 2+1D lattice $\mathbb{Z}_2$ gauge theory in the bulk.  The spatial lattice has the topology of a two-disk $B^2$ and the boundary is a circle of $L$ edges. 
For simplicity, we assume the lattice takes the form of a ``wheel" with $2L$ edges in total. 
 While this does not look like a 2+1D system in the large $L$ limit, every statement we will make below holds true even if one adds more edges in the bulk. 
 In this simplifying 2D spatial lattice, the total Hilbert space is a tensor product of $2L$ qubits, and is $2^{2L}$-dimensional and the  entire operator algebra is the matrix algebra Mat$(2^{2L},\mathbb{C})$ of $2^{2L}\times 2^{2L}$ matrices.

\begin{figure}
\centering
\begin{tikzpicture}[scale=0.7]  
    \draw[ultra thick] (45:2) arc[start angle=45, end angle=135, radius=2];

      \draw[ultra thick] (225:2) arc[start angle=225, end angle=315, radius=2];

    \draw[ultra thick] (315:2) arc[start angle=-45, end angle=45, radius=2];

    \draw[ultra thick] (135:2) arc[start angle=135, end angle=225, radius=2];

    \draw [ultra thick] (45:2) -- (225:2);
    \draw [ultra thick] (90:2) -- (270:2);
    \draw [ultra thick] (135:2) -- (315:2);
    \draw [ultra thick] (180:2) -- (360:2);
    \draw [ultra thick] (225:2) -- (45:2);

    \pgfmathsetmacro\x{1.3*cos(0)}
    \pgfmathsetmacro\y{1.3*sin(0)}
    \node at (\x, \y) {\textbf{$X$}};

 \pgfmathsetmacro\x{1.3*cos(40)}
    \pgfmathsetmacro\y{1.3*sin(40)}
    \node at (\x, \y) {\textbf{$X$}};
 
 \pgfmathsetmacro\x{1.3*cos(90)}
    \pgfmathsetmacro\y{1.3*sin(90)}
    \node at (\x, \y) {\textbf{$X$}};

 \pgfmathsetmacro\x{1.3*cos(130)}
    \pgfmathsetmacro\y{1.3*sin(130)}
    \node at (\x, \y) {\textbf{$X$}};

 \pgfmathsetmacro\x{1.3*cos(180)}
    \pgfmathsetmacro\y{1.3*sin(180)}
    \node at (\x, \y) {\textbf{$X$}};

  \pgfmathsetmacro\x{1.3*cos(220)}
    \pgfmathsetmacro\y{1.3*sin(220)}
    \node at (\x, \y) {\textbf{$X$}};
    
 \pgfmathsetmacro\x{1.3*cos(270)}
    \pgfmathsetmacro\y{1.3*sin(270)}
    \node at (\x, \y) {\textbf{$X$}};

 \pgfmathsetmacro\x{1.3*cos(310)}
    \pgfmathsetmacro\y{1.3*sin(310)}
    \node at (\x, \y) {\textbf{$X$}};

\node at (3,0) {$=1$};
\end{tikzpicture}
~~~~
\begin{tikzpicture}[scale=0.7]  
    \draw[ultra thick] (45:2) arc[start angle=45, end angle=135, radius=2];

      \draw[ultra thick] (225:2) arc[start angle=225, end angle=315, radius=2];

    \draw[ultra thick] (315:2) arc[start angle=-45, end angle=45, radius=2];

    \draw[ultra thick] (135:2) arc[start angle=135, end angle=225, radius=2];

    \draw [ultra thick] (45:2) -- (225:2);
    \draw [ultra thick] (90:2) -- (270:2);
    \draw [ultra thick] (135:2) -- (315:2);
    \draw [ultra thick] (180:2) -- (360:2);
    \draw [ultra thick] (225:2) -- (45:2);

    \pgfmathsetmacro\x{1.3*cos(0)}
    \pgfmathsetmacro\y{1.3*sin(0)}
    \node at (\x, \y) {\textbf{$Z$}};

 \pgfmathsetmacro\x{1.3*cos(50)}
    \pgfmathsetmacro\y{1.3*sin(50)}
    \node at (\x, \y) {\textbf{$Z$}};
 
 \pgfmathsetmacro\x{2.1*cos(22.5)}
    \pgfmathsetmacro\y{2.1*sin(22.5)}
    \node at (\x, \y) {\textbf{$Z$}};

\node at (3,0) {$=1$};
\end{tikzpicture}
~~~~~~
\begin{tikzpicture}[scale=0.7]  
    \draw[ultra thick] (45:2) arc[start angle=45, end angle=135, radius=2];

      \draw[ultra thick] (225:2) arc[start angle=225, end angle=315, radius=2];

    \draw[ultra thick] (315:2) arc[start angle=-45, end angle=45, radius=2];

    \draw[ultra thick] (135:2) arc[start angle=135, end angle=225, radius=2];

    \draw [ultra thick] (45:2) -- (225:2);
    \draw [ultra thick] (90:2) -- (270:2);
    \draw [ultra thick] (135:2) -- (315:2);
    \draw [ultra thick] (180:2) -- (360:2);
    \draw [ultra thick] (225:2) -- (45:2);

     \pgfmathsetmacro\x{2.3*cos(22.5)}
    \pgfmathsetmacro\y{2.3*sin(22.5)}
    \node at (\x, \y) {\textbf{$Z$}};

    \pgfmathsetmacro\x{2.3*cos(200)}
    \pgfmathsetmacro\y{2.3*sin(200)}
    \node at (\x, \y) {\textbf{$X$}};

 \pgfmathsetmacro\x{1.3*cos(177)}
    \pgfmathsetmacro\y{1.3*sin(177)}
    \node at (\x, \y) {\textbf{$X$}};
 
 \pgfmathsetmacro\x{2.3*cos(180-22.5)}
    \pgfmathsetmacro\y{2.3*sin(180-22.5)}
    \node at (\x, \y) {\textbf{$X$}};
 
\end{tikzpicture}
\caption{Left and middle: Local constraints of a 2+1D topological $\mathbb{Z}_2$ gauge theory on a disk. Right: Generators of the algebra of operators ${\cal A}_\text{TFT}(B^2)$ that commute with these constraints.}\label{fig:wheel}
\end{figure}
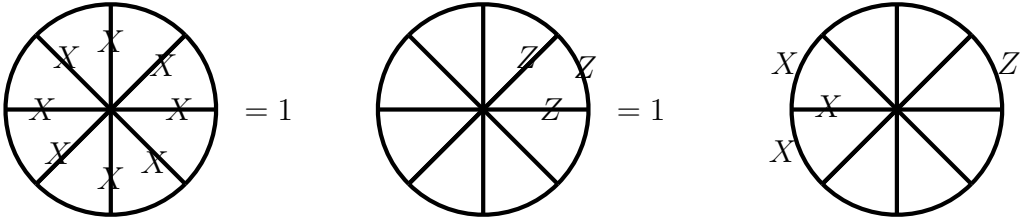

We impose the following constraint for every vertex in the bulk of the 2d spatial lattice:
 \ie\label{A}
 U_v=&\prod_{e\ni v} X_e =1 \,,~~~\forall ~e\notin \partial B^2\,.
 \fe
We do not however impose this constraint along the boundary circle, which is analogous to fact that in the continuum we only quotient by $B$ gauge transformations which vanish at the boundary (in this model the boundary charged matter lives on links, so the $B$ gauge field is analogous to the one that lives on bulk links).  We further impose another constraint for every face:
\ie\label{B}
U_f= &\prod_{e \in f} Z_e=1\,,~~~\forall ~f\,.
\fe
This can be interpreted in the continuum as the constraint for $A$ gauge transformations, and indeed we see that it can act nontrivially on boundary $X$ operators.  $U_v, U_f$ are precisely the local terms in the toric code Hamiltonian \cite{Kitaev:1997wr} (see Figure \ref{fig:wheel}).  Here we impose both of them strictly so that  the bulk theory is a topological $\mathbb{Z}_2$ gauge theory with no anyon excitations. Since we have $L+1$ constraints on this ``wheel" lattice, the  subspace in the total Hilbert space that respects these constraints is $2^{L-1}$-dimensional.

What is the algebra of operators  of this 2+1D system? 
It is easy to see that all the operators that obey both \eqref{A} and \eqref{B} are generated by $Z_e$ and triplets of $X_e$ on the boundary:
\ie\label{ATC}
{\cal A}_\text{TFT} (B^{2})&\equiv 
\Big\{
{\cal O}\in \text{Mat}(2^{2L},\mathbb{C})~\Big|~
{\cal O}U_v =U_v {\cal O} \,,~~\forall ~e\notin \partial B^2,~\text{and}~{\cal O}U_f =U_f{\cal O}\,,~~\forall ~f
\Big\}\\
&=
\Biggl\langle
~Z_e~,~ X_{e}X_{e+\frac12}X_{e+1}~
 \Biggr\rangle_{e\in \partial B^{2}}\,,
\fe 
where $e+\frac12$ stands for the vertical edge that intersects with $e$ and $e+1$. 
It is clear that ${\cal A}_\text{TFT} (B^2)$  is isomorphic to ${\cal A}_{\mathbb{Z}_2}(S^1)$ as abstract algebras:
\ie
{\cal A}_\text{TFT} (B^2) \cong {\cal A}_{\mathbb{Z}_2}(S^1)
\fe
For instance, a pair of $X_{e_1}X_{e_2}$ in ${\cal A}_{\mathbb{Z}_2}$ is mapped to a string of $X$'s trespassing the bulk in ${\cal A}_\text{TFT}$ as in Figure \ref{fig:restoration2}, which can be written as a product of the generators in \eqref{ATC}.  The analogue of this string in continuum language is a Wilson line for the $A$ gauge field.

While a pair of $X$'s violates disjoint additivity with respect to 1D regions, a string of $X$'s satisfies disjoint additivity with respect to 2D regions in the bulk. 
Indeed, as proven in section \ref{sec:lattice}, $\A_\text{TFT}(\Rbulk)$ obeys both Haag duality and disjoint additivity.  In particular the $2+1$ dimensional adjacency relation from the $U_v$ and $U_f$ constraints says that regions, i.e. collections of bulk/boundary links, which are not connected by a link are not adjacent, which is a ``short-range'' adjacency relation in the sense of section \ref{sec:lattice-DA}. We conclude that the natural spacetime dimension for this abstract algebra is 2+1D, where both Haag duality and disjoint additivity are obeyed.

\begin{figure}
\centering
\begin{tikzpicture}[scale=0.9]  
    \draw[ultra thick] (45:2) arc[start angle=45, end angle=135, radius=2];

      \draw[ultra thick] (225:2) arc[start angle=225, end angle=315, radius=2];

    \draw[ultra thick, OrangeRed] (315:2) arc[start angle=-45, end angle=45, radius=2];

    \draw[ultra thick, RoyalBlue] (135:2) arc[start angle=135, end angle=225, radius=2];

    \foreach \angle in {0, 45, 90, 135, 180, 225, 270, 315} {
    \pgfmathsetmacro\x{2*cos(\angle)}
        \pgfmathsetmacro\y{2*sin(\angle)}

        \pgfmathsetmacro\dx{0.2*cos(\angle)}  \pgfmathsetmacro\dy{0.2*sin(\angle)}
        \draw[ultra thick] (\x-\dx, \y-\dy) -- (\x+\dx, \y+\dy);
    }

    \pgfmathsetmacro\x{2.5*cos(22.5)}
    \pgfmathsetmacro\y{2.5*sin(22.5)}
    \node at (\x, \y) {\color{ForestGreen}\textbf{$X_{e_1}$}};

    \pgfmathsetmacro\x{2.5*cos(202.5)}
    \pgfmathsetmacro\y{2.5*sin(202.5)}
    \node at (\x, \y) {\color{ForestGreen}\textbf{$X_{e_2}$}};

    \node at (2.6, 0) {\footnotesize \textbf{\color{OrangeRed}$\R_1$}};
    \node at (-2.6, 0) {\footnotesize \textbf{\color{RoyalBlue}$\R_2$}};

\end{tikzpicture}~~~~~~
\begin{tikzpicture}[scale=0.9]  
    \draw[ultra thick] (45:2) arc[start angle=45, end angle=135, radius=2];

      \draw[ultra thick] (225:2) arc[start angle=225, end angle=315, radius=2];

    \draw[ultra thick] (315:2) arc[start angle=-45, end angle=45, radius=2];

    \draw[ultra thick] (135:2) arc[start angle=135, end angle=225, radius=2];

    \draw [ultra thick] (45:2) -- (225:2);
    \draw [ultra thick] (90:2) -- (270:2);
    \draw [ultra thick] (135:2) -- (315:2);
    \draw [ultra thick] (180:2) -- (360:2);
    \draw [ultra thick] (225:2) -- (45:2);

    \pgfmathsetmacro\x{2.4*cos(22.5)}
    \pgfmathsetmacro\y{2.4*sin(22.5)}
    \node at (\x, \y) {\color{ForestGreen}\textbf{$X_{e_1}$}};

    \pgfmathsetmacro\x{2.3*cos(202.5)}
    \pgfmathsetmacro\y{2.3*sin(202.5)}
    \node at (\x, \y) {\color{ForestGreen}\textbf{$X_{e_2}$}};
    
 \pgfmathsetmacro\x{1.3*cos(35)}
    \pgfmathsetmacro\y{1.3*sin(35)}
    \node at (\x, \y) {\color{ForestGreen}\textbf{$X$}};
 
 \pgfmathsetmacro\x{1.3*cos(80)}
    \pgfmathsetmacro\y{1.3*sin(80)}
    \node at (\x, \y) {\color{ForestGreen}\textbf{$X$}};

 \pgfmathsetmacro\x{1.3*cos(125)}
    \pgfmathsetmacro\y{1.3*sin(135)}
    \node at (\x, \y) {\color{ForestGreen}\textbf{$X$}};

 \pgfmathsetmacro\x{1.3*cos(170)}
    \pgfmathsetmacro\y{1.3*sin(170)}
    \node at (\x, \y) {\color{ForestGreen}\textbf{$X$}};
 
  \draw[ultra thick, ForestGreen,dashed] (45:.7) arc[start angle=45, end angle=180, radius=.7];
\draw [ultra thick, ForestGreen,dashed] (45:.7) -- (1.9,.5);
\draw [ultra thick, ForestGreen,dashed] (180:.7) -- (-1.8,-.8);

\end{tikzpicture}
\caption{Restoration of disjoint additivity in toric code. Left: Disjoint additivity is violated by a pair of Pauli $X_e$ operators in the $\mathbb{Z}_2$-even subalgebra ${\cal A}_{\mathbb{Z}_2}$ on a 1+1D lattice. Right: A string of $X$'s in the algebra ${\cal A}_\text{TFT}$ for the  2+1D topological $\mathbb{Z}_2$ gauge theory (i.e., low-energy limit of toric code) on a disk does not violate disjoint additivity.}\label{fig:restoration2}
\end{figure}
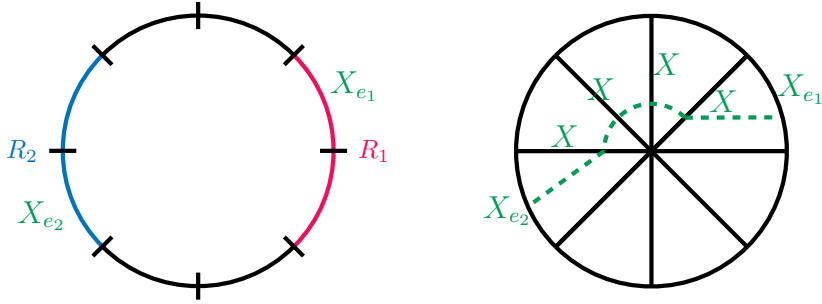

So far in this section we have discussed the invariant boundary algebras $\mathcal{A}_G(R)$ in the special case $G=\mathbb{Z}_p$.  In the continuum this allowed us to use the simple Lagrangian \eqref{ZpL}.  On the lattice however it is not hard to generalize our discussion to arbitrary finite group $G$.  To set this up however it is necessary to use a formulation which is dual to the lattice presentation we just gave in the $\mathbb{Z}_2$ case, so that the boundary degrees of freedom are on vertices rather than edges and the constraints which are not imposed at the boundary are the face constraints rather than the vertex constraints.  This is necessary because in the non-abelian case the face constraints do not form a group so we need to have the vertex constraints be the ones we impose at the boundary to get $\mathcal{A}_G(R)$.  Disjoint additivity is then restored in the same as as in the $\mathbb{Z}_2$ case: neutral pairs of charged operators on the boundary are connected through the bulk by a Wilson line, now with the Wilson line located on a connected path of edges as opposed to a path in the dual graph.

\section*{Acknowledgements}

We thank Netta Engelhardt, Nick Holfester, Corey Jones, Michael Levin, Dave Penneys, Salvatore Pace, Pedro Ribeiro, Nathan Seiberg, Sahand Seifnashri, and Xiao-Gang Wen for stimulating discussions.  
We also thank Hong Liu and Javier Magan for comments on the draft. 
The work of DH is supported by the Packard Foundation, the US Department of Energy under grants DE-SC0012567 and DE-SC0020360, and the MIT Department of Physics.  SHS is supported by the Simons Collaboration on Ultra-Quantum Matter, which is a grant from the Simons Foundation (651444, SHS). 
JS is supported by the DOE Early Career Award DE-SC0021886,  the Packard Foundation Award in Quantum Black Holes and Quantum Computation, the DOE QuantISED DE-SC0020360 contract 578218, and by the Heising-Simons Foundation grant 2023-443. 
MS is supported by the U.S. Department of Energy, Office of Science, Office of High Energy Physics of U.S. Department of Energy under grant Contract Number  DE-SC0012567 (High Energy Theory research) and by  MIT. Part of this work was completed during the Kavli Institute for Theoretical Physics (KITP) program ``Generalized Symmetries in Quantum Field Theory: High Energy Physics, Condensed Matter, and Quantum Gravity", which is supported in part by grant NSF PHY-2309135 to the  KITP. 
The authors of this paper were ordered alphabetically.

\appendix

\section{Covariant regions for algebraic field theory}
\label{app:causality}

In the main text we claimed that for a relativistic quantum field theory in a Lorentzian spacetime $M$, it is natural to assign algebras to regions $R\subseteq M$ obeying the requirement
\be
R''=R,
\ee
where $R'$ is the interior of the causal complement of $R$.  In this appendix we present some causal structure results which support this definition.

\subsection{General proposal}

We will first state some motivating principles:
\bi
\item [(1)] In relativistic field theory, fields are operator-valued distributions that must be smeared against smooth test functions of compact support.  It is natural to characterize such operators by the spacetime regions in which their supports lie, so the natural regions to consider are open spacetime regions. In particular we do not want closed regions, as then we would have to worry about what happens to the support of the function at the edge of the region.
We also do not want codimension-one spatial regions (for example closed achronal sets), because except for free fields these sets are too small for smearing to give a good operator (see e.g. \cite[section 2.2]{Witten:2023qsv}).
\item[(2)] We want the complement operation to send regions to other regions of the same type.  This is not true for the set-theoretic complement, which maps open sets to closed sets, or for the standard causal complement that we define presently which does the same, so we need a better complement that respects principle (1). 
\item[(3)] We would like our definitions to make sense in general spacetimes $M$ not obeying any particular causality requirements such as time orientability or global hyperbolicity.  For example in \cite{Harlow:2023hjb} it was recently argued that time-unorientable spacetimes must be included in quantum gravity due to the gauging of $\mathcal{CRT}$ symmetry, and global hyperbolicity can fail in portions of spacetime where quantum gravity becomes important, such as the final stage of black hole evaporation.
\ei

Now some standard definitions (see e.g. \cite{Wald:1984rg}, although we will not assume time orientability as is done there).  Let $M$ be a Lorentzian spacetime.  A continuous, piecewise-smooth curve $\gamma:(0,1)\to M$ is said to be \textit{chronal} if its tangent vector is timelike wherever it exists, and it is said to be \textit{causal} if its tangent vector is timelike or null wherever it exists.\footnote{In both cases, the tangent vector is not allowed to change its local time orientation across a non-smooth point of the curve.  Also strictly speaking we should include continuous curves obtained as limits of these piecewise-smooth ones, see e.g. the discussion on pages 192-193 of \cite{Wald:1984rg}.  This does not matter in the definitions of $I$, $D$, and $J$, but it matters whenever we invoke lemma 8.1.5 from \cite{Wald:1984rg} below.}  Given a subset $S\subseteq M$, we define $I(S)\subseteq M$ to be the set of points in $M$ that are connected to $S$ by a chronal curve, and we define $J(S)$ to be the set of points in $M$ that are connected to $S$ by a causal curve.  A set $S$ is said to be \textit{achronal} if no two distinct points in $S$ are connected by a chronal curve, and it is said to be \textit{acausal} if no two distinct points are connected by a causal curve.  A point $p\in M$ is said to be an \textit{endpoint} of a causal curve $\gamma$ if every neighborhood $U_p$ of $p$ contains either all points $\gamma(t)$ with $t<\epsilon$ or all points $\gamma(t)$ with $1-t<\epsilon$ for some $\epsilon>0$.  A causal curve $\gamma$ is said to be \textit{inextendible} if it has no endpoint.  The \textit{domain of dependence} of $S$, denoted $D[S]$, is the set of points with the property that every inextendible causal curve through that point also intersects $S$.  A closed achronal set $\Sigma$ is called a \textit{Cauchy surface} if $D(\Sigma)=M$, and $M$ is called \textit{globally hyperbolic} if it has a Cauchy surface.  If $M$ is time-oriented then we can further define $I^{\pm}(S)$ as the set of points connected to $S$ by a future/past-directed chronal curve, and $J^\pm(S)$ as the set of points connected to $S$ by a future/past-directed causal curve.

To implement principle (2) we need to define a complement operation.  The first such notion we can consider is the \textit{causal complement} (see section III.4 of \cite{Haag:book}), which for any subset $S\subseteq M$ is defined by
\be
S^c=M-J(S).
\ee
It is simple to show the inclusion 
\be\label{c1}
S\subseteq (S^c)^c,
\ee
and if this inclusion is saturated, then the subset $S$ is said to be \textit{causally complete}. It is also simple to show that $S^c$ is causally complete for any $S\subseteq M$:
\be\label{c2}
((S^c)^c)^c=S^c
\ee
Equations \eqref{c1} and \eqref{c2} are reminiscent of two properties of von Neumman algebras: the double commutant property $\A''=\A$, and the fact that $\mathcal{X}'$ is a von Neumann algebra (and thus obeys $\mathcal{X}'''=\mathcal{X}'$) for any adjoint-closed collection $\mathcal{X}$ of bounded operators on a Hilbert space.
These similarities suggest that we might have a Haag duality relation of the type $\A(S)'=\A(S^c)$ for causally complete sets $S$.  Unfortunately, however, the causal complement of an open set is always closed  (this is because if $S$ is open then $J(S)$ is open since a causal curve from $p$ to $S$ can always be deformed to a chronal curve from $p$ to $S$).  A theory of regions based on the causal complement is thus not consistent with principle (1).  

To fix this, we introduce a modified complement operation (see \cite{Bousso:2022hlz,Akers:2023fqr} for previous uses of this definition), the \textit{spatial complement}
\be
S'=\mathrm{Int}\left(S^c\right).
\ee
Restricting now to open $S$, we again have
\be\label{p1}
S\subseteq S''
\ee
and 
\be\label{p2}
S'=S'''.
\ee
Indeed to show \eqref{p1}, given $p\in S$ we want a neighborhood $U_p$ of $p$ which has no causal curves to $S'$.  But $S$ itself is a neighborhood of $p$ with this property, as it has no causal curves to $S^c$ and $S'\subseteq S^c$.  To show \eqref{p2}, from \eqref{p1} we immediately have $S'\subseteq S'''$ so we need only show that $S'''\subseteq S'$.  Given $p\in S'''$, we have a neighborhood $U_p$ such that there are no causal curves from $U_p$ to $S''$. Since $S\subseteq S''$, this means there are also no causal curves from $U_p$ to $S$, and therefore that $p\in S'$. 

The spatial complement therefore respects the von Neumann algebra properties \eqref{p1},\eqref{p2}, and by construction if $S=S''$ then $S$ must be open as required by principle (1).  Moreover a similar argument shows that if $S_1''=S_1$ and $S_2''=S_2$ then we also have $(S_1\cap S_2)''=S_1\cap S_2$, which is another true statement about von Neumann algebras.  We therefore propose taking regions to be subsets $R\subseteq M$ obeying $R''=R$.  As these arguments did not use any special features of $M$, this definition is consistent with principle (3).\footnote{In the literature, more general open sets are sometimes discussed \cite{Haag:book}, but the algebras in these sets do not always satisfy Haag duality. For example we can take $S$ to be the union of $R_1$ and $R_2$ in figure \ref{dafig}.  In a massless theory this region has $\A(S)\subsetneq \A(S'')$, since there can be right-moving null operators between the two regions that do not leave any imprint in $R_1$ or $R_2$ but are localized in $(R_1\cup R_2)''$.  Taking the commutant of both sides and using Haag duality for $S''$ (which is a valid region), we get $\A(S')\subsetneq \A(S)'$ so Haag duality is violated for $S$.  The question of when an algebra in a region with $S = S''$ can be generated by fields in a smaller open set is an interesting and subtle one. For recent discussion see \cite{Strohmaier:2023opz}.}

In what follows it will be useful to know that if $S$ is open then $S'$ has the additional nice feature of being \textit{regular}, which means that it is equal to the interior of its closure.  This is true for $S'$ because it is the interior of the closed set $S^c$.  The complement of a regular open set is always a regular closed set, which means a set which is equal to the closure of its interior.

\subsection{Relation to spatial regions in globally hyperbolic spacetimes}
We have found a definition of region that fulfills our requirements, but it is important to understand how it relates to the possibly more conventional definition in globally hyperbolic spacetimes where one takes a region $R$ to be the domain of dependence of a regular open subset $T$ of an acausal Cauchy surface $\Sigma$ for $M$.  In this subsection we will show that when $M$ is globally hyperbolic all regions of the latter type are included by our definition, but that our definition also includes more regions which should indeed be included.  

First some lemmas:\footnote{We are grateful to Pedro Ribeiro for these two lemmas, as well as a version of theorem \ref{thm1} which replaces $R''$ by $(R^c)^c$ and requires $R$ to be open, which he kindly explained on mathoverflow (\href{https://mathoverflow.net/questions/497556/is-a-causally-complete-region-always-a-domain-of-dependence-in-a-lorentzian-spac}{link here}).}
\begin{lem}\label{lem1}
Let $M$ be a time-oriented spacetime, and $S\subseteq M$ be an open set in $M$.  Then for all $p,q\in S'$ we have $J^+(p)\cap J^-(q)\subseteq S'$.  This property of $S'$ is called \textit{causal convexity}. 
\end{lem}
\begin{proof}
The steps of the following proof are shown pictorially in figure \ref{fig:app-lemma-1}.

Fix a point $r \in J^+(p) \cap J^-(q)$.
Our goal is to show that $r$ is in $S'$, i.e., we must show that there is an open neighborhood $U_r$ containing $r$ such that no point in $U_r$ is connected to $S$ by a causal curve.
To accomplish this, we use the fact that $S'$ is open to choose a point $p' \in S'$ that is in the chronological past of $p,$ and a point $q' \in S'$ that is in the chronological future of $q$.
Because $r$ and $p'$ are connected by a curve that has a causal segment and a chronal segment, they can also be connected by a curve that is purely chronal (see for example \cite[proposition 2.18]{Penrose:causality}).
Similar considerations hold for $r$ and $q'.$
So we have $r \in I^+(p') \cap I^-(q').$
This is an open set, which we will call $U_r$.
No point in $U_r$ can be connected by a causal curve to any point in $S$, since any causal curve from $S$ to $U_r$ could be extended to a causal curve from $S$ to $p'$ or $q',$ and these points are in $S'$.
Therefore we have
\begin{equation}
    r \in U_r \subseteq S^c,
\end{equation}
hence $r \in \text{Int}(S^c) = S',$ as desired.
\end{proof}

\begin{figure}
    \centering
    \begin{tikzpicture}
        \draw [dashed] (0, -1) to[out=0, in=270] ++(1, 0.5) to[out=90, in=0] ++(-1, 0.5);
        \node at (0.5, -0.5) {$S$};
        \draw [dashed] (1,-0.5) to ++(2, 2);
        \draw [dashed] (1,-0.5) to ++(2, -2);
        \node at (1.5, -0.5) {$S'$};

        \draw [RoyalBlue, fill opacity=0.3, fill=RoyalBlue] (3, -1.5) to ++(1, 1) to ++(-1, 1) to ++(-1, -1) to ++(1, -1);
        \node [circle, black, fill=RoyalBlue, minimum size=5pt, inner sep=0] at (3, -1.5) {};
        \node [circle, black, fill=RoyalBlue, minimum size=5pt, inner sep=0] at (3, 0.5) {};
        \node at (3.3, -1.5) {$p$};
        \node at (3.3, 0.5) {$q$};

        \begin{scope}[xshift=6cm]
        \draw [dashed] (0, -1) to[out=0, in=270] ++(1, 0.5) to[out=90, in=0] ++(-1, 0.5);
        \draw [dashed] (1,-0.5) to ++(2, 2);
        \draw [dashed] (1,-0.5) to ++(2, -2);

        \draw [BrickRed, fill opacity=0.3, fill=BrickRed] (3, -2) to ++(1.5, 1.5) to ++(-1.5, 1.5) to ++(-1.5, -1.5) to ++(1.5, -1.5);
           \node [circle, black, fill=BrickRed, minimum size=5pt, inner sep=0] at (3, -2) {};
            \node [circle, black, fill=BrickRed, minimum size=5pt, inner sep=0] at (3, 1) {};
            \node at (3.3, -2) {$p'$};
            \node at (3.3, 1) {$q'$};
        
        \draw [RoyalBlue, fill opacity=0.3, fill=RoyalBlue] (3, -1.5) to ++(1, 1) to ++(-1, 1) to ++(-1, -1) to ++(1, -1);
        \node [circle, black, fill=RoyalBlue, minimum size=5pt, inner sep=0] at (3, -1.5) {};
        \node [circle, black, fill=RoyalBlue, minimum size=5pt, inner sep=0] at (3, 0.5) {};
        \node at (3.3, -1.5) {$p$};
        \node at (3.3, 0.5) {$q$};
        \end{scope}
    \end{tikzpicture}
    \caption{Left: An open set $S$, its causal complement $S',$ and two points $p, q$ with the set $J^+(p) \cap J^-(q)$ shaded in.
    Right: A choice of point $q'$ in the chronological future of $q$, and a point $p'$ in the chronological past of $p,$ such that $p'$ and $q'$ are still in $S'.$
    Every point in $J^+(p) \cap J^-(q)$ lies in the open set $I^+(p') \cap I^-(q'),$ which must lie in $S'.$}
    \label{fig:app-lemma-1}
\end{figure}
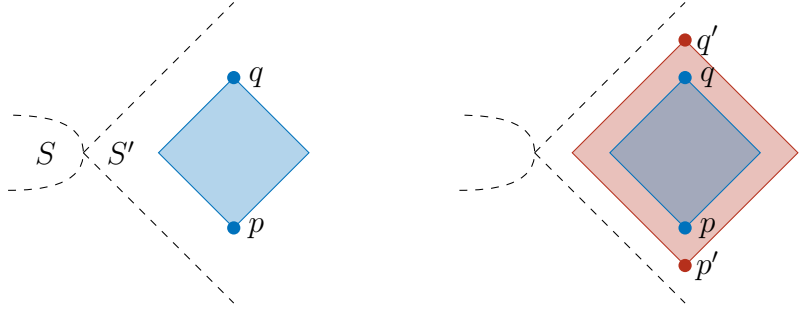
\begin{lem}\label{lem2}
Let $M$ be a globally hyperbolic spacetime, with $R\subseteq M$ a subset obeying $R''=R$.  Then $R$ is a globally hyperbolic embedded Lorentzian submanifold.  
\end{lem}
\begin{proof}
Since $R=R''$ is an interior, it is open, so it is an embedded submanifold in $M$ with a smooth Lorentzian metric.
We wish to show that this embedded submanifold admits a Cauchy surface.
Equivalently, by the remark on page 209 of \cite{Wald:1984rg}, we may show that $J^+(p) \cap J^-(q)$ is compact in the subspace topology of $R$ for every $p, q \in R$, and that $R$ is ``strongly causal'' in the sense that no causal curve comes arbitrarily close to intersecting itself.
(See \cite[page 196]{Wald:1984rg} for a precise definition.)

Since $M$ is globally hyperbolic, it is strongly causal, and this property is inherited by $R$.
Since $M$ is globally hyperbolic, $J^+(p) \cap J^-(q)$ is compact for any $p, q \in M$.
But since the preceding lemma shows that this set is contained in $R$, it is also compact for the subspace topology of $R$.
It follows that $R$ is globally hyperbolic, as desired.
\end{proof}
\noindent We then have the following theorem:
\begin{thm}\label{thm1}
Let $M$ be a globally hyperbolic spacetime and $R\subseteq M$ a subset obeying $R''=R$.  Then there exists an acausal set $T\subseteq M$ such that $D(T)=R$.
\end{thm}
\begin{proof}
By lemma \ref{lem2} we know that $R$ is globally hyperbolic, so it possesses a Cauchy surface: a closed achronal set $T$ such that every inextendible causal curve in $R$ intersects $T$. Thus within $R$ we have $D_R(T)=R$.  Moreover the construction of $T$ described on page 209 of \cite{Wald:1984rg} actually ensures that $T$ is acausal and not just achronal.  $T$ is acausal in $M$ since it is acausal in $R$ and any causal curve which exited $R$ and re-entered would violate the causal convexity of $R$.
We will now show that in $M$ we also have $D(T)=R$.  Clearly we have $R = D_R(T) \subseteq D(T)$, since if a point $p$ has the property that every inextendible causal curve passing through $p$ in $R$ must intersect $T$, then this property also holds for every inextendible causal curve passing through $p$ in $M$, since all such curves restrict to inextendible curves in $R$.  

For the other direction, we must show $D(T) \subseteq R$.
First we will show that the interior of $D(T)$ is contained in $R''$.
To see this, fix a $p\in \mathrm{Int}(D(T))$.
This means $p$ has a neighborhood $U_p$ which is contained in $D(T)$.  To show that $p$ is in $R$, we will show that it is in $R'',$ for which it suffices to show that there are no causal curves from $U_p$ to $R'$.
Since $U_p$ is in $D(T)$, any such causal curve could be extended to a causal curve connecting $R'$ to $T$.
But this contradicts the assumption $T \subseteq R$, establishing the inclusion
\begin{equation}
 \text{Int}(D(T)) \subseteq R''.
\end{equation}
Since $R''=R$, we will be done if we can show that $\text{Int}(D(T)) = D(T)$, i.e. that $D(T)$ is open.

\begin{figure}
    \centering
    \begin{tikzpicture}
        \draw [thick] (0, 0) to ++(2,2);
        \draw [thick] (0,0) to ++(2, -2);
        \node at (2.3, 0) {$T$};
        \draw (0, 0) to[out=10, in=170] ++(1, 0) to[out=-10, in=190] ++(1, 0);
        \node [circle, inner sep=0, minimum size=5pt, fill=RoyalBlue] at (1, 1) {};
        \node at (1.2, 0.6) {$p$};
        \draw [dashed] (1,1) circle (0.9);
        \draw [thick, ForestGreen] (0.7, 1.3) to ++(-2.5, -2.5);
        \node [circle, inner sep=0, minimum size=5pt, fill=ForestGreen] at (0.7, 1.3) {};
        \node at (1, 1.5) {$p_0$};
        
        \begin{scope}[xshift=7cm]
        \draw [thick] (0, 0) to ++(2,2);
        \draw [thick] (0,0) to ++(2, -2);
        \draw (0, 0) to[out=10, in=170] ++(1, 0) to[out=-10, in=190] ++(1, 0);

        \draw [thick, dashed, ForestGreen] (0.7, 1.3) to ++(-2.5, -2.5);
        \node [circle, inner sep=0, minimum size=5pt, fill=ForestGreen] at (0.7, 1.3) {};
        \draw [thick,dashed, ForestGreen] (0.8, 1.2) to ++(-2.5, -2.5);
        \node [circle, inner sep=0, minimum size=5pt, fill=ForestGreen] at (0.8, 1.2) {};
        \draw [thick,dashed, ForestGreen] (0.9, 1.1) to ++(-2.5, -2.5);
        \node [circle, inner sep=0, minimum size=5pt, fill=ForestGreen] at (0.9, 1.1) {};
        \node at (0.7, 1.6) {$p_n$};
        \node at (0.7-2.8, 1.6-3.2) {$\gamma_n$};

        \draw [ultra thick, RoyalBlue] (1,1) to ++(-2.5, -2.5);
        \node [circle, inner sep=0, minimum size=5pt, fill=RoyalBlue] at (1, 1) {};
        \node at (1.2, 0.6) {$p$};
        \node at (-0.4, -1) {$\gamma$};
        \end{scope}
    \end{tikzpicture}
    \caption{Left: A point $p$ that is assumed to lie on the boundary $D(T)$, together with an open set that contains points outside of $D(T)$. Each of these points must contain an inextendible causal curve that does not intersect $T$.
    Right: By taking smaller and smaller open sets around $p$, one can construct a sequence $p_n$ of points approaching $p$ with causal curves $\gamma_n$ that approach some inextendible causal curve $\gamma$ passing through $p.$
    But any such curve would have to intersect $T$, which is in $R$, but this is a contradiction since $R$ is open and the causal curves $\gamma_n$ never enter $R$.}
    \label{fig:app-theorem-1}
\end{figure}
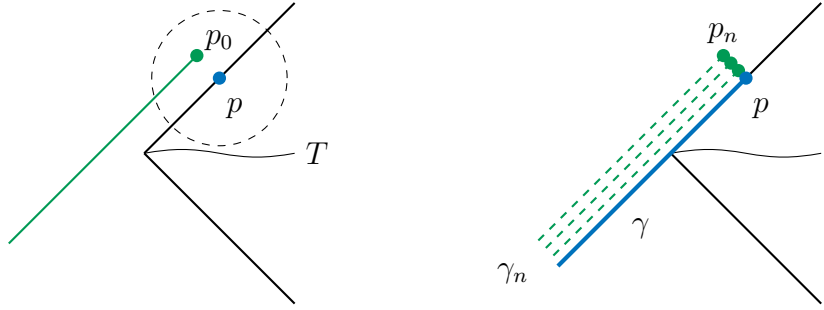

Suppose, toward contradiction, that there were a point $p\in D(T)-\mathrm{Int}(D(T))$.  For any neighborhood $U_p$ of $p$ there would be a point $p_0\in U_p$ which is not in $D(T)$; see figure \ref{fig:app-theorem-1}.
We could therefore construct a sequence $\{p_n\}$ of points converging to $p$ each of which is not in $D(T)$, and therefore a sequence $\{\gamma_n\}$ of inextendible causal curves that pass through $p_n$ but not $T$.
Since these curves do not intersect $T$, they can never enter $R$, since every inextendible causal curve in $R$ intersects $T$.
By lemma 8.1.5 of \cite{Wald:1984rg} there would then be a limiting inextendible causal curve $\gamma$ through $p$ with the property that each point $r$ on $\gamma$ is a limit of a sequence of points $r_{n_k}$ on a subsequence $\gamma_{n_k}$.  Since $p\in D(T)$, this would mean that $\gamma$ must intersect $T$, and therefore also $R$.  But this is impossible since there cannot be a sequence of points which are not in $R$ that converges to a point in $R$, since $R$ is open.  Therefore $D(T)=\mathrm{Int}(D(T))$, and $D(T)$ is open. 
\end{proof}
This theorem shows that in globally hyperbolic spacetimes, our definition of a region is closely related to the idea of the domain of dependence of a regular open subset of an acausal Cauchy surface in $M$.  We can motivate this further by giving the following converse result:\footnote{This theorem is inspired by a close cousin due to Netta Engelhardt, which shows that if $T$ is a closed achronal set that lies between two Cauchy surfaces then $D(T)$ is causally complete.  This will soon appear in \cite{EngLiu}.}

\begin{thm}\label{thm2}
Let $M$ be a globally hyperbolic spacetime with an acausal Cauchy surface $\Sigma$, and let $T\subseteq \Sigma$ be a regular open set in $\Sigma$.\footnote{We emphasize that it is essential here that $\Sigma$ is acausal; if it were merely achronal then there are counterexamples.  A simple counterexample is to take $T$ to be two half-infinite open spatial intervals in Minkowski space whose endpoints are connected by a null curve.  Then $D(T)'=\varnothing$ so $D(T)''$ is all of $M$.  It is also important for $T$ to be regular, otherwise by removing points in $T$ we can change $D(T)$ without changing $D(T)''$.}  Then $D(T)''=D(T)$.
\end{thm}
The proof of this theorem relies on the following technical lemma:
\begin{lem}\label{lem3}
    Let $M$ be a globally hyperbolic spacetime with Cauchy surface $\Sigma,$ and let $T \subseteq \Sigma$ be acausal in $M$ and open in $\Sigma$.
    Then $D(T)$ is open.
\end{lem}
\begin{proof}
This proof is adapted from \cite[lemma 14.43]{ONeill:book}.

Suppose, toward contradiction, that there were a point $p$ in the boundary of $D(T)$.
Then for any open set $U_p$ containing $p,$ there would be points in $U_p$ that do not lie in $D(T)$.
As in the proof of theorem \ref{thm1} --- see again figure \ref{fig:app-theorem-1} --- we could use this to construct a sequence of inextendible causal curves $\gamma_n$ that never intersect $T$ and that converge to a causal curve $\gamma$ passing through $p.$
Since $p$ is in $D(T)$ and $T$ is acausal, the curve $\gamma$ must intersect $T$ at a unique point, which we call $q.$
Since $\gamma_n$ converges to $\gamma,$ there is a subsequence of points $q_{n_k}$ on the curves $\gamma_{n_k}$ that converge to $q.$
Each of the points $q_{n_k}$ is outside of $D(T)$.
So to find a contradiction, we need only find an open set $V$ with $q \in V \subseteq D(T)$; the sequence $q_{n_k}$ would lie outside of $V$, but would converge to a point in $V$, forming a contradiction.

We have therefore reduced our problem to proving the claim $T \subseteq \text{Int}(D(T)).$
Again we will proceed by contradiction.
Suppose there were a point $q \in T$ not in the interior of $D(T)$.
We will localize the problem as much as possible by performing the rest of our analysis in a \textit{convex normal neighborhood}.
This is an open set in which (i) any two points are connected by a unique geodesic, and (ii) two points are causally (chronologically) separated if and only if the unique geodesic connecting them is causal (timelike).
In any Lorentzian spacetime every point is contained in such a neighborhood --- see for example theorem 8.12 of \cite{Wald:1984rg} or proposition 1.13 of \cite{Penrose:causality}.
We then choose an open set $N$ containing $q \in T$ with the properties that (i) $\bar{N}$ is compact, (ii) $\bar{N}$ is contained in $I^+(T) \cup T \cup I^-(T)$, and (iii) $\bar{N}$ is contained in a convex normal neighborhood.  For example we can take $N$ to be an open ball around $q$ in some Riemannian metric on $M$  which is contained in the (open) intersection of $I^+(T) \cup T \cup I^-(T)$ and a convex normal neighborhood of $q$.

\bfig
\includegraphics[height=4cm]{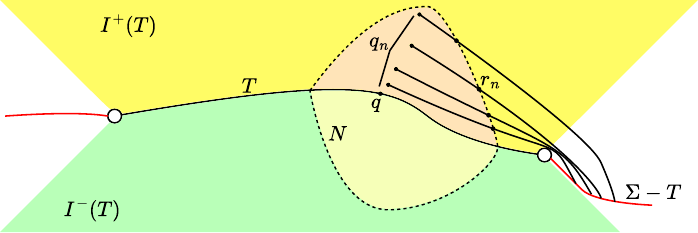}
\caption{Illustration of the proof of lemma \ref{lem3}.  The point $q\in T$ is hypothesized to be in the boundary of $D(T)$.  The points $q_n$ which are not in $D(T)$ approach $q$ from the future, and the points $r_n$ where causal curves from $q_n$ to $\Sigma-T$ exit $N$ have a limit point $r$ which cannot be in $T$ or $I^{\pm}(T)$. This is a contradiction since $r$ must be in $\overline{N}\subseteq I^+(T)\cup I^-(T)\cup T$.}\label{openfig}
\efig
Since we have assumed $q$ is not in the interior of $D(T)$, we can proceed as in the first paragraph of this proof to construct a sequence of points $q_n \in N$ converging to $q,$ lying on inextendible causal curves $\zeta_n$ that never intersect $T$, but that converges to an inextendible causal curve $\zeta$ that passes through $q.$
Without loss of generality, we will take the points $q_n$ to be in the future of $q,$ and we will take the curves $\zeta_n$ to be past-directed, see figure \ref{openfig} for an illustration.

Since $\bar{N}$ is compact, each of the curves $\zeta_n$ must eventually exit $\bar{N}$; call the point where this happens $r_n.$
Since $\bar{N}$ is compact, there is a convergent subsequence $r_{n_k} \to r$ with $r\in \bar{N}$.
Since the sequence $q_{n_k}$ converges to $q,$ and since the sequence $r_{n_k}$ converges to $r,$ and since everything is taking place within a convex normal neighborhood, proposition 1.11 of \cite{Penrose:causality} guarantees that the geodesic segment connecting $q_{n_k}$ to $r_{n_k}$ converges to the geodesic segment connecting $q$ to $r.$
Each of the geodesic segments in this sequence is causal, which means that $q$ and $r$ lie on a shared causal curve.
By our assumption that the curves $\zeta_n$ were past-directed, this gives us $r \in J^-(q).$
Since $T$ is acausal, we learn $r \notin T.$
We also learn that $r$ is not in $I^+(T),$ since then we would have a causal curve from $q$ to $r$ to a point in $T$, contradicting acausality of $T$.
Finally, $r$ is not in $I^-(T),$ since the points $r_{n}$ are all in the chronological future of $T$ --- the points $q_n$ are in $I^+(T)$, and the causal curves connecting $q_n$ to $r_n$ all lie in $I^+(T) \cup T \cup I^-(T)$, but they never intersect $T$, so they cannot leave $I^+(T)$ before reaching $r_n.$
So we have learned that $r$ is not in $I^+(T) \cup T \cup I^-(T)$, which means it should not be in $\bar{N}$, but this contradicts that we above found $r\in \bar{N}$.
\end{proof}
We now can prove theorem \ref{thm2}:
\begin{proof}
By \eqref{p2} it is enough to show that $D(T)$ is the spatial complement of an open set.  We will show that
\be
D(T)=D(\mathrm{Int}(\Sigma-T))',
\ee
where $\mathrm{Int}(\Sigma-T)$ means the interior within $\Sigma$ and $D(\mathrm{Int}(\Sigma-T))$ is open by lemma \ref{lem3}.  We will first show that $D(\mathrm{Int}(\Sigma-T))'\subseteq D(T)$.  Indeed say $p\in D(\mathrm{Int}(\Sigma-T))'$.  This means there is a neighborhood $U_p$ such that there are no causal curves from $U_p$ to $D(\mathrm{Int}(\Sigma-T))$.  Suppose that $p$ is not in $D(T)$.  Then by global hyperbolicity there would be a causal curve from $p$ to $\Sigma-T$.  Choosing a point $p'\in I(p)\cup U_p$ we could deform this to a chronal curve from $p'$ to $\Sigma-T$.  Since $T$ is a regular open set $\Sigma-T$ is a regular closed set, meaning that $\Sigma-T=\ol{\mathrm{Int}(\Sigma-T)}$.  We can therefore further deform this chronal curve to a chronal curve from $p'$ to $\mathrm{Int}(\Sigma-T)$.\footnote{To spell out the last step, assume we have a chronal curve from $p'$ to a point $q \in \Sigma - T$ that is not already in the interior of $\Sigma - T$.  Within a convex normal neighborhood of $q$, one can choose a point $\hat{p}$ on the chronal curve which necessarily has the geodesic square distance $d^2(\hat{p}, q) < 0.$ Since $\Sigma - T$ is the closure of $\mathrm{Int}(\Sigma - T)$, there are points in $\mathrm{Int}(\Sigma - T)$ in arbitrarily small neighborhoods around $q$; since geodesic squared distances are continuous, one can deform $q$ slightly into one of these neighborhoods while remaining timelike to $\hat{p},$ hence remaining timelike to $p'.$}
This contradicts no causal curves from $U_p$ to $D(\mathrm{Int}(\Sigma-T))$.  Conversely we will show $D(T)\subseteq D(\mathrm{Int}(\Sigma-T))'$.  Indeed for any $p\in D(T)$, since $D(T)$ is open there is a neighborhood $U_p$ such that every inextendible causal curve through $U_p$ intersects $T$.  In fact we can just take $U_p$ to be $D(T)$. Suppose now that $p$ were not in $D(\mathrm{Int}(\Sigma-T))'$.  Then for every neighborhood $U_p$ there would be a causal curve from $U_p$ to $D(\mathrm{Int}(\Sigma-T))$.  In particular this would be true for $U_p=D(T)$.  This curve however would be extendible to a causal curve from $T$ to  $\mathrm{Int}(\Sigma-T)$, which contradicts that $\Sigma$ is acausal.
\end{proof}

\bfig
\includegraphics[height=5cm]{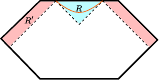}
\caption{An example of a region $R$ (shaded blue) in the Schwarzschild geometry that obeys $R=R''$ but is not the domain of dependence of an open subset of a Cauchy surface.  An achronal set $T$ of the type promised by theorem \eqref{thm1} is shaded orange, but it is not a subset of any Cauchy surface.  Nonetheless we expect the algebra of operators in $R$ to make sense for quantum field theory in this background (for example in the Hilbert space built on the Hartle-Hawking state), and for this algebra to obey Haag duality and disjoint additivity.}\label{schwarzfig}
\efig
Thus every region obtained as the domain of dependence of a regular open subset of an acausal Cauchy surface is a region also by our criterion.  On the other hand, although theorem \ref{thm1} promises that $R=D(T)$ for some acausal $T$, it does \textit{not} promise that $T$ is a regular open subset of a Cauchy surface, acausal or otherwise.  In fact there are globally hyperbolic spacetimes with regions $R$ obeying $R=R''$ where this is not true, see figure \ref{schwarzfig} for an example in the maximally extended Schwarzschild geometry.  Our attitude is that algebras of operators assigned to such regions make sense and should still be expected to obey Haag duality and disjoint additivity, so we include them in our definition.  On the other hand if we \textit{do} have $T=R\cap \Sigma$ for some Cauchy surface $\Sigma$, then $T$ is indeed a regular open subset of $\Sigma$.  That it is open is obvious since $R=R''$ is open, to see that it is regular we note that the intersection of any regular open set $R$ with a closed set $\Sigma$ is always a regular open set of $\Sigma$ in the subspace topology (this follows from writing out the definitions and using that the complement of $R$ is a regular closed set).  

\subsection{Spatially disjoint regions in globally hyperbolic spacetimes}
\bfig
\includegraphics[height=4cm]{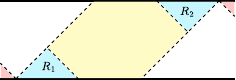}
\caption{An example of a globally hyperbolic spacetime with two spatially disjoint regions $R_1$ and $R_2$ such that $R_1\cup R_2$ is not a region.  The horizontal solid black lines are future/past spacelike singularities (for example this could be part of an elongated wormhole), $R_1$ and $R_2$ are the indicated blue regions, $(R_1\cup R_2)'$ is shaded red, and the yellow points are in $(R_1\cup R_2)''$ but not in $R_1\cup R_2$.  We could try modifying equation \eqref{dadef} to have $\A((R_1\cup R_2)'')$ on the left-hand side to account for such examples, but this would be wrong since e.g. there could be left-moving operators in the yellow region which leave no imprint in $R_1$ or $R_2$.}\label{examplefig}
\efig
In our definition of disjoint additivity for relativistic theories, in addition to the regions $R_1$ and $R_2$ being spatially disjoint (which recall we defined by \eqref{spdis}), we also required that $(R_1\cup R_2)''=R_1\cup R_2$.  This latter condition is automatic in simple examples, but in general it is necessary to impose by hand.  Indeed figure \ref{examplefig} shows an example of two spatially disjoint regions in a globally hyperbolic spacetime whose union is not a region.  On the other hand such regions are rather unusual, and we will now show that this situation cannot arise for spatially disjoint regions that are each obtained as the domain of dependence of a regular open subset of a Cauchy surface:
\begin{thm}
Let $M$ be a globally hyperbolic spacetime, and $T_{1},T_{2}$ be regular open subsets of acausal Cauchy surfaces $\Sigma_1,\Sigma_2$ respectively such that $D(T_1)$ and $D(T_2)$ are spatially disjoint.  Then 
\be\label{thm3res}
\left(D(T_1)\cup D(T_2)\right)''=D(T_1)\cup D(T_2).
\ee
\end{thm}
\begin{proof}
The proof has two parts.  We first will show that
\be\label{cupD}
D(T_1)\cup D(T_2)=D(T_1\cup T_2),
\ee
and then that there is an acausal Cauchy surface $\Sigma$ that contains both $T_1$ and $T_2$ as regular open subsets.  
Moreover the union of two regular open sets with disjoint closures is regular, so \eqref{thm3res} then follows directly from theorem \ref{thm2}.\footnote{We note that readers who are willing to assume that $\Sigma_1=\Sigma_2$ only need to read the first part of the proof.}  

It is automatic that $D(T_1)\cup D(T_2)\subseteq D(T_1\cup T_2)$, so we need to show the opposite inclusion.  Say that $p\in D(T_1\cup T_2)$.  There must be a causal curve from $p$ to $T_1$ or $T_2$, so without loss of generality we consider a curve that goes from $p$ to $T_1$.  We further take $p$ to be in the future of $\Sigma_1$, again without loss of generality (if $p$ lies in $\Sigma_1$ then it must lie in $T_1$ so it is already in $D(T_1)$).   We now argue that $J^-(p)\cap \Sigma_1$ is a connected subset of $\Sigma_1$.  We will do this by showing that it is homeomorphic to $\dot{J}^-(p)\cap J^+(\Sigma_1)$, where the dot indicates taking the topological boundary.  This latter space is connected since every point in $\dot{J}^-(p)$ lies on a past-directed null geodesic from $p$ that lies entirely in $\dot{J}^-(p)$.\footnote{This follows from the time-reverse of theorem 8.1.6 in \cite{Wald:1984rg}, together with global hyperbolicity since if the past-directed future-inextendible null geodesic through a point $q\in \dot{J}^-(p)$ promised by theorem 8.1.6 didn't come from $p$ then it would not reach any Cauchy surface which lies to the future of $p$.}  Moreover by lemma 8.1.1 of \cite{Wald:1984rg} there exists a past-pointing timelike vector field on $M$.  Each integral curve of this vector field is an inextendible chronal curve, and each of these integral curves which intersects $J^-(p)\cap \Sigma_1$ must also intersect $\dot{J}^-(p)\cap J^+(\Sigma_1)$ since otherwise it would not be able to get to a Cauchy surface which is to the future of $p$.  Similarly any integral curve which intersects $\dot{J}^-(p)\cap J^+(\Sigma_1)$ must also intersect $J^-(p)\cap \Sigma_1$, and none of these intersections can happen more than once since both surfaces are achronal.  Thus we can define a homeomorphism $f:\dot{J}^-(p)\cap J^+(\Sigma_1)\to J^-(p)\cap \Sigma_1$ using these integral curves (see figure \ref{connfig}).  Thus $J^-(p)\cap \Sigma_1$ is connected. Next, we note that $J(T_2)$ is an open set since it is equal to $J(D(T_2))$ and $D(T_2)$ is open by lemma \ref{lem3}, so $J(T_2)\cap \Sigma_1$ is an open subset of $\Sigma_1$.  Since $p\in D(T_1\cup T_2)$, it must be the case that 
\be
J^-(p)\cap \Sigma_1\subseteq T_1 \cup \big(J(T_2)\cap \Sigma_1\big).
\ee
Moreover $T_1$ and $J(T_2)\cap \Sigma_1$ must be disjoint, since $D(T_1)$ and $D(T_2)$ are spatially disjoint. Since $J^-(p)\cap \Sigma_1$ is connected, it therefore must lie entirely either in $T_1$ or in $J(T_2)\cap \Sigma_1$, and since it has at least one point in $T_1$ it must lie in $T_1$.  Therefore $p$ is in $D(T_1)$.  Had we instead assumed that the causal curve from $p$ went to $T_2$ we would have concluded $p\in D(T_2)$, so either way we have $p\in D(T_1)\cup D(T_2)$ and \eqref{cupD} is established.   

\bfig
\includegraphics[height=4cm]{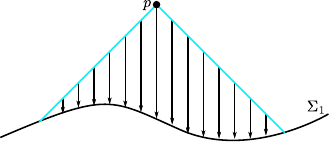}
\caption{Using the integral curves of a a timelike vector field to define a homeomorphism between $\dot{J}^-(p)\cap J^+(\Sigma_1)$ (shaded blue) and $J^-(p)\cap \Sigma_1$.}\label{connfig}
\efig
We now turn to constructing an acausal Cauchy surface containing $T_1$ and $T_2$.\footnote{This construction is inspired by the proof of lemma 14 in \cite{Headrick:2014cta}.}  As a warmup, we first note that if we define
\be
Y_0=I^+(\Sigma_1)\cap I^+(\Sigma_2),
\ee
then $\Sigma_0 \equiv \dot{Y_0}$ is an acausal Cauchy surface (here the dot again indicates taking the topological boundary). Indeed let $t_1$, $t_2$ be time functions such that $\Sigma_1$, $\Sigma_2$ are the slices $t_1=0$, $t_2=0$ respectively.  By definition a time function is a continuous function that is strictly increasing along any causal curve, and in any globally hyperbolic spacetime a time function exists such that each slice of constant $t$ is an acausal Cauchy surface (see theorem 8.3.14 of \cite{Wald:1984rg}). Moreover any acausal Cauchy surface is a constant $t$ slice for some time function (see theorem 5.15 of \cite{Bernal:2005qf}).  The surface $\Sigma_0$ can be viewed as the $u=0$ slice for the time function $u=\min(t_1,t_2)$, so it is necessarily acausal.  Moreover any inextendible causal curve must cross $\Sigma_0$, since it must begin to the past of both $\Sigma_1$ and $\Sigma_2$ and then later enter $Y_0$.   

\bfig
\includegraphics[height=5cm]{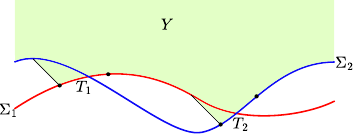}
\caption{Constructing an achronal Cauchy surface containing $T_1$ and $T_2$.  $\Sigma_1$ is shaded red, $\Sigma_2$ is shaded blue, $Y$ is shaded green, $T_1$ is the region between the two black dots on $\Sigma_1$, and $T_2$ is the region between the two black dots on $\Sigma_2$.  The Cauchy surface is the boundary of $Y$.}\label{cauchyfig}
\efig
This Cauchy surface however does not necessarily contain $T_1$ and $T_2$, so to get one that does we modify the definition to
\begin{align}\nonumber
Y&=Y_0\cup I^+(T_1)\cup I^+(T_2)\\
\Sigma&=\dot{Y}.
\end{align}
See figure \ref{cauchyfig} for an illustration of this construction.  We now argue that $\Sigma$ is an achronal Cauchy surface containing $T_1$ and $T_2$; we will modify it again later to make it acausal.  We can write $\Sigma$ as
\begin{align}\nonumber
\Sigma=&\big(\dot{Y}_0-I^+(T_1)\cap \dot{Y}_0-I^+(T_2)\cap \dot{Y}_0\big)\\\nonumber
&\cup \big(\dot{I}^+(T_1)-Y_0\cap \dot{I}^+(T_1)-I^+(T_2)\cap\dot{I}^+(T_1)\big)\\
&\cup\big(\dot{I}^+(T_2)-Y_0\cap \dot{I}^+(T_2)-I^+(T_1)\cap\dot{I}^+(T_2)\big),\label{Sigunion}
\end{align}
with the first factor in the union being acausal since $\dot{Y}_0$ is and the second two being achronal since the boundary of the future of any set is achronal (see theorem 8.1.3 of \cite{Wald:1984rg}).  Any future-directed chronal curve from $\dot{I}^+(T_1)$ or $\dot{I}^+(T_2)$ must immediately enter $I^+(T_1)$ or $I^+(T_2)$ respectively and not leave, and thus must enter $Y$ and not leave.  Therefore it cannot again intersect $\Sigma$.  Similarly any future-directed chronal (and in fact causal) curve from $\dot{Y}_0$ must immediately enter $Y_0$ and not leave, and therefore cannot again intersect $\Sigma$.  Therefore $\Sigma$ is achronal.  Moreover every inextendible causal curve must begin to the past of $\Sigma_1$ and $\Sigma_2$, and eventually enter $Y_0$, so it must cross $\Sigma$.  Therefore $\Sigma$ is an achronal Cauchy surface. Finally to see that $T_1$ and $T_2$ are contained in $\Sigma$, we note that $T_1$ is contained in $\dot{I}^+(T_1)$ but not in $Y_0$ or $I^+(T_2)$, the former because $Y_0$ does not intersect $\Sigma_1$ and the latter because $T_1$ and $T_2$ are spatially disjoint.  Therefore $T_1$ is contained in the middle factor of the union in \eqref{Sigunion}, and thus in $\Sigma$.  The same argument shows that $T_2$ is also contained in $\Sigma$.  

Finally we need to show how to deform $\Sigma$ to make it acausal while still containing $T_1$ and $T_2$.  The idea is to choose a timelike vector field $t^\mu$ on $M$ (which always exists, for example by raising the index of the gradient of a time function), and then define a new metric
\be
\wt{g}_{\mu\nu}(x)=g_{\mu\nu}(x)-\lambda(x)t_\mu(x) t_\nu(x)
\ee
with $\lambda(x)>0$.  This ``opens up the lightcones'', so a causal curve in the metric $g$ is a chronal curve in the metric $\wt{g}$ and a null curve in the metric $\wt{g}$ is spacelike in the metric $g$.  We emphasize that in general we will need to consider position-dependent $\lambda$, as a constant $\lambda$ only works if $\Sigma_1$ and $\Sigma_2$ are compact.  We then define
\begin{align}\nonumber
\wt{Y}&=Y_0\cup \wt{I}^+(T_1)\cup \wt{I}^+(T_2)\\
\wt{\Sigma}&=\dot{\wt{Y}}.
\end{align}
We now argue that for sufficiently small $\lambda$, $\wt{\Sigma}$ is an acausal Cauchy surface in the old metric $g$.  We first note that since $\dot{\wt{I}}^+(T_1)$ and $\dot{\wt{I}}^+(T_2)$ are achronal in the new metric, they are acausal in the old metric.  Moreover any future-directed causal curve in the old metric which starts in either of them is chronal in the new metric, and hence enters $\wt{Y}$ immediately and cannot leave.  A future-directed causal curve in the old metric which starts in $\dot{Y}_0$ stays in $Y_0$ just as before, and therefore also in $\wt{Y}$, so $\wt{\Sigma}$ is acausal in the old metric.  Moreover any inextendible causal curve $\gamma$ in the old metric must start out in the past of both $\Sigma_1$ and $\Sigma_2$, and we can choose $\lambda$ to be small enough that $\wt{I}^+(T_1)$ does not intersect $I^-(\Sigma_1)$ and $\wt{I}^+(T_2)$ does not intersect $I^-(\Sigma_2)$.\footnote{To see that such a $\lambda$ exists, note that since $\Sigma_1$ is acausal, at each point in $\Sigma_1$ we can choose $\lambda$ to be small enough that the local widened lightcone in a convex normal neighborhood does not intersect $\Sigma_1$.  This ensures that no future-directed causal curve in the new metric can cross $\Sigma_1$ more than once, so $\Sigma_1$ remains acausal in the new metric.  We can then make the same choice for $\Sigma_2$, taking the minimum of the two $\lambda$s for any points in common and then smoothing as needed.}  Therefore $\gamma$ must start out not in $\wt{Y}$.  On the other hand it must end up in the future of $\Sigma_1$ and $\Sigma_2$ since they are Cauchy surfaces in the old metric, and therefore $\gamma$ must cross into $\wt{Y}$.  Thus $\wt{\Sigma}$ is an acausal Cauchy surface in the old metric.  

The last thing we need to show is that $\wt{\Sigma}$ contains $T_1$ and $T_2$.  This will only work for sufficiently small $\lambda$, since if $T_1\cap \wt{I}^+(T_2)$ or $T_2 \cap \wt{I}^+(T_1)$ ends up being non-empty due to the lightcone opening in $\wt{g}$ then $T_1$ or $T_2$ will not be in $\wt{\Sigma}$. The idea is to choose $\lambda$ on each null generator of $\dot{I}^+(T_1)$ and $\dot{I}^+(T_2)$ to be small enough that the new generators in $\dot{\wt{I}}^+(T_1)$ and $\dot{\wt{I}}^+(T_2)$ are close enough to those in $\dot{I}^+(T_1)$ and $\dot{I}^+(T_2)$ to still avoid $T_2$ and $T_1$ respectively.  For each generator we can do this because $T_1$ and $T_2$ are spatially disjoint, so the distance in $\Sigma_1$ from $\ol{T}_1$ to each generator of $\dot{I}^+(T_2)$ is greater than zero.\footnote{This is because in a metric space the distance between a point $p$ and a non-empty closed set $C$ which does not contain $p$ is greater than zero.  Indeed since $C$ is closed, there must be a ball of radius $\epsilon>0$ centered at $p$ which does not intersect $C$. In our case $C$ is $\ol{T}_1$, $p$ is the intersection of a generator of $\dot{I}^+(T_2)$ with $\Sigma_1$, and the topology on $\Sigma_1$ is a metric space induced by an arbitrary choice of Riemannian metric on $M$. Here $\overline{T_1}$ does not contain $p$ because $D(T_1)$ and $D(T_2)$ are spatially disjoint.}
\end{proof}

\section{A general lattice proof of disjoint additivity}
\label{app:general-DA}

In section \ref{sec:lattice-DA}, we considered a general lattice system carrying the action of a compact Lie group $\S$, defined what it means for two regions $R_1$ and $R_2$ to be ``non-adjacent,'' and showed that disjoint additivity is satisfied for non-adjacent regions in the constrained Hilbert space $\H_S$.
To recapitulate, we first assumed that each $U(s)$ factorizes between complementary regions, $U(s) = U_{R}(s) \otimes U_{R'}(s)$.
We then defined $\N$ to be the subgroup of $\S$ that acts trivially on both $R_1$ and $R_2,$ and assumed the existence of Lie subgroups $\S_1, \S_2 \subseteq \S$ that mutually commute, that can be added to $\N$ to generate all of $\S,$ and such that $\S_1$ acts trivially on $R_2$ and vice versa.
In section \ref{sec:lattice-DA}, we made the additional simplifying assumption $\S = \S_1 \times \S_2 \times \N,$ but we claimed this was not necessary.
Here we give the more general proof.
Note that above and below, when we say that an operator $O$ acts trivially on the region $R$, what we mean is $O \subseteq \A(R)'.$

In section \ref{sec:lattice-DA}, we showed that it suffices to prove the identity $(O_1 O_2)_\S = (O_1)_\S (O_2)_\S$ for $O_j \in \A(R_j)$ and $(O_j)_\S$ defined as in equation \eqref{eq:aS}.
The left-hand side of this equation is
\begin{equation} \label{eq:app-prefactorized-integral}
    (O_1 O_2)_\S
        = \int ds\, U(s) O_1 O_2 U(s)^{\dagger}.
\end{equation}
The integrand does not change if we right-multiply $s$ by an element $x \in \N,$ since $x$ is assumed to act trivially on both $R_1$ and $R_2.$
It is also true that $\N$ is a normal subgroup of $\S,$ since for $s \in \S$ and $x \in \N$ we have
\begin{align}
    \begin{split}
    U(s x s^{-1}) O_j U(s x s^{-1})^{\dagger}
        & = U(s) U(x) \left[ U_{R_j}(s^{-1}) O_j U_{R_j}(s)\right] U(x^{-1}) U(s^{-1}) \\
        & = U(s) \left[ U_{R_j}(s^{-1}) O_j U_{R_j}(s) \right] U(s^{-1}) \\
        & = O_j.
    \end{split}
\end{align}
Moreover, $\N$ is topologically closed within $\S$ by its definition as an intersection,\footnote{We have $U(\N) = U(\S) \cap \A(R_1 \cup \R_2)'$; both sets on the right-hand side are closed in the strong topology, and the representation $U$ is a homeomorphism onto its image with respect to the strong topology.} so by the closed subgroup theorem it is a Lie subgroup of $\S.$
Putting this all together, the integral in equation \eqref{eq:app-prefactorized-integral} passes to a Haar integral over the quotient group $\S/\N$:
\begin{equation} \label{eq:prefactorized-integral-2}
    (O_1 O_2)_\S
        = \int_{\S/\N} d[s]\, U(s) O_1 O_2 U(s)^{\dagger}.
\end{equation}
We will now show that this integral factorizes into a piece that acts only on $R_1$ and a piece that acts only on $R_2.$

To accomplish this, we will first show that $\S$ possesses a dense subset of elements of the form $s^{(1)} s^{(2)} x$ with $s^{(j)} \in \S_j$ and $x \in \N.$
To see this, we note that by assumption (iii) above, elements of the form
\begin{equation}
    s^{(1)}_{1} s^{(2)}_1 x_1 \dots s^{(1)}_n s^{(2)}_n x_n
\end{equation}
are dense in $\S$.
But every such element can be rewritten in the form $s^{(1)} s^{(2)} x$ by writing
\begin{equation}
    x_{n-1} s^{(1)}_n s^{(2)}_n x_n
        = (s^{(1)}_n s^{(2)}_n) \left[ (s^{(1)}_n s^{(2)}_n)^{-1} x_{n-1} (s^{(1)}_n s^{(2)}_n)\right] x_n,
\end{equation}
using normality of $\N$ to conclude that the term in brackets is in $\N$, and then proceeding by induction using the assumption that $\S_1$ and $\S_2$ commute.

For each element of the form $s^{(1)} s^{(2)} x,$ the choice of representatives $s^{(1)} s^{(2)}$ is not necessarily unique, since if $\S_j$ has nontrivial intersection with $\N,$ then elements in this intersection can be moved freely between $s^{(j)}$ and $x.$
We will see however that this is the only non-uniqueness, so there is a map from elements of the form $s^{(1)} s^{(2)} x$ to the Lie group
\begin{equation} \label{eq:product-group}
    (\S_1 / (\S_1 \cap \N)) \times (\S_2 / (\S_2 \cap \N))
\end{equation}
given by
\begin{equation}
    s^{(1)} s^{(2)} x \mapsto ([s^{(1)}],[s^{(2)}]).
\end{equation}
To show this, suppose we have
\begin{equation}
    s^{(1)} s^{(2)} x = \tilde{s}^{(1)} \tilde{s}^{(2)} \tilde{x}.
\end{equation}
Straightforward manipulations yield the expression
\begin{equation}
    1 = s^{(1)} s^{(2)} (\tilde{s}^{(1)} \tilde{s}^{(2)})^{-1} \left[ (\tilde{s}^{(1)} \tilde{s}^{(2)}) x \tilde{x}^{-1} (\tilde{s}^{(1)} \tilde{s}^{(2)})^{-1} \right]
\end{equation}
The term in square brackets is some term $\hat{x} \in \N$, and commutativity of $\S_1$ and $\S_2$ lets us rewrite the whole expression as
\begin{equation}
    \hat{x}^{-1} = s^{(1)} (\tilde{s}^{(1)})^{-1} s^{(2)} (\tilde{s}^{(2)})^{-1},
\end{equation}
The right-hand side must act trivially on both $R_1$ and $R_2$.
Since all elements in $\S_2$ act trivially on $R_1,$ we conclude that $(s^{(1)} \tilde{s}^{(1)})^{-1}$ must act trivially on $R_1$; hence it is an element of $\N$, and the group elements $s^{(1)}$ and $\tilde{s}^{(1)}$ are in the same equivalence class of the quotient $\S_1 / (\S_1 \cap \N).$
Similar considerations hold for $\S_2$.

We now have a well defined map from a dense subset of $\S$ onto the product group in equation \eqref{eq:product-group}.
By similar arguments to those given in the preceding paragraph, it is easy to show that this map is a homomorphism, and that its kernel is exactly $\N$.
So if we can show that our map is continuous, then it will extend to all of $\S$, and the first isomorphism theorem will guarantee the equivalence
\begin{equation} \label{eq:quotient-product}
    \S / \N
        \cong (\S_1 / (\S_1 \cap \N)) \times (\S_2 / (\S_2 \cap \N)).
\end{equation}
But because $\S$ is a Lie group, continuity of a homomorphism only needs to be checked at the identity; so if we have a sequence
\begin{equation} \label{eq:g-convergence}
    s_{j}^{(1)} s_{j}^{(2)} x_j \to 1,
\end{equation}
then we must show that the sequences $[s_j^{(1)}]$ and $[s_j^{(2)}]$ both converge to $[1].$
This part of the argument is a bit technical, so we give it in a separate subsection below.
The results of that subsection establish the isomorphism in equation \eqref{eq:quotient-product}.
This allows us to rewrite equation \eqref{eq:prefactorized-integral-2} as
\begin{align}
    \begin{split}
    (O_1 O_2)_\S
        & = \left( \int_{(\S_1 / (\S_1 \cap \N))} d[s_1]\, U(s_1) O_1 U(s_1)^{\dagger} \right) \left( \int_{(\S_2 / (\S_2 \cap \N))} d[s_2]\, U(s_2) O_2 U(s_2)^{\dagger} \right),
    \end{split}
\end{align}
and now each integral on the right-hand side can be lifted to the full group $\S$ to obtain
\begin{align}
    \begin{split}
    (O_1 O_2)_\S
        & = (O_1)_\S (O_2)_\S,
    \end{split}
\end{align}
completing the proof of disjoint additivity.

\subsection{A continuity lemma for Lie groups}
\label{app:continuity-lemma}

Here we prove that every identity-converging sequence
\begin{equation} \label{eq:app-group-convergence}
    s_j^{(1)} s_j^{(2)} x_j \to 1
\end{equation}
yields converging sequences in the quotient spaces, $[s_j^{(1)}] \to 1$ and $[s_j^{(2)}] \to 1.$
In the proof, we will make use of the factorization property
\begin{equation}
    U(s) = U_{R}(s) \otimes U_{R'}(s),
\end{equation}
which implies that each map $s \mapsto U_{R}(s)$ is a projective representation of the Lie group $\S$, and moreover that this map is projectively continuous.
Consequently, we can pass to the quotient and view $U_R$ as a continuous homomorphism of $\S$ into the projective unitary group $PU(\H).$
The image of $\S$ under this map is a compact Lie group, which will allow us to use the tools of Lie theory to study the maps $U_R$ and $U_{R'}.$
Note however that this group may not be isomorphic to $\S$, as $U_R$ can have a kernel.

To begin our proof, we start with the continuity assumption of section \ref{sec:Lie-constraint-definitions}, so that \eqref{eq:app-group-convergence} implies\footnote{When we write expressions like this, we are always using the strong topology: $O_n \to O$ means $O_n |\psi\rangle \to O |\psi\rangle$ for every vector $|\psi\rangle.$}
\begin{equation}
    U(s_j^{(1)} s_j^{(2)} x_j) \to 1.
\end{equation}
We then use the factorization assumption, giving
\begin{equation}
    U_{R_1}(s_j^{(1)} s_j^{(2)} x_j) \otimes U_{R_1'}(s_j^{(1)} s_j^{(2)} x_j) \to 1.
\end{equation}
Because $s_j^{(2)}$ and $x_j$ act trivially on $R_1,$ this gives
\begin{equation}
	 \chi_j U_{R_1}(s_j^{(1)}) \otimes U_{R_1'}(s_j^{(1)} s_j^{(2)} x_j) \to 1
\end{equation}
for some sequence of phases $\chi_j,$ which implies that $U_{R_1}(s_j^{(1)})$ converges projectively to the identity.

For the rest of the appendix we will drop the superscript, and simply discuss a sequence $s_j$ in $\S_1$ for which $U_{R_1}(s_j)$ converges projectively to the identity.
Our goal is to show $[s_j] \to 1,$ which means we must find a sequence $z_j \in \S_1 \cap \N$ with $s_j z_j \to 1$.
Since we assumed in section \ref{sec:Lie-constraint-definitions} that the representation was faithful, we may equivalently check the condition
\begin{equation}
    U(s_j z_j) \to 1.
\end{equation}

We begin by writing
\begin{equation} \label{eq:app-g-splitting}
    [U(s_{j})] = [U_{R_1}(s_j) \otimes U_{R_1'}(s_j)],
\end{equation}
where brackets around a unitary denote the projective quotient (not to be confused with brackets around a group element, which denote the quotient by $\S_1 \cap \N$).
So far we know
\begin{equation} \label{eq:app-convergence}
   [U_R(s_j)\otimes 1]= [U(s_j)] [1 \otimes U_{R_1'}(s_j^{-1})] \to 1.
\end{equation}
The idea is that while $[U(s_j)]$ and $[1 \otimes U_{R_1'}(s_j^{-1})]$ are both sequences in the compact Lie group generated by $[U_{R_1}(\S_1) \otimes 1]$ and $[1 \otimes U_{R_1'}(\S_1)],$ they live in different Lie subgroups.
One of them is contained in the subgroup $[U(\S_1)]$, while the other is contained in the subgroup $[1 \otimes U_{R_1'}(\S_1)]$.
But equation \eqref{eq:app-convergence} roughly says that these two sequences become very close to one another for large $j$ --- therefore we might hope to find a sequence in the intersection that approximates either sequence arbitrarily well; we will find our group elements $z_j$ by performing such an approximation in $[U(\S_1)]$ for the sequence $[1 \otimes U_{R_1'}(s_j^{-1})]$.

Every Lie group admits a right-invariant metrization --- see for example \cite[theorem 1.22]{montgomery2018topological}.
So there is a metric $d(\cdot, \cdot)$ on $[U_{R_1}(\S_1) \otimes 1] \times [1 \otimes U_{R_1'}(\S_1)]$ that induces the Lie group topology and that is invariant under right-multiplication of both entries.
From this and from equation \eqref{eq:app-convergence}, one deduces
\begin{equation} \label{eq:app-distance-convergence}
    d([U(s_j)], [1 \otimes U_{R_1'}(s_j)]) \to 0.
\end{equation}
Since $[U(\S_1)]$ is compact, any open set containing it has, within its interior, some $\epsilon$-sized thickening around $[U(\S_1)]$ --- see figure \ref{fig:app-thickening-a}.\footnote{If a compact set $K$ is contained in an open set $\Omega,$ then because the distance function is continuous, there is a minimum distance between $K$ and the complement of $\Omega$. Covering $K$ by balls with radii less than this minimum distance gives an $\epsilon$-thickening of $K$ that is contained in $\Omega.$}
So equation \eqref{eq:app-distance-convergence} tells us that $[1 \otimes U_{R_1'}(s_j)]$ eventually lies within any open set containing $[U(\S_1)].$
In fact, this tells us that it eventually enters any open set containing the intersection $[U(\S_1)] \cap [1 \otimes U_{R_1'}(\S_1)]$ --- see figure \ref{fig:app-thickening-b}.
For if $\Omega$ is an open set containing the intersection, then we may take the union with the complement of $[1 \otimes U_{R_1'}(\S_1)]$ to obtain an open set containing $[U(\S_1)]$; the sequence $[1 \otimes U_{R_1'}(s_j)]$ must eventually lie within this open set, but it cannot ever enter the complement of $[1 \otimes U_{R_1'}(\S_1)],$ so instead it must eventually lie within $\Omega.$

\begin{figure}[h]
    \centering
    \subcaptionbox{\label{fig:app-thickening-a}}[0.45\linewidth]
    {%
    \begin{tikzpicture}
        \draw [ultra thick, fill=lightgray, fill opacity=0.7] (0,0) to[out=80, in=210] ++(2, 1) to[out=210-180, in=180] ++(1,1) to[out=0, in=100] ++(1, -1.5) to[out=280, in=30] ++(-1, -2) to[out=210, in=-20] ++(-2, 0.75) to [out=160, in=260] (0,0);
        \draw [dashed] (2.2, 0) circle (2.5);
        \begin{scope}[scale=1.1, xshift=-0.2cm]
            \draw [thick, fill=lightgray, fill opacity=0.2] (0,0) to[out=80, in=220] ++(2, 1.1) to[out=220-180, in=180] ++(1,0.9) to[out=0, in=100] ++(1, -1.5) to[out=280, in=30] ++(-1, -2) to[out=210, in=-20] ++(-2, 0.75) to [out=160, in=260] (0,0);
        \end{scope}
        \node at (2.2, 0) {$K$};
        \node at (4.5, 2) {$\Omega$};
        \node at (1.8, 1.6) {$K_{\epsilon}$};
    \end{tikzpicture}
    }
    \subcaptionbox{\label{fig:app-thickening-b}}[0.45\linewidth]
    {
    \begin{tikzpicture}
        \fill [ultra thick, fill=RoyalBlue, fill opacity=0.7] (0,0) to[out=80, in=210] ++(2, 1) to[out=210-180, in=180] ++(1,1) to[out=0, in=100] ++(1, -1.5) to[out=280, in=30] ++(-1, -2) to[out=210, in=-20] ++(-2, 0.75) to [out=160, in=260] (0,0);
        \fill [ultra thick, fill=BrickRed, fill opacity=0.7] (-1,0) circle (1.75);
        \draw [ultra thick] (0,0) to[out=80, in=210] ++(2, 1) to[out=210-180, in=180] ++(1,1) to[out=0, in=100] ++(1, -1.5) to[out=280, in=30] ++(-1, -2) to[out=210, in=-20] ++(-2, 0.75) to [out=160, in=260] (0,0);
        \draw [ultra thick] (-1,0) circle (1.75);

        \begin{scope}
            \draw [dashed] (-0.15, 0) to[out=80, in=190] ++(0.95, 0.8) to[out=-80, in=70] ++(-0.1, -1.6) to[out=160, in=260] (-0.15, 0);
            \node at (0.75, 1.1) {$\Omega$};
        \end{scope}

        \node at (2.2, 0) {$[U(\S_1)]$};
        \node at (-1.45, 0) {$[1 \otimes U_{R_1'}(\S_1)]$};
    \end{tikzpicture}
    }

    \caption{(a) In a metric space, for every open set $\Omega$ containing a compact set $K$, there is an $\epsilon$-thickening $K_{\epsilon}$ that is contained entirely in $\Omega.$
    (b) Every open set $\Omega$ containing the intersection $[1 \otimes U_{R_1'}(\S_1)] \cap [U(\S_1)]$ can be supplemented by the complement of $[1 \otimes U_{R_1'}(\S_1)]$ to form an open set containing $[U(\S_1)].$
    A sequence in $[1 \otimes U_{R_1'}(\S_1)]$ that eventually enters this larger open set must necessarily enter $\Omega.$}
    \label{fig:app-thickening}
\end{figure}
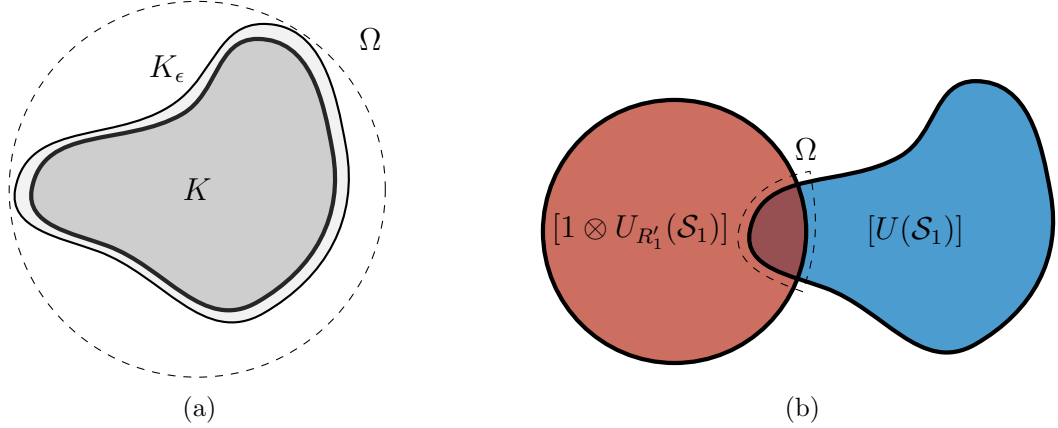

We now know that for any $\epsilon$-size thickening of $[U(\S_1)] \cap [U_{R_1'}(\S_1)],$ the sequence $[1 \otimes U_{R_1'}(s_j)]$ is eventually contained entirely within this set.
In other words, the distance
\begin{equation}
    \inf\{ d([1 \otimes U_{R_1'}(s_j)], p) \, |\, p \in [U(\S_1)] \cap [1 \otimes U_{R_1'}(\S_1)]\} 
\end{equation}
converges to zero, and by compactness of the intersection we may replace the infimum by a minimum.
Letting $[U(z_j^{-1})]$ be a point that attains the minimum for each $j,$ we find
\begin{align}
\begin{split}
    d([U(s_j z_j)],1)
        & = d([U(s_j)], [U(z_j^{-1})]) \\
        & \leq d([U(s_j)], [1 \otimes U_{R_1'}(s_j)])
            + d([1 \otimes U_{R_1'}(s_j)], [U(z_j^{-1})]) \\
        & \to 0.
    \end{split}
\end{align}
So we have demonstrated the existence of a sequence $z_j \in \S_1$ such that
\begin{enumerate}[(i)]
    \item $[U(z_j)] \in [U(\S_1)]\cap[1 \otimes U_{R_1'}(\S_1)]$, so $U(z_j)$ acts trivially on operators in regions  $R_1$ and $R_2$.  Therefore $z_j$ must be contained in $\N$.
    \item the sequence $U(s_j z_j)$ converges projectively to the identity.
\end{enumerate}
If we can show that $z_j$ may be chosen so that $U(s_j z_j)$ converges \textit{exactly} to the identity, without needing to introduce phase factors, then we are done.
To accomplish this, we can perform an identical analysis to the one we have already performed, but this time carry it out in the group $U(\S_1)$ with $U(1)$ as the quotient subgroup, instead of in the group $PU(\S_1)$ with $PU(\S_1 \cap \N)$ as the quotient subgroup.
Namely, since $U(s_j z_j)$ converges projectively to the identity, we know there is a sequence $U(s_j z_j) \in U(\S_1)$ and a sequence $\zeta_j \in U(1)$ such that we have
\begin{equation}
    \zeta_j U(s_j z_j) \to 1.
\end{equation}
By repeating exactly the steps we performed above, one finds that there exists a sequence $\lambda_j$ in the intersection $U(\S_1) \cap U(1)$ with
\begin{equation}
    \lambda_j U(s_j z_j) \to 1.
\end{equation}
But since $\lambda_j$ is in the intersection $U(\S_1) \cap U(1)$, it can be represented as $U(v_j)$ for some $v_j \in \S_1 \cap \N.$\footnote{$\lambda_j$ is a multiple of the identity, so it acts trivially on $R_1$, which is why $v_j$ must be in $\N$.}
We can therefore rewrite $z_j v_j = z_j'$ and obtain $U(s_j z_j') \to 1$ with $z_j' \in \S_1 \cap \N.$

\section{Majorana chain with an even number of fermions}
\label{sec:Even_Majo_chain_Haag_duality}

In this appendix we prove that the algebra of operators for even number of Majorana fermions in 1+1D satisfies Haag duality. 
The space is a 1D ring of $L=2\ell$ sites, with a single Majorana fermion $\chi_j$ operator per site. 
Here  $j=1,2,\cdots, L=2\ell$. 
The Majorana fermion operators obey the Clifford algebra $\{ \chi_j ,\chi_{j'}\} = 2\delta_{j j'}$. 
A subregion $\R$ is specified by a collection of sites, and we assign the algebra of operator $\cal {A}(R)$ to be those generated by $\chi_j$ with $j\in {\cal R}$.

These operators act on a $2^{\ell}$-dimensional Hilbert space, which furnishes an irreducible representation of the Clifford algebra. We choose the following representation for the Majorana fermions:
\ie
\label{eqn:Majo_Rep}
&\chi_1 = X \otimes \mathbf{1}\otimes\mathbf{1}\otimes \cdots\otimes \mathbf{1} ,\\
&\chi_2 = Y \otimes \mathbf{1} \otimes \mathbf{1}\otimes \cdots,\\
&\chi_3 = Z \otimes X \otimes \mathbf{1}\cdots,\\
&\chi_4 = Z \otimes Y \otimes \mathbf{1}\cdots,\\
\vdots
\fe
where there are $\ell$ tensor factors, each of which acts on a qubit composed from a pair of Majorana fermions. 
Hence each qubit site is made out of two Majorana sites. 
Note that we define spatial subregions $\cal R$ with respect to the Majorana sites. In particular,  we allow a subregion $\cal R$ to consist of odd number of Majorana sites. 
The fermion parity operator in this representation is 
\ie
\label{eqn:F-parity}
(-1)^F = Z\otimes Z\otimes Z\cdots 
\fe
which implements an inner automorphism $\chi_j\to -\chi_j$.

 For any subregion $\R$, it is clear that $\mathcal{A}(\R)\subseteq \mathcal{A}(\R^\prime)^\prime$. Hence, to prove Haag duality we need to prove $\mathcal{A}(\R^\prime)^\prime\subseteq  \mathcal{A}(\R)$. 
It turns out that it will be easier to prove the contrapositive,
\begin{align}
\label{eqn:contrapos_HD}
    x\notin \mathcal{A}(\R)\implies x\notin \mathcal{A}(\R^\prime)^\prime, ~\text{i.e.}~ \exists~ ~ b\in \mathcal{A}(\R^\prime) ~\text{such that}~ [x,b]_\pm\neq 0.
\end{align}
Here we assume $x$ is either a bosonic or a fermionic operator. 
We split the proof into two cases with even or odd number of sites in $\cal R$.

\subsection{Subregions with even number of fermions}\label{app:even}

For simplicity, we focus on the case where the subregion $R$ consists of lattice sites $j\in \{1,2,\cdots, 2r\}$. 
More general subregions can be treated by applying a unitary transformation that permutes the fermions.\footnote{Explicitly, this unitary is a product of $\pi_{ij}\equiv i(-1)^F {\chi_j-\chi_i\over \sqrt{2}}$ (with $i\neq j$), which acts on the fermions as $\pi_{ij} \chi_i \pi_{ij}^{-1} = \chi_j,\pi_{ij} \chi_j \pi_{ij}^{-1} = \chi_i$, and $\pi_{ij} \chi_k \pi_{ij}^{-1} = \chi_k$ if $k\neq i,j$. }
(Note that this also holds if $\R$ has disconnected components.)
The algebra $\mathcal{A}(\R)$ associated to subregion $\R$  is the algebra generated by $\{\chi_j\}_{j\in\{1,2,...,2r\}}$. 
With the representation \eqref{eqn:Majo_Rep}, it is clear that 
\begin{align}\label{app:AR}
{\cal A}({R}) = \{ X_{R} \otimes \mathbf{1}^{\otimes R'} ~|~X_{R}\in \text{Mat}(2^r,\mathbb{C})\},
\end{align}
where $O^{\otimes R'}$ is the tensor product of the local operator $O$ from every qubit site in region $R'$.
The sets of bosonic and fermionic operators in $\mathcal{A}(\R^\prime)$ are respectively
\begin{align}
&{\cal A}^+({R}') = \{ \mathbf{1}^{\otimes R}\otimes M_{R'}^+ ~|~M_{R'}^+ \in \text{Mat}(2^{\ell-r},\mathbb{C}), ~[(Z)^{\otimes R'} ,M_{R'}^+]=0\}\,,\\
&{\cal A}^-({R}') = \{ (Z)^{\otimes R}\otimes M_{R'}^-~|~M_{R'}^- \in \text{Mat}(2^{\ell-r},\mathbb{C}), ~\{(Z)^{\otimes R'} ,M_{R'}^-\}=0\}\,.
\end{align}

Since we assume $x\notin {\cal A}({R})$, using \eqref{app:AR}, there must exist a matrix $Y_{R'}\in \text{Mat}(2^{\ell-r},\mathbb{C})$ such that 
$[x, ~\mathbf{1}^{\otimes R} \otimes Y_{R'}~]\neq 0$.\footnote{Note that this is the ordinary commutator, not the supercommutator.  This is because for the ordinary commutator $A(R)$ is the commutant of the set of operators of the form $I\otimes Y_{R'}$.} 
We can decompose this matrix into its bosonic and fermionic part: $Y_{R'}=Y_{R'}^+ +Y_{R'}^-$, where $[(Z)^{\otimes R'},Y_{R'}^+ ]= 0 $ and $\{(Z)^{\otimes R'},Y_{R'}^- \}= 0$.\footnote{The Hilbert space decomposes into two  eigenspaces of $(-1)^F$, and $Y_{\R'}$ can be block-decomposed in this basis. Diagonal blocks are bosonic, and off-diagonal blocks are fermionic.}
It follows that at least one of the following holds:
\begin{align}
[x, ~\mathbf{1}^{\otimes R}\otimes Y_{R'}^+~]\neq 0~~\text{or}~~[x, ~\mathbf{1}^{\otimes R} \otimes Y_{R'}^-~]\neq 0 \,.
\end{align}

\paragraph{Case 1 $[x, ~\mathbf{1}^{\otimes R}\otimes Y_{R'}^+~]\neq 0$}: Then  the operator $\mathbf{1}^{\otimes R} \otimes Y_{R'}^+$, which is a bosonic operator in ${\cal A}^+({R}')$, does not commute with $x$, and we have  $x\notin {\cal A}({R}')'$. 

\paragraph{Case 2 $[x, ~\mathbf{1}^{\otimes R}\otimes Y_{R'}^-~]\neq 0$}: Then  the operator $(-1)^F(\mathbf{1}^{\otimes R}\otimes Y_{R'}^-)=(Z)^{\otimes R} \otimes (Z)^{\otimes R'} Y_{ R'}^-$, which is a fermionic operator in ${\cal A}^-({ R}')$, does not supercommute with $x$. That is, $x \, (-1)^F(\mathbf{1}^{\otimes R}\otimes Y_{ R'}^-) \neq (-1)^{|x|} \,(-1)^F(\mathbf{1}^{\otimes R}\otimes Y_{ R'}^-)\,x$. Hence $x\notin{\cal A}({ R}')'$.

In either case, we have shown that $x\notin {\cal A}({ R})$ implies $x\notin{\cal A}({ R}')'$, and hence the proof when $ R$ has even number of sites.

\subsection{Subregions with odd number of fermions}

Again, without loss of generality, we let $ R$ be the subregion consisting of sites $\{1,2,\cdots, 2r-1\}$. 
Let $\tilde {R} ={ R}\cup \{2r\}$ be the subregion with one more site at $2r$ included. 
Suppose $x\notin {\cal A}({ R})$, we split the proof into two cases. 

In the first case $x\notin {\cal A}(\tilde R)$, where $\tilde R$ has even number of sites. Then the proof in the Appendix \ref{app:even} implies that $x\notin {\cal A}(\tilde R')'$, which implies $x\notin {\cal A}({ R}')'$ as $ {\cal A}({ R}')'\subseteq{\cal A}(\tilde R')'$.

In the second case $x$ is in ${\cal A}(\tilde R)$ but not in ${\cal A}({ R})$. Then $x$ can be written as
\begin{align}
x= W_0 +W_1 \,\chi_{2r}\,,~~~W_0,W_1\in {\cal A}({ R})\,,
\end{align}
and $W_1\neq0$. By noting that $W_0$ and $W_1$ have opposite parity under $(-1)^F$, it is immediately clear that $x$ does not supercommute with $\chi_{2r}\in {\cal A}({ R}')$. Hence $x\notin {\cal A}({ R}')'$.

Putting everything together, we have proven that the local operator algebras of an even number of Majorana operators satisfy both  additivity and Haag duality.  
 
\bibliography{ref}

\bibliographystyle{JHEP}
\end{document}